\documentclass{tlp} % https://svn.haiti.cs.uni-potsdam.de/svn/reposWV/TeXInputs/trunk

\usepackage{times,helvet,courier}

\usepackage[utf8]{inputenc}
\usepackage[english]{babel}
\usepackage{url}
\usepackage{ifpdf}

\usepackage{amsmath,amsfonts,amssymb}
\usepackage{tikz}
\usetikzlibrary{arrows,automata}
\usetikzlibrary{shapes}% \usepackage{pgflibraryshapes}

\newif\ifrecentpgf
\def\requiredversion{2.00}
\ifx\pgfversion\requiredversion\recentpgftrue\else\recentpgffalse\fi

\newcommand{\naf}[1]{\ensuremath{\mathit{not}~{#1}}} % {\ensuremath{{\sim}{#1}}}

\newcommand{\body}[1]{\ensuremath{\mathit{body}(#1)}} % positive body atome
\newcommand{\head}[1]{\ensuremath{\mathit{head}(#1)}}
\newcommand{\poslits}[1]{\ensuremath{#1^+}}
\newcommand{\neglits}[1]{\ensuremath{#1^-}}

\newcommand\BS[1]{\text{\fontshape{n}\fontseries{bx}\selectfont#1}}
\newcommand\plus{\BS{+}}
\newcommand\minus{\BS{--}}
\newcommand\indet{\BS{?}}
\newcommand{\dom}[1]{\ensuremath{D(#1)}}
\newcommand{\vertices}[1]{\ensuremath{V[#1]}}
\newcommand{\edges}[1]{\ensuremath{E[#1]}}

\newcommand{\atomfont}[1]{\text{\textit{#1}}}
\newcommand{\vertex}[1]{\ensuremath{\atomfont{vertex}(#1)}}
\newcommand{\avertex}[1]{\ensuremath{\atomfont{vertexMIC}(#1)}}
\newcommand{\iinput}[1]{\ensuremath{\atomfont{input}(#1)}}
\newcommand{\edge}[2]{\ensuremath{\atomfont{edge}(#1,#2)}}
\newcommand{\aedge}[2]{\ensuremath{\atomfont{edgeMIC}(#1,#2)}}

\newcommand{\vlabel}[2]{\ensuremath{\atomfont{labelV}(#1,#2)}}
\newcommand{\vlabell}[3]{\ensuremath{\atomfont{labelV'}(#1,#2,#3)}}
\newcommand{\elabel}[3]{\ensuremath{\atomfont{labelE}(#1,#2,#3)}}
\newcommand{\elabell}[4]{\ensuremath{\atomfont{labelE'}(#1,#2,#3,#4)}}
\newcommand{\observed}[1]{\ensuremath{\atomfont{obs}(#1)}}
\newcommand{\obsvlabel}[2]{\ensuremath{\atomfont{observedV}(#1,#2)}}
\newcommand{\obselabel}[3]{\ensuremath{\atomfont{observedE}(#1,#2,#3)}}
\newcommand{\get}[2]{\ensuremath{\atomfont{get}(#1,#2)}}
\newcommand{\influence}[2]{\ensuremath{\atomfont{receive}(#1,#2)}}
\newcommand{\influencel}[3]{\ensuremath{\atomfont{receive'}(#1,#2,#3)}}
\newcommand{\activ}[1]{\ensuremath{\atomfont{active}(#1)}}
\newcommand{\inactiv}[1]{\ensuremath{\atomfont{inactive}(#1)}}
\newcommand{\contrary}[2]{\ensuremath{\atomfont{opposite}(#1,#2)}}
\newcommand{\bottom}[0]{\ensuremath{\atomfont{bottom}}}
\newcommand{\sedge}[2]{\ensuremath{\atomfont{edges}(#1,#2)}}
\newcommand{\reach}[2]{\ensuremath{\atomfont{reach}(#1,#2)}}
\newcommand{\cycle}[2]{\ensuremath{\atomfont{cycle}(#1,#2)}}

\newcommand{\sysfont}[1]{\textit{#1}}
\newcommand{\lparse}[0]{\sysfont{lparse}}
\newcommand{\cmodels}[0]{\sysfont{cmodels}}
\newcommand{\clasp}[0]{\sysfont{clasp}}
\newcommand{\claspD}[0]{\sysfont{claspD}}
\newcommand{\dlv}[0]{\sysfont{dlv}}
\newcommand{\gnt}[0]{\sysfont{gnt}}
\newcommand{\gringo}[0]{\sysfont{gringo}}

\newtheorem{definition}{Definition}[section]
\newtheorem{theorem}{Theorem}[section]

\newtheorem{corollary}[theorem]{Corollary}

\title[Detecting Inconsistencies in Large Biological Networks with ASP]{Detecting Inconsistencies in Large Biological Networks with Answer Set Programming}

\author[Martin~Gebser \and Torsten~Schaub \and Sven~Thiele \and Philippe~Veber]{%
  Martin~Gebser
  and
  Torsten~Schaub
  and
  Sven~Thiele
  \\
  University of Potsdam,
  Germany 
  \and
  Philippe~Veber
  \\
  Institut Cochin,
  Paris,
  France
  }

\submitted{21st June 2009}
\revised{15th October 2009}
\accepted{9th December 2009}

\begin{document}
\renewcommand\capsulename{Note}
\maketitle
\begin{abstract}
We introduce an approach to detecting inconsistencies in large biological
networks by using Answer Set Programming (ASP).
To this end,
we build upon a recently proposed notion of consistency between
biochemical/genetic reactions and high-throughput profiles of cell activity.
We then present an approach based on ASP to check
the consistency of large-scale data sets.
Moreover, we extend this methodology to provide explanations for inconsistencies
% in the data
by determining minimal representations of conflicts.
In practice, this can be used % in two ways, either 
to identify unreliable data or
to indicate missing reactions.
\end{abstract}
\begin{keywords}
  answer set programming, bio-informatics, consistency, diagnosis
\end{keywords}
%%% Local Variables: 
%%% mode: latex
%%% TeX-master: "paper"
%%% End: 

\begin{capsule}
  To appear in Theory and Practice of Logic Programming (TPLP).
\end{capsule}

\section{Introduction}\label{sec:introduction}

Molecular biology has seen a technological revolution with the establishment of
high-throughput methods in the last years.
These methods allow for gathering multiple orders of magnitude more measured
data than was procurable before.
Furthermore, there is an increasing number of biological repositories on
the web, such as 
KEGG, % \cite{kangot99},
Biomodels, % \cite{lebobrcododhlisascshsnhu06},
Reactome, % \cite{jogivadescbojagowumalebist08},
% or
MetaCyc, % \cite{cafofukakrlaparhshtiwazhka07},
and others,
incorporating thousands of biochemical reactions and genetic regulations.
However, both measurements as well as biological networks are prone to
considerable incompleteness, heterogeneity, and mutual inconsistency, which
makes it highly non-trivial to draw biologically meaningful conclusions in an
automated way.
As a consequence, appropriate representation and powerful reasoning tools are needed
to model complex biological systems, in the face of incompleteness
and inconsistency.

In this paper,
we deal with the analysis of high-throughput measurements in molecular biology,
like microarray data or metabolic profiles. % \cite{Joyce2006}.
Up to now, it is still common practice to use expression profiles
merely for detecting over- or under-expressed genes under specific conditions,
leaving the task of making biological sense out of % tens of 
a multitude of gene identifiers to
human experts.
However,
many efforts have also been made to better utilize
high-throughput data, in particular, by integrating them into
large-scale models of transcriptional regulations or metabolic
processes \cite{Friedman2000,revue-fba}.  

One possible approach consists of investigating the compatibility
between experimental measurements and knowledge
available in reaction databases.
This can be done by using formal frameworks, for instance, the ones developed
in \cite{pmid14597655} and \cite{pmid16556482}.
A crucial feature of this methodology is its ability to cope with qualitative knowledge
(for instance, reactions lacking kinetic details) and noisy data.
In what follows, we rely upon the so-called \emph{Sign Consistency Model} (SCM)
due to \cite{pmid16556482}. %Siegel~et~al. \cite{pmid16556482}.
SCM imposes constraints between experimental measurements and a
graph representation of cellular interactions, called an
\emph{influence graph} \cite{SouleComplexus}.
Such a graph provides an over-approximation of the actual biological model, 
where an ``influence'' is modeled by a disjunctive causal rule.
This is particularly well-suited for dealing with incomplete (missing reactions)
or unreliable (noisy data) information.

Building on the SCM framework,
we develop declarative techniques based on \emph{Answer Set Programming}
(ASP) \cite{baral02a,gelfond08a} to detect and explain inconsistencies in large data sets.
This approach has several advantages.
First, it allows us to formulate biological problems in a declarative way,
thus easing the communication with biological experts.
Second, although we do not detail it here,
the rich modeling language facilitates integrating different
knowledge representation and reasoning techniques, like abduction, explanation,
planning, prediction, etc., in a uniform and transparent way
(cf.~\cite{geguivscsithve09a} for such extensions).
And finally, modern ASP solvers are based on advanced Boolean constraint solving
technology and thus provide us with highly efficient inference engines.
Apart from modeling the aforementioned biological problems in ASP,
our major concern lies with the scalability of the approach.
To this end, we % do not only illustrate \REWc{our application domain}{T: More?!} on an example but,
apply our methods to the gene-regulatory network of yeast \cite{guelzim,snf2-ko} and,
moreover, design an artificial yet biologically meaningful benchmark suite 
% that allows us to 
indicating that an ASP-based approach scales well on the considered class
of applications.
Notably, to the best of our knowledge, 
the functionalities we provide go beyond the ones of the only comparable
approach \cite{gubomosi09a}.

To begin with,
we introduce SCM in Section~\ref{sec:influence}.
Section~\ref{sec:asp} gives the syntax and semantics of ASP used in
our application.
In Section~\ref{sec:checking}, we develop an ASP formulation of
checking the consistency between experimental profiles and influence graphs.
We further extend this approach in Section~\ref{sec:diagnosis}
to identifying minimal representations of conflicts
if the experimental data is inconsistent with an influence graph.
In Section~\ref{sec:refinements},
we describe simple yet effective techniques for input reduction
along with a connectivity property that is used to refine the
encoding presented in Section~\ref{sec:diagnosis}.
Section~\ref{sec:benchmark} is dedicated to an empirical evaluation of our
approach along with an exemplary case study on yeast.
For making our methods easily accessible,
an available web service is presented in Section~\ref{sec:web}.
% illustrating our application domain.
% We report preliminary experimental results on randomly generated yet
% biologically meaningful instances.
Section~\ref{sec:discussion} concludes the paper with a discussion and % an
outlook on future work.
Finally,
% \ref{sec:sb:primer} gives a brief introduction to systems biology.
\ref{app:proof:consistency} and \ref{app:proof:diagnosis} contain proofs of
soundness and completeness for our problem formulations in ASP.

%%% Local Variables: 
%%% mode: latex
%%% TeX-PDF-mode: t 
%%% TeX-master: "paper"
%%% End: 

\section{Influence Graphs and Sign Consistency Constraints}\label{sec:influence}

Influence graphs \cite{SouleComplexus} are a common representation for a wide range of
dynamical systems.  
In the field of genetic networks, they have been
investigated for various classes of systems, %dynamical systems, 
ranging from ordinary differential equations \cite{SouleRevue} to synchronous
\cite{ruet} and asynchronous \cite{richard07} Boolean networks. 
Influence graphs have also been introduced in the field of qualitative reasoning
\cite{kuipers94a} to describe physical systems where a detailed quantitative
description is unavailable. 
In fact, this has been the main motivation for using influence graphs for knowledge
representation in the context of biological systems.

An \emph{influence graph} is a directed graph whose vertices are the
input and state variables of a system and whose edges express the effects of
variables on each other.
\begin{definition}[Influence Graph]
  An \emph{influence graph} is a directed graph $(V,E,\sigma)$,
  where $V$ is a set of vertices,
  $E$ a set of edges, and 
  $\sigma : E \rightarrow \{\plus,\minus\}$ a (partial) labeling of the edges.
\end{definition}
An edge $j {\,\rightarrow\,} i$ means that the variation of~$j$ in
time influences the level of~$i$.
Every edge $j {\,\rightarrow\,} i$ of an influence graph can be labeled with a sign, either~$\plus$
or~$\minus$, denoted by~$\sigma(j,i)$,
where~$\plus$ ($\minus$)
indicates that $j$ tends to increase (decrease)~$i$.
An example influence graph is given in
Figure~\ref{fig:operon}; it represents a simplified model of the operon
lactose in \emph{E.~coli}. %\comment{T2P+S: Can you address the issue of referree~1?}

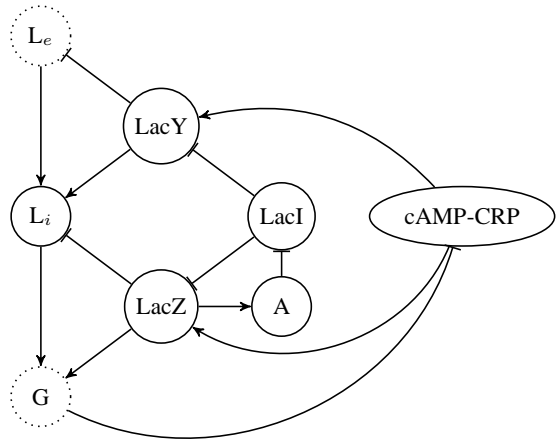
\begin{figure}
  \centering
%\ifrecentpgf
%% color picture
% \begin{tikzpicture}[->,semithick,>=stealth',scale=0.8]
%   \tikzstyle{every state}=[draw, circle, fill=none,text=black]%, minimum size=1cm]
%   \node[state,dotted] (Le) at (0,6)  {L$_e$};
%   \node[state] (Li) at (0,3) {L$_i$};
%   \node[state,dotted] (G) at (0,0) {G};
%   \node[state] (LacY) at (2,4.5) {LacY};
%   \node[state] (LacZ) at (2,1.5) {LacZ};
%   \node[state] (LacI) at (4,3) {LacI};
%   \node[state] (A) at (LacZ -| LacI) {A};
%   \node[state,ellipse] (cAMPCRP) at (7,3) {cAMP-CRP};
% 
%   \path
%      (Le) edge[green] (Li)
%      (Li) edge[green] (G)
%      (LacY) edge[red,-|] (Le)
%      (LacY) edge[green] (Li)
%      (LacZ) edge[red,-|] (Li)
%      (LacZ) edge[green] (G)
%      (LacI) edge[red,-|] (LacY)
%      (LacI) edge[red,-|] (LacZ)
%      (LacZ) edge[green] (A)
%      (A) edge[red,-|] (LacI)
%      (cAMPCRP) edge[bend right=30,green] (LacY)
%      (cAMPCRP) edge[bend left=50,green] (LacZ)
%      (G) edge[red,bend right=50,-|] (cAMPCRP)
%      ;
% \end{tikzpicture}

%% blck white picture
\begin{tikzpicture}[->,semithick,>=stealth',scale=0.8]
  \tikzstyle{every state}=[draw, circle, fill=none,text=black]%, minimum size=1cm]
  \node[state,dotted] (Le) at (0,6)  {L$_e$};
  \node[state] (Li) at (0,3) {L$_i$};
  \node[state,dotted] (G) at (0,0) {G};
  \node[state] (LacY) at (2,4.5) {LacY};
  \node[state] (LacZ) at (2,1.5) {LacZ};
  \node[state] (LacI) at (4,3) {LacI};
  \node[state] (A) at (LacZ -| LacI) {A};
  \node[state,ellipse] (cAMPCRP) at (7,3) {cAMP-CRP};

  \path
     (Le) edge (Li)
     (Li) edge (G)
     (LacY) edge[-|] (Le)
     (LacY) edge (Li)
     (LacZ) edge[-|] (Li)
     (LacZ) edge (G)
     (LacI) edge[-|] (LacY)
     (LacI) edge[-|] (LacZ)
     (LacZ) edge (A)
     (A) edge[-|] (LacI)
     (cAMPCRP) edge[bend right=30] (LacY)
     (cAMPCRP) edge[bend left=50] (LacZ)
     (G) edge[bend right=50,-|] (cAMPCRP)
     ;
\end{tikzpicture}
%\fi
\caption{Simplified model of operon lactose in \emph{E.~coli},
  represented as an influence graph. 
  The vertices represent either genes,
  metabolites, or proteins, while the edges indicate the regulations
  among them. 
  Edges with an arrow stand for positive
  regulations (activations), while edges with a tee head stand for
  negative regulations (inhibitions). 
  Vertices G and L$_e$ are
  considered to be inputs of the system, that is, 
  their signs are not constrained via their incoming edges.}
  \label{fig:operon}
\end{figure}

In SCM, \emph{experimental profiles} are supposed to come from steady
state shift experiments where, initially, the system is at steady state, 
then perturbed using control parameters, and eventually,
it settles into another steady state.
It is assumed that the data measures the differences between 
the initial and the final state.
Thus, for genes, proteins, or metabolites,
we know whether the concentration has increased or decreased,
while quantitative values are unavailable, unessential, or unreliable.
By~$\mu(i)$, we denote the sign, again either $\plus$ or~$\minus$,
of the variation of a species~$i$ between the initial and the final condition.
One can easily enhance this setting to also considering
null (or more precisely, non-significant) variations,
by exploiting the concept of sign algebra \cite{kuipers94a}.

%\subsection{Sign consistency constraints}

Given an influence graph
(as a representation of cellular interactions) and
a labeling of its vertices with signs
(as a representation of experimental profiles),
we now describe the constraints that relate both.
Informally, for every non-input vertex~$i$,  % this constraint says that
its variation~$\mu(i)$ ought to be explained by the
influence of at least one predecessor~$j$ of~$i$ in the influence graph.
Thereby, the \emph{influence} of~$j$ on~$i$ is given by the
sign $\mu(j)\sigma(j,i)\in\{\plus,\minus\}$,
where the multiplication of signs is derived from that of numbers.
Sign consistency constraints can then be formalized as follows.
\begin{definition}[Sign Consistency Constraints]\label{def:consistency}
  Let $(V,E,\sigma)$ be an influence graph and
  $\mu : V \rightarrow \{\plus,\minus\}$ a (partial) vertex labeling.
%   Let $(V,E,\sigma)$ be an influence graph,
%   where $V$ is the set of vertices,
%   $E$ the set of edges, and 
%   $\sigma : E \rightarrow \{\plus,\minus\}$ a labeling of the edges.
%   Furthermore, let
%   $\mu : V \rightarrow \{\plus,\minus\}$ be a vertex labeling.

  Then, $(V,E,\sigma)$ and $\mu$ are \emph{consistent},
  if there are some total extensions
  $\sigma' : E \rightarrow \{\plus,\minus\}$ of~$\sigma$
  and $\mu' : V \rightarrow \{\plus,\minus\}$ of~$\mu$
  such that $\mu'(i)$ is consistent for each non-input vertex $i\in V$,
  where $\mu'(i)$ is consistent,
  if there is some edge $j {\,\rightarrow\,} i$ in~$E$ such that
  $\mu'(i)=\mu'(j)\sigma'(j,i)$.
%   Then, for some non-input vertex $i\in V$,
%   the sign~$\mu(i)$ of~$i$ is consistent,
%   if there is some edge $j {\,\rightarrow\,} i$ in~$E$ such that
%   $\mu(i)=\mu(j)\sigma(j,i)$. 
\end{definition}
%
% The notion of (sign) consistency is extended to whole influence graphs in the
% natural way, requiring the sign of each non-input vertex to be consistent.
% In practice, influence graphs and experimental profiles are 
% likely to be partial.
% Thus, we say that a partial labeling of the vertices is
% consistent with a partially labeled influence graph,
% if there is some consistent extension of vertex and edge labelings
% to all vertices and edges. 
%
Note that labelings~$\sigma$ and~$\mu$ of vertices and edges, respectively,
are admitted to be partial.
This occurs frequently in practice where the kind of an influence
may depend on environmental factors or experimental data may not
include all elements of a biological system.
In order to decide whether a partially labeled influence graph and
a partial experimental profile are mutually consistent,
we thus consider the possible totalizations of them.
If at least one total edge and one total vertex labeling
(extending the given labelings)
are such that the signs of all non-input vertices are explained, it is
sufficient for mutual consistency.

\begin{table}
  \begin{center}
  \begin{tabular}{|c|cccccccc|}
    \cline{1-9}
    Species & L$_e$ & L$_i$ & G & LacY & LacZ & LacI & A
    & cAMP-CRP\\
    \cline{1-9}
    $\mu_1$ & \minus & \minus & \minus & \minus & \minus & \plus & \minus&\plus \\ 
    $\mu_2$ & \plus &\plus & \minus&\plus &\minus &\plus &\minus &\minus \\ 
    $\mu_3$ & \plus & \indet & \minus &\indet &\indet &\plus & \indet&\indet \\ 
    $\mu_4$ &\indet&\indet &\indet &\minus &\plus &\indet &\indet &\plus \\ 
    \cline{1-9}
  \end{tabular}
  \end{center}%
  \caption{Some vertex labelings (reflecting measurements of two steady states) for the
    influence graph depicted in Figure~\ref{fig:operon};
    unobserved values indicated by question mark `\indet'.\label{tab:operondata}}  
\end{table}
Table~\ref{tab:operondata} gives four vertex labelings
for the influence graph in Figure~\ref{fig:operon}.
Total labeling $\mu_1$ is consistent with the influence graph: the variation of
each vertex (except for input vertex~L$_e$)
can be explained by the effect of one of its regulators.
For instance, in $\mu_1$, LacY receives a positive influence from
cAMP-CRP as well as a negative influence from LacI, 
the latter accounting for the decrease of LacY.
The second labeling, $\mu_2$, is not consistent:
LacY receives only negative influences from cAMP-CRP and
LacI, and its increase cannot be explained.
Partial vertex labeling $\mu_3$ is consistent with the influence graph
in Figure~\ref{fig:operon}, as setting the signs of L$_i$, LacY, LacZ, A, and cAMP-CRP
to \plus, \minus, \minus, \minus, and \plus, respectively,
extends~$\mu_3$ to a consistent total labeling. 
In contrast, $\mu_4$ cannot be extended consistently.

% It is not hard to check that an input vertex $i$ (like G and L$_e$ in the above example)
% can be turned into a non-input vertex without altering consistency by 
% introducing an artificial vertex~$i'$ along with equally labeled edges $i \rightarrow i'$ and $i' \rightarrow i$.
% In this way, given a sign for~$i$, 
% the sign of $i'$ can always be picked such that both $i$ and $i'$
% receive an appropriate influence.
% For simplicity and without loss of generality,
% we in the remainder of this paper focus on influence graphs not containing input vertices.

%%% Local Variables: 
%%% mode: latex
%%% TeX-PDF-mode: t 
%%% TeX-master: "paper"
%%% End: 

\section{Answer Set Programming}\label{sec:asp}

This section provides a brief introduction to ASP,
% (see \cite{baral02a} for details),
a declarative problem solving paradigm
offering a rich modeling language \cite{lparseManual,potasscoManual}
along with highly efficient inference engines
based on Boolean constraint solving technology \cite{gilima06a,gekanesc07b,drgegrkakoossc08a}. 
The basic idea of ASP is to encode a problem as a logic program such that its
answer sets represent solutions. % to the original problem.

In view of our application,
we take advantage of the elevated expressiveness of disjunctive programs,
capturing problems at the second level of the polynomial
hierarchy \cite{eitgot95a}.
A \emph{disjunctive logic program} $P$ % over an alphabet $\mathcal{A}$ 
is a finite set of \emph{rules} of the form
\begin{equation}\label{eq:r}
a_1;\dots;a_l \leftarrow a_{l+1},\dots,a_m,\naf{a_{m+1}},\ldots,\naf{a_n}\ \text{,}
\end{equation}
where $a_i$ is an \emph{atom} for $1\!\leq\! i\!\leq\! n$.
% A \emph{literal} is an atom~$a$ or its (default) negation~$\naf{a}$.
A rule $r$ as in~(\ref{eq:r}) is called
a \emph{fact} if $l\!=\!m\!=\!n\!=\!1$, %$l\!=\!1$ and $n\!=\!0$,
and an \emph{integrity constraint} if $l\!=\!0$.
Let
$\head{r}=\{a_1,\dots, a_l\} $ be the \emph{head} of $r$,
\(
\body{r}=\{a_{l+1},\dots,a_m,\naf{a_{m+1}},\dots,\linebreak[1]\naf{a_n}\}
\)
be the \emph{body} of~$r$,
as well let
\(
\poslits{\body{r}}=\{a_{l+1},\dots,a_m\}
\)
and
\(
\neglits{\body{r}}=\{a_{m+1},\dots,a_n\}
\).
%
% Given a set $L$ of literals,
% let
% \(
% \poslits{L}=\{a\in\mathcal{A}\mid a\in L\}
% \)
% and
% \(
% \neglits{L}=\{a\in\mathcal{A}\mid \naf{a}\in L\}
% \).

An interpretation is represented by the set of atoms that are true in it. % the interpretation.
A \emph{model} of a program $P$ is an interpretation in
which all rules of $P$ are true according to the standard definition of
truth in propositional logic.
Apart from letting `$;$' and '$,$' stand for disjunction and conjunction,
respectively,
this implies treating rules and default negation `$\mathit{not}$' as
implications and classical negation, respectively.
Note that the (empty) head of an integrity constraint is false in every interpretation,
while the empty body is true in every interpretation.
Answer sets of~$P$ are particular models of~$P$
satisfying an additional stability criterion.
Roughly, a set~$X$ of atoms is an answer set,
if for every rule of form~(\ref{eq:r}),
$X$ contains a minimum of atoms among $a_1,\dots,a_l$ 
whenever $a_{l+1},\dots,a_m$ belong  to~$X$ 
and no   $a_{m+1},\dots,a_n$ belongs to~$X$.
However, the disjunction in heads of rules, in general, is not exclusive.
Formally, an \emph{answer set}~$X$ of a program~$P$ is a $\subseteq$-minimal model of
\[
\{
\head{r}\leftarrow\poslits{\body{r}}
\mid 
r\in P,\neglits{\body{r}}\cap X=\emptyset
\}\ \text{.}
\]
For example, program
\(
\{a;b{\,\leftarrow}.\quad\!\!\!\! c;d{\,\leftarrow\,}a,\naf{b}.\quad\!\!\!\! {\leftarrow\,} b.\}
\)
has answer sets~$\{a,c\}$ and~$\{a,d\}$.

Although answer sets are usually defined on ground (i.e., variable-free) programs,
ASP allows for non-ground problem encodings,
where schematic rules stand for their ground instantiations.
Grounders, such as \gringo\ \cite{potasscoManual} and \lparse\ \cite{lparseManual},
are capable of combining a problem enco\-ding and an instance (typically a set of ground facts)
into an equivalent ground program, which is then processed by an ASP solver.
We follow this methodology and provide encodings for the problems considered below.

%%% Local Variables: 
%%% mode: latex
%%% TeX-PDF-mode: t 
%%% TeX-master: "paper"
%%% End: 

\section{Checking Consistency}\label{sec:checking}

We now come to the first main question addressed in this paper, namely,
how to check whether an experimental profile is consistent with
a given influence graph.
Note that, if the profile provides us with a sign for each
vertex of the influence graph,
the task can be accomplished % in polynomial time.
simply by checking whether each non-input vertex receives
at least one influence matching its variation.
However, as soon as the experimental profile has missing values
(which is very likely in practice), % it is shown in that
the problem becomes NP-hard \cite{ECCS05}.
In fact, a Boolean satisfiability problem over clauses $  C_1,\dots,C_m  $ and variables
$  x_1,\dots,x_n  $ can be reduced as follows:
introduce unlabeled input vertices $x_1,\dots,x_n$, 
non-input vertices $C_1,\dots,C_m$ labeled~$\plus$,
and edges $x_j {\,\rightarrow\,} C_i$ labeled~$\plus$ ($\minus$) if
$x_j$ occurs positively (negatively) in $C_i$.
It is not hard to check that
the labeling of $C_1,\dots,C_m$ by~$\plus$ is consistent 
with the obtained influence graph iff the conjunction of
$  C_1,\dots,C_m  $ is satisfiable.

We next provide a logic program such that each of its answer sets matches a
consistent extension of vertex and edge labelings.
Our encodings as well as instances are available at \cite{bioasptoolchain}.
% For clarity, we here present them in a simplified manner and omit
% some convenient but unessential encoding optimizations.
The program for consistency checking is composed of three parts, described in the following subsections.

\subsection{Problem Instance}\label{subsec:instance} 
An influence graph as well as an experimental profile are given by ground facts.
For each species~$i$,
we introduce a fact $\vertex{i}$, and for each edge $j {\,\rightarrow\,} i$,
a fact $\edge{j}{i}$. 
If $s\in\{\plus,\minus\}$ is known to be the variation
of a species~$i$ or the sign of an edge $j {\,\rightarrow\,} i$,
it is expressed by a fact $\obsvlabel{i}{s}$ or
$\obselabel{j}{i}{s}$, respectively.
Finally, a vertex~$i$ is declared to be input via a fact $\iinput{i}$.

For example, the negative regulation $\text{LacI} {\,\rightarrow\,} \text{LacY}$
in the influence graph shown in Figure~\ref{fig:operon} and observation~$\plus$
for LacI (as with $\mu_3$ in Table~\ref{tab:operondata})
give rise to the following facts:
\begin{equation}\label{eq:ex:lac:facts}
  \begin{array}{l}
    \vertex{\text{LacI}}. \\
    \vertex{\text{LacY}}. \\[1mm]
    \edge{\text{LacI}}{\text{LacY}}. \\[1mm]
    \obsvlabel{\text{LacI}}{\plus}. \\
    \obselabel{\text{LacI}}{\text{LacY}}{\minus}.
  \end{array}  
\end{equation}
Note that the absence of a fact of form $\obsvlabel{\text{LacY}}{s}$
means that the variation of LacY is unobserved (as with $\mu_3$).
In~(\ref{eq:ex:lac:facts}), we use $\text{LacI}$ and $\text{LacY}$ 
as names for constants associated with the species in Figure~\ref{fig:operon},
but not as first-order variables.
Similarly,
for uniformity of notations, $\plus$ and~$\minus$ are written in~(\ref{eq:ex:lac:facts})
for constants identifying signs.
%
% For the
% lactose operon and the profile $\mu_3$ given in Table~\ref{}, the
% encoding is:
% \begin{equation}\label{eq:factencoding}
% \begin{array}{r@{{}\leftarrow{}}l}
%  sign(p). \\
%  sign(n). \\
%  vertex(crp). \\
%  vertex(lacZ). \\
%  \dots \\
%  edge(crp,lacZ). \\
%  edge(crp,LacY).\\
%  \dots \\
% % obs\_elabel(crp,lacZ,p) & obs\_elabel(crp,lacY,n) & \dots \\
% % obs\_vlabel(lacI,p). & obs\_vlabel(glucose,n).
% \end{array}
% \end{equation}

\subsection{Generating Solution Candidates} 
As mentioned above,
our goal is to check whether an experimental profile is consistent with an
influence graph.
If so, it is witnessed by total labelings of the vertices and edges,
which are generated via the following rules:
\begin{equation}\label{eq:total}
  \begin{array}{r@{{}\leftarrow{}}l}
  \vlabel{V}{\plus}   ;\vlabel{V}{\minus}    & \vertex{V}.  \\
  \elabel{U}{V}{\plus};\elabel{U}{V}{\minus} & \edge{U}{V}.
  \end{array}  
\end{equation}

Moreover, the following rules ensure that known labels are
respected by total labelings:
\begin{equation}\label{eq:known}
  \begin{array}{r@{{}\leftarrow{}}l}
  \vlabel{V}{S}    & \obsvlabel{V}{S}.    \\
  \elabel{U}{V}{S} & \obselabel{U}{V}{S}.
  \end{array}  
\end{equation}

Note that the stability criterion for answer sets demands that
a known label derived via a rule in~(\ref{eq:known}) is also derived
via~(\ref{eq:total}), thus, excluding the opposite label.
In fact, the disjunctive rules used in this section could actually
be replaced with non-disjunctive rules via ``shifting'' \cite{geliprtr91},%
\footnote{Alternatively, one could also use cardinality constraints
  (cf.~\cite{lparseManual}),
  which would however preclude a comparison with \textit{dlv} in Section~\ref{sec:benchmark}.}
given that our first encoding results in a so-called \emph{head-cycle-free} (HCF) \cite{bendec94a}
ground program.
However, similar disjunctive rules are also used in Section~\ref{sec:diagnosis}
where they cannot be compiled away.
Also note that HCF programs, for which deciding answer set existence stays in NP,
are recognized as such by disjunctive ASP solvers \cite{dlv03a,drgegrkakoossc08a}.
Hence, the purely syntactic use of disjunction, as done here, is not harmful to efficiency.
%
%% \begin{center}
%% \begin{tabular}{rcl}
%%   vlabel(V,p);vlabel(V,n) & $\leftarrow$ & vertex(V).\\
%%   elabel(U,V,p);vlabel(V,n) & $\leftarrow$& edge(U,V).\\
%% \end{tabular}
%% \end{center}
%% where the positive (resp. negative) sign is represented by the atom $p$ (resp. $n$).
%% The third part of the program is used to select those solutions that
%% satisfy the constraint~\ref{eq:consistency}.

% Given the facts in~(\ref{eq:ex:lac:facts}),
The following ground rules are obtained by combining
the schematic rules in~(\ref{eq:total}) and~(\ref{eq:known}) with
the facts in~(\ref{eq:ex:lac:facts}):
\begin{equation}\label{eq:ex:lac:guess}
  \begin{array}{r@{{}\leftarrow{}}l}
  \vlabel{\text{LacI}}{\plus}   ;\vlabel{\text{LacI}}{\minus}    & \vertex{\text{LacI}}.  \\
  \vlabel{\text{LacY}}{\plus}   ;\vlabel{\text{LacY}}{\minus}    & \vertex{\text{LacY}}.  \\[1mm]
  \elabel{\text{LacI}}{\text{LacY}}{\plus};\elabel{\text{LacI}}{\text{LacY}}{\minus} &
    \edge{\text{LacI}}{\text{LacY}}. \\[1mm]
  \vlabel{\text{LacI}}{\plus}    & \obsvlabel{\text{LacI}}{\plus}.    \\
  \elabel{\text{LacI}}{\text{LacY}}{\minus} & \obselabel{\text{LacI}}{\text{LacY}}{\minus}.
  \end{array}  
\end{equation}
One can check that the program consisting of the facts
in~(\ref{eq:ex:lac:facts}) and the rules in~(\ref{eq:ex:lac:guess}) admits two answer sets,
the first one including $\vlabel{\text{LacY}}{\plus}$ and the second one including $\vlabel{\text{LacY}}{\minus}$.
On the remaining atoms, both answer sets coincide by containing the atoms in~(\ref{eq:ex:lac:facts})
along with $\vlabel{\text{LacI}}{\plus}$ and $\elabel{\text{LacI}}{\text{LacY}}{\minus}$.

\subsection{Testing Solution Candidates}
We now check whether generated total labelings
satisfy the sign consistency constraints stated in Definition~\ref{def:consistency},
requiring an influence of sign~$s$
for each non-input vertex~$i$ with variation~$s$.
We thus define $\influence{i}{s}$ to indicate that~$i$
receives an influence of sign~$s$:
% ,
% where the encoding additionally contains facts $\sign{p}$ and $\sign{n}$:
%
\begin{equation}\label{eq:influence}
  \begin{array}{r@{{}\leftarrow{}}l}
  \influence{V}{\plus}  & \elabel{U}{V}{S}, \vlabel{U}{S}. \\
  \influence{V}{\minus} & \elabel{U}{V}{S}, \vlabel{U}{T}, S\neq T.
  \end{array}  
\end{equation}

Inconsistent labelings, where a non-input vertex does not receive
any influence matching its variation, are then ruled out by integrity
constraints of the following form:
\begin{equation}\label{eq:inconsistent}
  \begin{array}{@{{}\leftarrow{}}l}
  \vlabel{V}{S}, \naf{\influence{V}{S}}, \naf{\iinput{V}}.
  \end{array}  
\end{equation}
Note that the schematic rules in~(\ref{eq:influence}) and~(\ref{eq:inconsistent})
are given in the input language of grounder \gringo\ \cite{potasscoManual}.
% , available at \cite{gringo}.
This allows us to omit an explicit listing of some ``domain predicates'' in the bodies of rules,
which would be necessary when using \lparse\ \cite{lparseManual}.
At \cite{bioasptoolchain}, we provide encodings for \gringo\ and also (more verbose ones) for \lparse.

Starting from the answer sets described in the previous subsection, the included atoms
$\elabel{\text{LacI}}{\text{LacY}}{\minus}$
and $\vlabel{\text{LacI}}{\plus}$ allow us to derive $\influence{\text{LacY}}{\minus}$ via
a ground instance of the second rule in~(\ref{eq:influence}),
while $\influence{\text{LacY}}{\plus}$ is not derivable.
After adding $\influence{\text{LacY}}{\minus}$,
the solution candidate containing $\vlabel{\text{LacY}}{\minus}$ satisfies the ground
instance of the integrity constraint in~(\ref{eq:inconsistent}) obtained by 
substituting $\text{LacY}$ for~$V$ and~$\minus$ for~$S$.
Assuming LacI to be an input, as it can be declared via fact $\iinput{\text{LacI}}$,
we thus obtain an answer set containing $\vlabel{\text{LacY}}{\minus}$,
expressing a decrease of~LacY.
In contrast, since $\influence{\text{LacY}}{\plus}$ is underivable,
the solution candidate containing $\vlabel{\text{LacY}}{\plus}$ violates
the following ground instance of~(\ref{eq:inconsistent}):
\begin{equation*}
  \begin{array}{@{{}\leftarrow{}}l}
  \vlabel{\text{LacY}}{\plus}, \naf{\influence{\text{LacY}}{\plus}}, \naf{\iinput{\text{LacY}}}.
  \end{array}  
\end{equation*}
That is, the solution candidate with $\vlabel{\text{LacY}}{\plus}$ does not pass the consistency test.

\subsection{Soundness and Completeness}
\label{sec:checking:theorems}

By letting $\tau((V,E,\sigma),\mu)$ denote the set of facts representing the
problem instance induced by an influence graph $(V,E,\sigma)$ and a vertex
labeling~$\mu$,
and~$P_C$ the logic program consisting of the rules given in~(\ref{eq:total}),
(\ref{eq:known}), (\ref{eq:influence}), and~(\ref{eq:inconsistent}), respectively,
we can show the following soundness and completeness results.
\begin{theorem}[Soundness]\label{thm:cons-sound}
Let $(V,E,\sigma)$ be an influence graph and
$\mu : V \rightarrow \{\plus,\minus\}$ a (partial) vertex labeling.

If there is an answer set of $P_C\cup\tau((V,E,\sigma),\mu)$,
then $(V,E,\sigma)$ and~$\mu$ are consistent.
\end{theorem}
\begin{theorem}[Completeness]\label{thm:cons-compl}
Let $(V,E,\sigma)$ be an influence graph and
$\mu : V \rightarrow \{\plus,\minus\}$ a (partial) vertex labeling.

If $(V,E,\sigma)$ and $\mu$ are consistent,
then there is an answer set of $P_C \cup \tau((V,E,\sigma),\mu)$.
\end{theorem}
%
% The detailed proofs of these results are given in~\ref{app:proof:consistency}.

The following correspondence result is immediately obtained from
Theorem~\ref{thm:cons-sound} and~\ref{thm:cons-compl}.
\begin{corollary}[Soundness and Completeness]\label{col:cons-iff}
Let  $(V,E,\sigma)$ be an influence graph and
$\mu : V \rightarrow \{\plus,\minus\}$ a (partial) vertex labeling.

Then, 
$(V,E,\sigma)$ and~$\mu$ are consistent
iff
there is an answer set of $P_C\cup\tau((V,E,\sigma),\mu)$.
\end{corollary}

%%% Local Variables: 
%%% mode: latex
%%% TeX-master: "paper"
%%% End: 

\section{Identifying Minimal Inconsistent Cores}\label{sec:diagnosis}

In view of the usually large amount of data,
it is crucial to provide concise explanations
whenever an experimental profile is inconsistent with an influence graph
(i.e., if the logic program given in the previous section has no answer set).
To this end, we adopt a strategy that was successfully applied on real biological data \cite{coliRIAMS}.
The basic idea is to isolate minimal subgraphs of an influence graph
such that the vertices and edges cannot be labeled consistently.  % in a consistent way.
This task is closely related to extracting Minimal Unsatisfiable Cores
(MUCs) \cite{scalableMUC} in the context of Boolean satisfiability
(SAT).
In allusion,
we call a minimal subgraph of an influence graph
whose vertices and edges cannot be labeled consistently
a \emph{Minimal Inconsistent Core} (MIC),
whose formal definition is as follows.\footnote{% \cite{mitchell05a}.
We note that verifying a MUC is D$^{\text{P}\!}$-complete \cite{scalableMUC,papyan82a},
and the same applies to MICs in view of the reduction of SAT described in Section~\ref{sec:checking}.
However, solving a decision problem is not sufficient for our application
because we also need to provide MIC candidates to verify.
% As claimed in \cite{grmapi08a}, $\Sigma_2^P$-hardness of
% checking whether a set of clauses belongs to at least one MUC
% follows from a complexity result in \cite{eitgot92}.
% Since our computational task includes guessing MIC candidates
% in order to find the MICs among them,
% $\Sigma_2^P$-hardness claimed for the related MUC problem \cite{grmapi08a}
% justifies the choice of disjunctive ASP as target framework.
% In fact, while deciding the existence of a MIC (checking inconsistency) is in co-NP,
As regards checking inconsistency of an (a priori unknown) MIC candidate,
we are unaware of ways to
accomplish such a co-NP test in non-disjunctive ASP without 
destroying the candidate at hand.}
\begin{definition}[Minimal Inconsistent Core]\label{def:mic}
  Let $(V,E,\sigma)$ be an influence graph and
  $\mu : V \rightarrow \{\plus,\minus\}$ a (partial) vertex labeling.

  Then, a subset~$W$ of~$V$ is a \emph{Minimal Inconsistent Core} (MIC),
  if
  \begin{enumerate}
  \item
    for all total extensions $\sigma' : E \rightarrow \{\plus,\minus\}$ of $\sigma$
    and $\mu' : V \rightarrow \{\plus,\minus\}$ of $\mu$,
    there is some non-input vertex $i\in W$ such that $\mu'(i)$ is inconsistent, and
  \item
    for every $W'\subset W$,
    there are some total extensions
    $\sigma' : E \rightarrow \{\plus,\minus\}$ of $\sigma$
    and $\mu' : V \rightarrow \{\plus,\minus\}$ of $\mu$
    such that $\mu'(i)$ is consistent for each non-input vertex $i\in W'$.
  \end{enumerate}
\end{definition}
To encode MICs, we make use of three important observations made on Definition~\ref{def:mic}.
First, the inherent inconsistency of a MIC's vertices stipulated in
the first condition must be implied by the MIC and its external regulators,
while vertices not connected to the MIC cannot contribute anything.
Moreover, the second condition on proper subsets prohibits the inclusion
of an input vertex in a MIC, as it could always be removed without
affecting inherent (in)consistency of the remaining vertices' variations.
Finally, for establishing consistency of all proper subsets of a MIC,
it is sufficient to consider subsets excluding a single vertex of the MIC,
given that their consistency carries forward to all smaller subsets.

For illustration, consider the influence graph and the MIC in
Figure~\ref{fig:mic-illustration}.
One can check that the observed simultaneous increase of~{\bf B} and~{\bf D}
is not consistent with the influence graph,
but the reason for this might not be apparent at first glance. 
However, once the MIC consisting of~{\bf A} and~{\bf D} is extracted, 
we see that the increase of~{\bf B} implies an increase of~{\bf A}, 
so that the observed increase of~{\bf D} cannot be explained.
Note that the elucidation of inherent inconsistency provided by a MIC
takes its vertices along with their regulators into account, the latter
being incapable of jointly explaining the variations of all vertices in the MIC.

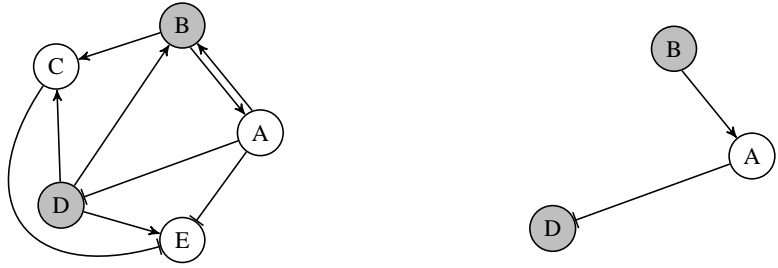
\begin{figure}[tb]
\centering
\begin{tabular}{l@{\hspace{30mm}}l}
  \begin{tikzpicture}[->,semithick,>=stealth']
 %   \scriptsize
    \tikzstyle{species}=[draw, circle, fill=none,text=black,minimum size=17pt]
    \tikzstyle{posspecies}=[style=species,text=black,text opacity=1,minimum size=17pt,fill=lightgray]
%     \tikzstyle{posspecies}=[style=species,text=black,minimum size=17pt,label=160:+]
    \tikzstyle{negspecies}=[style=species,text=black,text opacity=1,minimum size=17pt,fill=darkgray ]% label=60:-]%
    \node[species] (A) at (0:1.5)  {A};
    \node[posspecies] (B) at (72:1.5) {B};
    \node[species] (C) at (144:1.5) {C};
    \node[posspecies] (D) at (220:1.5) {D};
    \node[species] (E) at (288:1.5) {E};
    
    \path
    (A.110) edge (B.-50)
    (B.-70) edge (A.130)
    (A) edge[-|]  (D)
    (A) edge[-|] (E)
    (D) edge (E)
    (B) edge (C)
    (D) edge (B)
    (D) edge (C)
    (C) edge[bend right=70,looseness=1.8,-|] (E);
  \end{tikzpicture}
&
  \begin{tikzpicture}[->,semithick,>=stealth']
%    \scriptsize
    \tikzstyle{species}=[draw, circle, fill=none,text=black,minimum size=17pt]
%     \tikzstyle{posspecies}=[style=species,text=black,minimum size=17pt,label=160:+]%,
    \tikzstyle{posspecies}=[style=species,text=black,text opacity=1,minimum size=17pt,fill=lightgray]
%     \tikzstyle{negspecies}=[style=species,text=black,text opacity=1,minimum size=17pt,fill=red, fill opacity=0.8 ]% label=60:-]%
     \tikzstyle{negspecies}=[minimum size=17pt]%,fill=darkgray
    \node[species] (A) at (0:1.5)  {A};
    \node[posspecies] (B) at (72:1.5) {B};
    \node[negspecies] (C) at (144:1.5) {};
    \node[posspecies] (D) at (220:1.5) {D};
    \node[negspecies] (E) at (288:1.5) {};
    
    \path
    (B.-70) edge (A.130)
    (A) edge[-|]  (D);
  \end{tikzpicture}
%\vspace{-2.5ex}
\end{tabular}
  \caption{A partially labeled influence graph and a MIC consisting of~{\bf A} and~{\bf D}.}
  \label{fig:mic-illustration}
%\vspace{-2.5ex}
\end{figure}

We next provide an encoding for identifying MICs,
where a problem instance, that is,
an influence graph along with an experimental profile,
is represented by facts as specified in Section~\ref{subsec:instance}.
The encoding then consists of three parts: 
the first generating MIC candidates,
the second asserting inconsistency,
and the third verifying minimality.

\subsection{Generating MIC Candidates} 

The generating part comprises rules 
in~(\ref{eq:known}) for deriving known vertex and edge labels.
In addition, it includes the following rules:
\begin{equation}\label{eq:active}
  \begin{array}{r@{{}\leftarrow{}}l}
  \activ{V};\inactiv{V}                      & \vertex{V}, \naf{\iinput{V}}. \\[1mm]
  \aedge{U}{V}                               & \edge{U}{V}, \activ{V} . \\
  \avertex{U}                                & \aedge{U}{V}. \\
  \avertex{V}                                & \activ{V}. \\[1mm]
  \vlabel{V}{\plus}   ;\vlabel{V}{\minus}    & \avertex{V}.  \\
  \elabel{U}{V}{\plus};\elabel{U}{V}{\minus} & \aedge{U}{V}.
  \end{array}  
\end{equation}
The first rule permits guessing non-input vertices forming a MIC candidate.
Such vertices are marked as $\atomfont{active}$.
The subgraph of the influence graph consisting of the active vertices,
their regulators, and the connecting edges provides the context of the 
MIC candidate.%
\footnote{In Definition~\ref{def:mic}, (in)consistency is checked only
for the (non-input) vertices in a MIC, while other vertices' variations do not need to
be explained.
Hence, guessing unobserved vertex (and edge) labels
can be restricted to vertices belonging to or connected to the MIC,
which reduces combinatorics.
%This optimization has not been described in \cite{gescthusve08b}.
}
The vertices and edges contributing to this subgraph are identified via
$\atomfont{vertexMIC}$ and $\atomfont{edgeMIC}$.
The guessing of (unobserved) vertex and edge labels is restricted to them
in the last two rules of~(\ref{eq:active}).
Finally, note that the rules in~(\ref{eq:known}) propagate known labels
also for vertices and edges not correlated to the MIC candidate, viz.,
to the active vertices.
This does not incur additional combinatorics; rather,
it reduces derivations depending on MIC candidates.

\subsection{Testing for Inconsistency}\label{subsec:muc:inconsistent}
By adapting the methodology used in \cite{eitgot95a},
the following subprogram
makes sure that the active vertices cannot be labeled consistently,
% belong to a subgraph
% that cannot be labeled consistently,
% adapting the methodology described in \cite{eitgot95a},
taking (implicitly) into account all possible labelings for them,
their regulators, and connecting edges:%
\footnote{%
In the language of \gringo\ (and \lparse),
the expression $\contrary{U}{V}:\edge{U}{V}$ used below refers
to the conjunction of all ground atoms $\contrary{j}{i}$ for which $\edge{j}{i}$ holds.}%
\begin{equation}\label{eq:inconsistency}
  \begin{array}{r@{}c@{}l}
  \contrary{U}{V}  & {}\leftarrow{} &
                     \elabel{U}{V}{\minus}, \vlabel{U}{S}, \vlabel{V}{S}.
  \\
  \contrary{U}{V}  & {}\leftarrow{} &
                     \elabel{U}{V}{\plus},  \vlabel{U}{S}, \vlabel{V}{T}, S\neq T.
  \\[1mm]
  \bottom          & {}\leftarrow{} &
                     \activ{V}, \contrary{U}{V}:\edge{U}{V}.
  \\               & {}\leftarrow{} &
                     \naf{\bottom}.
  \\[1mm]
  \multicolumn{3}{@{}l}{
  \begin{array}{r@{}c@{}l}
  \vlabel{V}{\plus}  & {}\leftarrow{} &
                       \bottom, \vertex{V}.
  \\
  \vlabel{V}{\minus} & {}\leftarrow{} &
                       \bottom, \vertex{V}.
  \\
  \elabel{U}{V}{\plus}  & {}\leftarrow{} &
                          \bottom, \edge{U}{V}.
  \\
  \elabel{U}{V}{\minus} & {}\leftarrow{} &
                          \bottom, \edge{U}{V}.
  \end{array} 
  }
  \end{array}  
\end{equation}
%
% \begin{center}
% \em
% \begin{tabular}{rcl}
%   contr\_infl(U,V) &$\leftarrow$& elabel(U,V,S), vlabel(U,S), vlabel(V,n), edge(U,V), sign(S), active(V).\\
%   contr\_infl(U,V) &$\leftarrow$& elabel(U,V,S), vlabel(U,T),
%   vlabel(V,p), edge(U,V), sign(S), \\
%   & & sign(T), S != T, active(V).\\
%   bottom &$\leftarrow$& active(V),contr\_infl(U,V) : edge(U,V),sign(S),vertex(V).\\
%   vlabel(V,S) &$\leftarrow$& bottom, vertex(V), sign(S).\\
%   elabel(U,V,S) &$\leftarrow$& bottom, edge(U,V), sign(S).\\
%   bottom &$\leftarrow$& not~bottom.\\
% \end{tabular}
% \end{center}
In this (part of the) encoding,
% predicate 
$\contrary{U}{V}$ indicates that the influence of regulator~$U$
on~$V$ is opposite to the variation of~$V$.
If all regulators of an active vertex~$V$ have such an opposite influence,
the sign consistency constraint for~$V$ is violated,
in which case atom $\bottom$ along with all labels for vertices and edges are derived.
Note that the stability criterion for an answer set~$X$ imposes
that $\bottom$ and all labels belong to~$X$ only
if the active vertices cannot be labeled consistently.
Finally, integrity constraint ${\leftarrow}\,\naf{\bottom}$
necessitates the inclusion of~$\bottom$ in any answer set,
thus, stipulating an inevitable sign consistency constraint 
violation for some active vertex.
%
% As an answer set is minimal with respect to
% inclusion, if there is a means not to produce bottom (that is, if
% there is a means to avoid inconsistencies), it has to be
% selected. Finally, the integrity constraint excluding the absence of
% $\bottom$ assures that the active set should describe an unsatisfiable
% subset of constraints.

Reconsidering our example in Figure~\ref{fig:mic-illustration},
the ground instances of (\ref{eq:active}) permit guessing
$\activ{\text{\bf A}}$ and $\activ{\text{\bf D}}$.
When labeling {\bf A} with~$\plus$ (or assuming $\vlabel{\text{\bf A}}{\plus}$ to be true),
we derive $\contrary{\text{\bf A}}{\text{\bf D}}$ and~$\bottom$,
producing in turn all labels for vertices and edges.
Furthermore, setting the sign of {\bf A} to~$\minus$ (or $\vlabel{\text{\bf A}}{\minus}$ to true)
makes us derive $\contrary{\text{\bf B}}{\text{\bf A}}$, which again gives
$\bottom$ and all labels for vertices and edges.
We have thus verified that the sign consistency constraints for {\bf A} and {\bf D}
cannot jointly be satisfied, given the observed increases of {\bf B} and {\bf D}.
That is, active vertices {\bf A} and {\bf D} are sufficient to explain the inconsistency
between the observations and the influence graph.

\subsection{Testing for Minimality}\label{subsec:minimal}
% Finally, 
It remains to be verified whether the sign consistency constraints for all active vertices
are necessary to identify an inherent inconsistency.
This test is based on the idea that, excluding any single active vertex, 
the sign consistency constraints for the other active vertices should be 
satisfied by appropriate labelings,
which can be implemented as follows:
%This conception is implemented in the following subprogram:
%
\begin{equation}\label{eq:minimal}
  \begin{array}{r@{}c@{}l}
  \vlabell{W}{V}{\plus}; \vlabell{W}{V}{\minus} & {} \leftarrow {} &
    \activ{W}, \avertex{V}.
  \\
  \elabell{W}{U}{V}{\plus}; \elabell{W}{U}{V}{\minus} & {} \leftarrow {} &
    \activ{W}, \aedge{U}{V}.
  \\[1mm]
  \multicolumn{3}{@{}l}{
  \begin{array}{@{}r@{}c@{}l}
  \vlabell{W}{V}{S} & {} \leftarrow {} &
    \activ{W}, \obsvlabel{V}{S}.
  \\
  \elabell{W}{U}{V}{S} & {} \leftarrow {} &
    \activ{W}, \obselabel{U}{V}{S}.
  \\[1mm]
  \influencel{W}{V}{\plus} & {} \leftarrow {} & 
    \elabell{W}{U}{V}{S}, \vlabell{W}{U}{S}, V\neq W.
  \\
  \influencel{W}{V}{\minus} & {} \leftarrow {} & 
    \elabell{W}{U}{V}{S}, \vlabell{W}{U}{T}, V\neq W, S\neq T.
  \\[1mm]
                       & {} \leftarrow {} &
    \vlabell{W}{V}{S}, \activ{V}, V\neq W, \naf{\influencel{W}{V}{S}}.\hspace*{-10mm}
  \end{array}
  }
  \end{array}  
\end{equation}
% \begin{center}
% \em
% \begin{tabular}{rcl}
%   vlabel'(W,V,p) ; vlabel'(W,V,n) & $\leftarrow$& vertex(V),vertex(W),active(W).\\
%   elabel'(W,U,V,p) ; elabel'(W,U,V,n) & $\leftarrow$& edge(U,V),vertex(W),active(W).\\
%   vlabel'(W,V,S) &$\leftarrow$& obs\_vlabel(V,S),vertex(W),active(W).\\
%   elabel'(W,U,V,S) &$\leftarrow$& obs\_elabel(U,V,S),vertex(W),active(W).\\
% \end{tabular}
% \begin{tabular}{rcl}
%   infl'(W,V,p) &$\leftarrow$& elabel'(W,U,V,S), vlabel'(W,U,S), edge(U,V), sign(S), \\
%   && active(V),vertex(W),active(W).\\
%   infl'(W,V,n) &$\leftarrow$& elabel'(W,U,V,S), vlabel'(W,U,T), S != T, edge(U,V), sign(S), sign(T), \\
%       &&         active(V),vertex(W),active(W).\\
% \end{tabular}
% \begin{tabular}{rcl}
% & $\leftarrow$ & vlabel'(W,V,S), not infl'(W,V,S), V != W,vertex(V), active(V),vertex(W), active(W),\\
% &  & sign(S).\\
% \end{tabular} 
% \end{center}
% Basically, this part of the program searches consistent labelings for
% all vertices of the active set, except one. And this is done for each
% vertex in the active set.
This subprogram is similar to the consistency check
encoded via the rules in~(\ref{eq:total}), (\ref{eq:known}), (\ref{eq:influence}), and (\ref{eq:inconsistent}).
However, sign consistency constraints are only checked for
active vertices, and they must be satisfiable for all but 
one arbitrary active vertex~$W$.
In fact, labelings such that the variations of all active vertices
but~$W$ are explained witness the fact that~$W$
cannot be removed from a MIC candidate without re-establishing consistency.
As~$W$ ranges over all (non-input) vertices of an influence graph,
each active vertex is taken into consideration.
Regarding computational complexity,
recall from Section~\ref{sec:checking} that checking consistency is NP-complete.
As a consequence, one cannot easily identify conditions to select a particular
witness for consistency of a MIC candidate minus some vertex~$W$, and so we do not
encode any such conditions.
This leads to the potential of multiple answer sets comprising the same MIC
but different witnesses, in particular, if many vertices and edges 
belong to the context of the MIC.

For the influence graph in Figure~\ref{fig:mic-illustration},
it is easy to see that the sign consistency constraint for~{\bf A}
is satisfied by setting the sign of~{\bf A} to~$\plus$,
expressed by atom $\vlabell{\text{\bf D}}{\text{\bf A}}{\plus}$
in the ground rules obtained from the above encoding part.
In turn, the sign consistency constraint for~{\bf D} is satisfied
by setting the sign of~{\bf A} to~$\minus$.
This is reflected by atom $\vlabell{\text{\bf A}}{\text{\bf A}}{\minus}$,
allowing us to derive $\influencel{\text{\bf A}}{\text{\bf D}}{\plus}$.
That is, the ground instance of the above integrity constraint
containing $\vlabell{\text{\bf A}}{\text{\bf D}}{\plus}$ is satisfied.
The fact that atoms $\vlabell{\text{\bf D}}{\text{\bf A}}{\plus}$ and
$\vlabell{\text{\bf A}}{\text{\bf A}}{\minus}$,
used for explaining the variation of either~{\bf A} or~{\bf D}, respectively,
disagree on the sign of~{\bf A} also shows that
jointly considering~{\bf A} and~{\bf D} yields an inconsistency.

\subsection{Soundness and Completeness}\label{sec:diagnosis:theorems}

Similar to Section~\ref{sec:checking:theorems},
we can show the soundness and completeness for our MIC extraction
encoding~$P_D$, consisting of the rules
in~(\ref{eq:known}), (\ref{eq:active}), (\ref{eq:inconsistency}), and~(\ref{eq:minimal}),
respectively.
\begin{theorem}[Soundness]\label{thm:diag-sound}
Let  $(V,E,\sigma)$ be an influence graph and
$\mu : V \rightarrow \{\plus,\minus\}$ a (partial) vertex labeling.

If $X$ is an answer set of $P_D\cup\tau((V,E,\sigma),\mu)$,
then $\{i \mid \activ{i}\in X\}$ is a MIC.
\end{theorem}
\begin{theorem}[Completeness]\label{thm:diag-compl}
Let  $(V,E,\sigma)$ be an influence graph and
$\mu : V \rightarrow \{\plus,\minus\}$ a (partial) vertex labeling.

If $W\subseteq V$ is a MIC, then there is an answer
set~$X$ of $P_D\cup\tau((V,E,\sigma),\mu)$
such that $\{i \mid \activ{i}\in X\}=W$.
\end{theorem}

The following correspondence result is immediately obtained from
Theorem~\ref{thm:diag-sound} and~\ref{thm:diag-compl}.
\begin{corollary}[Soundness and Completeness]\label{col:diag-iff}
Let  $(V,E,\sigma)$ be an influence graph and
$\mu : V \rightarrow \{\plus,\minus\}$ a (partial) vertex labeling.

Then, 
$W\subseteq V$ is a MIC
iff
there is an answer set~$X$ of $P_D\cup\tau((V,E,\sigma),\mu)$
such that $\{i \mid \activ{i}\in X\}=W$.
\end{corollary}
As mentioned above,
several answer sets may represent the same MIC
because witnesses needed for minimality testing are
not necessarily unique.

%%% Local Variables: 
%%% mode: latex
%%% TeX-PDF-mode: t 
%%% TeX-master: "paper"
%%% End: 

\section{Refinements}
\label{sec:refinements}
In this section, we detail two encoding extensions aiming at
the improvement of grounding and solving efficiency.
First, input reduction checks for some simple cases to identify and
distinguish uncritical vertices.
Second, background knowledge about MICs' connectivity can be exploited
to more precisely render potential MIC candidates. 

\subsection{Input Reduction}\label{sec:reduction}

It is not unlikely in practice that biological networks include simple tractable substructures
or that parts of experimental observations are easily explained.
Dealing with such particular cases before doing complex computations
(like checking consistency or finding MICs) is therefore advisable.
Given an influence graph $(V,E,\sigma)$ and a partial vertex labeling~$\mu$
capturing experimental data, we below describe conditions to identify
vertices that can always be labeled consistently.
Such vertices can then be marked as (additional) inputs to exclude their
sign consistency constraints from consistency checking and to make
explicit that they cannot belong to any MIC.
% The idea is to add a vertex~$i$ to a set~$I$,
% initially consisting of all input vertices from~$V$,
% if any of the following cases is applicable:
Any of the following conditions is sufficient to identify a vertex~$i$
as effectively unconstrained:
\begin{enumerate}
\item There is a regulation $i {\,\rightarrow\,} i$ in~$E$ such that
      $\sigma(i,i)=\plus$, that is,
      $i$ supports its variation.
\item There is a regulation $j {\,\rightarrow\,} i$ in~$E$ such that
      $\sigma(j,i)$ is undefined.
      In fact, undetermined regulations are used in practice to model
      influences that vary, e.g., relative to environmental conditions.
      Any variation of the target~$i$ of such a regulation
      can be explained by assigning the appropriate label
      to $j {\,\rightarrow\,} i$ (w.r.t.\ the label of~$j$).
\item There are regulations $j {\,\rightarrow\,} i,k {\,\rightarrow\,} i$ in~$E$
      such that $\mu(j)\sigma(j,i)\!=\!\plus$ and $\mu(k)\sigma(k,i)\!=\nolinebreak\!\minus$.
      That is, any variation of~$i$ is already explained by the given observations.
\item An observed variation $\mu(i)$ of~$i$ is explained if there is
      some regulation $j {\,\rightarrow\,} i$ in~$E$ such that $\mu(j)\sigma(j,i)=\mu(i)$.
      Any further regulations targeting~$i$ can be ignored.
\item If for all regulations $i {\,\rightarrow\,} k$ in~$E$,
      we have that~$k$ is an input,
      then the variation of~$i$ is insignificant for its targets.
      In this case, if~$i$ is unobserved ($\mu(i)$ is undefined) and target
      of at least one regulation $j {\,\rightarrow\,} i$ in~$E$, % such that $j\neq i$,
      we can assign an appropriate label to~$i$
      (w.r.t.\ the labels of~$j$ and $j {\,\rightarrow\,} i$)
      without any further conditions.
%      (Requirement $j\neq i$ excludes $i {\,\rightarrow\,} i$ in~$E$ such that
%       $\sigma(i,i)=\minus$ as explanation for~$i$.)
\item There is a regulation $j {\,\rightarrow\,} i$ in~$E$ such that $j$
      is unobserved ($\mu(j)$ is undefined), an input, and
      all targets $k\neq i$ of~$j$ ($j {\,\rightarrow\,} k$ belongs to~$E$) are inputs.
      Without any further conditions, we can assign an appropriate label to~$j$
      for explaining the variation of~$i$.
\end{enumerate}

The reduction idea is to mark a vertex~$i$ as additional input, 
if it meets one of the above conditions.
Since the two last conditions inspect inputs,
they may become applicable to further vertices once inputs are added.
Hence, checking the conditions and adding inputs needs to be done
exhaustively.
As we see below, this can easily be encoded in ASP.

Reconsidering the influence graph and partial observations in
Figure~\ref{fig:mic-illustration}, we see that vertex~{\bf B}
receives an influence from~{\bf D} matching its observed increase.
Thus, the fourth condition applies to already explained vertex~{\bf B}.
Moreover, vertex~{\bf E}
is unobserved and does not regulate anything.
That is, the fifth condition applies to~{\bf E}, and its variation can
simply be picked from influences it receives from 
{\bf A}, {\bf C}, and~{\bf D}.
After establishing that~{\bf E} can be labeled consistently,
we find that~{\bf C} does not regulate any critically constrained vertex.
Applying again the fifth condition,
we notice that the variation of~{\bf C} is actually insignificant.

\begin{figure}[tb]
\centering
%% colored picture
%   \begin{tikzpicture}[->,semithick,>=stealth']
%  %   \scriptsize
%     \tikzstyle{species}=[draw, circle, fill=none,text=black,minimum size=17pt]
%     \tikzstyle{posspecies}=[style=species,text=black,text opacity=1,minimum size=17pt,fill=green, fill opacity=0.5,]
% %     \tikzstyle{posspecies}=[style=species,text=black,minimum size=17pt,label=160:+]
%     \tikzstyle{negspecies}=[style=species,text=black,text opacity=1,minimum size=17pt,fill=red, fill opacity=0.8 ]% label=60:-]%
%     \node[species] (A) at (0:1.5)  {A};
%     \node[posspecies,dotted] (B) at (72:1.5) {B};
%     \node[species,dotted] (C) at (144:1.5) {C};
%     \node[posspecies] (D) at (220:1.5) {D};
%     \node[species,dotted] (E) at (288:1.5) {E};
%     
%     \path
%     (A.110) edge[green] (B.-50)
%     (B.-70) edge[green] (A.130)
%     (A) edge[red,-|]  (D)
%     (A) edge[red,-|] (E)
%     (D) edge[green] (E)
%     (B) edge[green] (C)
%     (D) edge[green] (B)
%     (D) edge[green] (C)
%     (C) edge[bend right=70,looseness=1.8,-|,red] (E);
%   \end{tikzpicture}
%   \caption{A partially labeled influence graph with uncritical vertices
%            surrounded by dots.}
%   \label{fig:reduction}
%%% black white picture  
  \begin{tikzpicture}[->,semithick,>=stealth']
 %   \scriptsize
    \tikzstyle{species}=[draw, circle, fill=none,text=black,minimum size=17pt]
    \tikzstyle{posspecies}=[style=species,text=black,text opacity=1,minimum size=17pt,fill=lightgray]
%     \tikzstyle{posspecies}=[style=species,text=black,minimum size=17pt,label=160:+]
    \tikzstyle{negspecies}=[style=species,text=black,text opacity=1,minimum size=17pt,fill=darkgray ]% label=60:-]%
    \node[species] (A) at (0:1.5)  {A};
    \node[posspecies,dotted] (B) at (72:1.5) {B};
    \node[species,dotted] (C) at (144:1.5) {C};
    \node[posspecies] (D) at (220:1.5) {D};
    \node[species,dotted] (E) at (288:1.5) {E};
    
    \path
    (A.110) edge (B.-50)
    (B.-70) edge (A.130)
    (A) edge[-|]  (D)
    (A) edge[-|] (E)
    (D) edge (E)
    (B) edge (C)
    (D) edge (B)
    (D) edge (C)
    (C) edge[bend right=70,looseness=1.8,-|] (E);
  \end{tikzpicture}
  \caption{A partially labeled influence graph with uncritical vertices
           surrounded by dots.}
  \label{fig:reduction}
\end{figure}
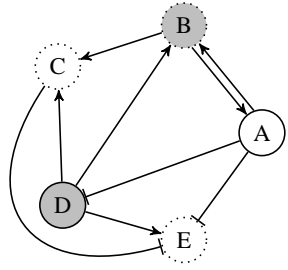

Figure~\ref{fig:reduction} shows the situation resulting from the
identification of uncritical vertices by
iteratively applying the above conditions.
The fact that only~{\bf A} and~{\bf D} are critically constrained tells us that
only they can belong to a MIC.
As a consequence, the MIC containing~{\bf A} and~{\bf D},
shown on the right-hand side of Figure~\ref{fig:mic-illustration},
is the only one in this example.

The aforementioned idea to mark uncritical vertices as $\atomfont{input}$
can be encoded as follows:
% The basic idea of encoding the identification of uncritical vertices is
% to mark them as $\atomfont{input}$, having the effect that their sign
% consistency constraints are not subject to checks
% when deciding consistency or finding MICs.
% This is realized by the following subprogram:
%
\begin{equation*}
  \begin{array}{r@{{}\leftarrow{}}l}
  \observed{V}    & \obsvlabel{V}{S}. \\
  \get{V}{\plus}  & \obselabel{U}{V}{S}, \obsvlabel{U}{S}. \\
  \get{V}{\minus} & \obselabel{U}{V}{S}, \obsvlabel{U}{T}, S\neq T. \\[1mm]
  \iinput{V}      & \obselabel{V}{V}{\plus}.\\
  \iinput{V}      & \edge{U}{V}, \naf{\obselabel{U}{V}{\plus}},\naf{\obselabel{U}{V}{\minus}}.\\
  \iinput{V}      & \get{V}{\plus}, \get{V}{\minus}. \\
  \iinput{V}      & \obsvlabel{V}{S}, \get{V}{S}. \\
  \iinput{V}      & \edge{U}{V}, \iinput{W}:\edge{V}{W}, \naf{\observed{V}}. \\
  \iinput{V}      & \edge{U}{V}, \iinput{W}:\edge{U}{W}:W\neq V, \iinput{U},\naf{\observed{U}}.
  \end{array}
\end{equation*}
Auxiliary predicates $\atomfont{obs}$ and $\atomfont{get}$
are used to exhibit whether either variation has been observed
for a vertex and whether a particular influence is received for certain,
respectively.
The last six rules check the described conditions (in the same order)
and mark a vertex as $\atomfont{input}$ if one of them applies.
Importantly, the above rules are stratified
and thus yield a unique set of derived input vertices.
This allows us to  perform the reduction efficiently within grounding,
without deferring to any procedural implementation
external to ASP.

The situation shown in Figure~\ref{fig:reduction} is reflected
by the reduction encoding deriving atoms
$\iinput{\text{\bf B}}$, $\iinput{\text{\bf C}}$, and $\iinput{\text{\bf E}}$
from an instance (cf.\ Section~\ref{subsec:instance})
corresponding to the depicted influence graph and observed variations.
Consistency checking and MIC identification
(cf.\ Section \ref{sec:checking} and~\ref{sec:diagnosis})
can then focus on the remaining non-input vertices~{\bf A} and~{\bf D}.

%%% Local Variables: 
%%% mode: latex
%%% TeX-master: "paper"
%%% End: 

\subsection{Exploiting Strongly Connected Components for MIC Extraction}\label{sec:scc}

In what follows,
we introduce a connectivity property of MICs that can be
used to further refine the encoding presented in Section~\ref{sec:diagnosis}.
Incorporating additional background knowledge into the problem encoding
is straightforward (as soon as such knowledge is established).
In practice, ancillary (and actually redundant) conditions may significantly
narrow and thus speed up both the grounding and the solving process.

\paragraph{MIC Connectivity Property.}%\label{subsec:connect}

For analyzing interactions within a MIC,
we make use of a graph described in the following.
Let $(V,E,\sigma)$ be an influence graph and
$\mu : V \rightarrow \{\plus,\minus\}$ be a (partial) vertex labeling,
and let $\dom{\mu}$ denote the set of vertices labeled by~$\mu$.
% Then, by \dom{\sigma} and \dom{\mu},
% we denote the set of edges and the set of vertices,
% respectively, labeled by~$\sigma$ and~$\mu$.
For a set~$W\subseteq V$ of vertices,
we define a graph~$(\vertices{W},\edges{W})$ by:
\begin{eqnarray}
  \vertices{W} 
  & \!\!=\!\! &
  W\cup\{j \mid (j {\,\rightarrow\,} i) \in E,
         i\in W\}%\cap\dom{\sigma}\}
%\label{eq:subvertices}
\nonumber
\\
  \edges{W}  
%  \begin{array}[t]{@{}c@{}l}
  & \!\!=\!\! &
  \{
    (j {\,\rightarrow\,} i) \mid
    (j {\,\rightarrow\,} i)\in E,%\cap\dom{\sigma}, 
    i\in W
  \}  
% \\
  \cup
  \{
    (i {\,\rightarrow\,} j) \mid
    (j {\,\rightarrow\,} i)\in E,%\cap\dom{\sigma}, 
    i\in W, j\notin\dom{\mu}    
  \}\text{\,.}
%  \end{array}
%\)
\nonumber
\end{eqnarray}
The construction of~$(\vertices{W},\edges{W})$ is based on the idea
that a regulator~$j$ of some~$i\in W$ is connected to~$i$ via its
sign consistency constraint,
and a connection in the opposite direction applies if~$j$
is unlabeled by~$\mu$.
In fact, given some total extensions
$\sigma' : E \rightarrow \{\plus,\minus\}$  of~$\sigma$
and $\mu' : V \rightarrow \{\plus,\minus\}$ of~$\mu$,
we can check a matching influence of~$j$ on~$i$ by
$\mu'(i)=\mu'(j)\sigma'(j,i)$ or equivalently by $\mu'(j)=\mu'(i)\sigma'(j,i)$.
That is, provided that $\mu(j)$ is undefined,
$\mu'(i)$ constrains~$\mu'(j)$ by contraposition whenever 
$i$ does not receive a matching influence from any other regulator than~$j$.
This observation motivates the inclusion of inverse edges from vertices in~$W$
to regulators unlabeled by~$\mu$ in~$\edges{W}$.
% Finally, note that edges~$j {\,\rightarrow\,} i$ unlabeled by~$\sigma$, i.e.,
% $(j {\,\rightarrow\,} i)\notin\dom{\sigma}$, lead to an unconstrained target
% because $\sigma'(j,i)$ can be chosen as needed to establish
% $\mu'(i)=\mu'(j)\sigma'(j,i)$.
% Hence, the construction of $(\vertices{W},\edges{W})$ considers only labeled edges.

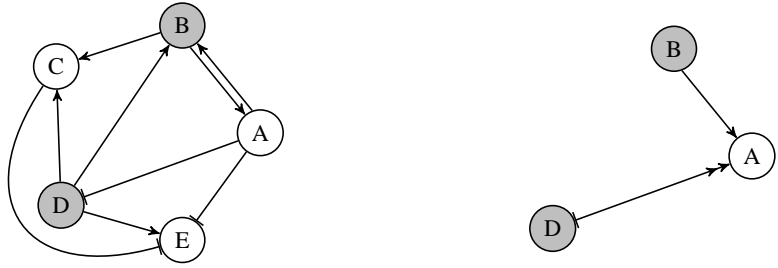
\begin{figure}[tb]
\centering
\begin{tabular}{l@{\hspace{30mm}}l}
  \begin{tikzpicture}[->,semithick,>=stealth']
 %   \scriptsize
    \tikzstyle{species}=[draw, circle, fill=none,text=black,minimum size=17pt]
    \tikzstyle{posspecies}=[style=species,text=black,text opacity=1,minimum size=17pt,fill=lightgray]
    \tikzstyle{negspecies}=[style=species,text=black,text opacity=1,minimum size=17pt,fill=darkgray ]% label=60:-]%
    \node[species] (A) at (0:1.5)  {A};
    \node[posspecies] (B) at (72:1.5) {B};
    \node[species] (C) at (144:1.5) {C};
    \node[posspecies] (D) at (220:1.5) {D};
    \node[species] (E) at (288:1.5) {E};
    
    \path
    (A.110) edge (B.-50)
    (B.-70) edge (A.130)
    (A) edge[-|]  (D)
    (A) edge[-|] (E)
    (D) edge (E)
    (B) edge (C)
    (D) edge (B)
    (D) edge (C)
    (C) edge[bend right=70,looseness=1.8,-|] (E);
  \end{tikzpicture}
&
  \begin{tikzpicture}[->,semithick,>=stealth']
%    \scriptsize
    \tikzstyle{species}=[draw, circle, fill=none,text=black,minimum size=17pt]
    \tikzstyle{posspecies}=[style=species,text=black,text opacity=1,minimum size=17pt,fill=lightgray]
     \tikzstyle{negspecies}=[minimum size=17pt]%,fill=darkgray
    \node[species] (A) at (0:1.5)  {A};
    \node[posspecies] (B) at (72:1.5) {B};
    \node[negspecies] (C) at (144:1.5) {};
    \node[posspecies] (D) at (220:1.5) {D};
    \node[negspecies] (E) at (288:1.5) {};
    
    \path
    (B.-70) edge (A.130)    
    (A) edge[<<-|]  (D);
  \end{tikzpicture}
%\vspace{-2.5ex}
\end{tabular}
  \caption{A partially labeled influence graph and the 
           graph $(\vertices{\{\text{\bf A},\text{\bf D}\}},\edges{\{\text{\bf A},\text{\bf D}\}})$.}
  \label{fig:mic-scc}
%\vspace{-2.5ex}
\end{figure}
%%% Local Variables: 
%%% mode: latex
%%% TeX-master: "paper"
%%% End: 

For illustration, the right-hand side of Figure~\ref{fig:mic-scc} shows graph
$(\vertices{\{\text{\bf A},\text{\bf D}\}},\edges{\{\text{\bf A},\text{\bf D}\}})$
resulting from the partially labeled influence graph on the left-hand side.
The single regulator~{\bf B} of~{\bf A} is labeled, and thus there is no
inverse edge from~{\bf A} to~{\bf B} in \edges{\{\text{\bf A},\text{\bf D}\}}.
On the other hand, {\bf A} is an unlabeled regulator of~{\bf D},
and so \edges{\{\text{\bf A},\text{\bf D}\}} includes an inverse edge
from~{\bf D} to~{\bf A}.
The addition of this edge turns the subgraph of
$(\vertices{\{\text{\bf A},\text{\bf D}\}},\edges{\{\text{\bf A},\text{\bf D}\}})$
induced by~{\bf A} and~{\bf D} into a strongly connected component.
In view that~{\bf A} and~{\bf D} belong to a MIC
(as discussed in Section~\ref{sec:diagnosis}),
we below show that this connectivity is not by chance. 
%
% We are now ready to formulate a connectivity property of MICs.
%
\begin{theorem}[MIC Connectivity]\label{thm:connectivity}
Let $(V,E,\sigma)$ be an influence graph and
$\mu : V \rightarrow \{\plus,\minus\}$ a (partial) vertex labeling.

If $W\subseteq V$ is a MIC, then all vertices in~$W$
belong to the same strongly connected component in
$(\vertices{W},\edges{W})$.
\end{theorem}
The proof is omitted in view of space limitations and can be obtained from the authors.

\paragraph{Optimized MIC Encoding.}%\label{subsec:optenc}

We now apply     Theorem~\ref{thm:connectivity} to improve the basic
MIC extraction encoding (cf.\ Section~\ref{sec:diagnosis}) in two aspects:
adding (redundant) constraints for search space pruning and
adding positive body literals for reducing grounding efforts.
The following rules pave the way by determining the (non-trivial) strongly connected
components in $(V,\edges{V})$ as an over-approximation of
the ones in $(\vertices{W},\edges{W})$ for any $W\subseteq V$:
\begin{equation}\label{eq:reach}
  \begin{array}{@{}r@{{}\leftarrow{}}l}
  \sedge{U}{V} & \edge{U}{V},\naf{\iinput{V}}.\\
  \sedge{V}{U} & \edge{U}{V},\naf{\iinput{V}},\naf{\obsvlabel{U}{\plus}},\naf{\obsvlabel{U}{\minus}}.\hspace*{-10mm}\\[1mm]
  \reach{U}{V} & \sedge{U}{V}.\\
  \reach{U}{V} & \sedge{U}{W},\reach{W}{V},\vertex{V}.\\[1mm]
%  \cycle{V}{V} & \vertex{V},\naf{\iinput{V}}.\\
  \cycle{U}{V} & \reach{U}{V},\reach{V}{U}, U \neq V.
  \end{array}
\end{equation}
The first rule simply collects edges whose targets are not input,
while the second rule adds edges in the inverse direction for
unobserved regulators.
Reachability w.r.t.\ the so obtained graph is determined via
the third and the fourth rule.
Finally, predicate $\atomfont{cycle}$ indicates whether two (distinct)
vertices reach each other in $(V,\edges{V})$ relative to an
influence graph $(V,E,\sigma)$ and a (partial) vertex labeling~$\mu$.
In fact, if two vertices belong to a MIC~$W\subseteq V$,
then mutual reachability in $(\vertices{W},\edges{W})$ implies the
same in $(V,\edges{V})$, in view that $\vertices{W}\subseteq V$ and
$\edges{W}\subseteq\edges{V}$.
Conversely, if two vertices do not reach each other in $(V,\edges{V})$,
then they cannot jointly belong to any MIC.

The over-approximation of potential MICs provides an easy means to prune
the search space by adding the following integrity constraint:
\begin{equation}\label{eq:cycle}
  \begin{array}{@{}r@{{}\leftarrow{}}l}
  & \activ{U},\activ{V},U<V,\naf{\cycle{U}{V}}.
  \end{array}
\end{equation}
The constraint makes the fact explicit that distinct vertices of a MIC
must reach each other in $(V,\edges{V})$,
and it immediately refutes MIC candidates that do not satisfy this condition.

After making use of Theorem~\ref{thm:connectivity} to narrow search,
we now shift the focus to grounding.
As a matter of fact, the quadratic space complexity of the minimality test's
ground instantiation, as encoded in~(\ref{eq:minimal}),
is a major bottleneck in scaling.
The knowledge about potential pairwisely connected vertices in MICs,
represented by integrity constraint~(\ref{eq:cycle}),
also allows us to include positive body literals in order to restrict the scope of
minimality tests:

\begin{equation}\label{eq:minopt}
  \begin{array}{@{}r@{}c@{}l}
  \vlabell{W}{V}{\plus}; \vlabell{W}{V}{\minus} & {} \leftarrow {} &
    \activ{W}, \activ{V},   \cycle{V}{W}.%\hspace*{-10mm}
  \\
  \vlabell{W}{U}{\plus}; \vlabell{W}{U}{\minus} & {} \leftarrow {} &
    \activ{W}, \aedge{U}{V}, \cycle{V}{W}.\hspace*{-10mm}
  \\
  \elabell{W}{U}{V}{\plus}; \elabell{W}{U}{V}{\minus} & {} \leftarrow {} &
    \activ{W}, \aedge{U}{V}, \cycle{V}{W}.\hspace*{-10mm}
  \\[1mm]
  \multicolumn{3}{@{}l}{
  \begin{array}{@{}r@{}c@{}l}
  \vlabell{W}{V}{S} & {} \leftarrow {} &
    \activ{W}, \obsvlabel{V}{S}, \cycle{V}{W}.
  \\
  \vlabell{W}{U}{S} & {} \leftarrow {} &
    \activ{W}, \obsvlabel{U}{S}, \edge{U}{V}, \cycle{V}{W}.
  \\
  \elabell{W}{U}{V}{S} & {} \leftarrow {} &
    \activ{W}, \obselabel{U}{V}{S}, \cycle{V}{W}.
  \\[1mm]
  \influencel{W}{V}{\plus} & {} \leftarrow {} & 
    \elabell{W}{U}{V}{S}, \vlabell{W}{U}{S}.
  \\
  \influencel{W}{V}{\minus} & {} \leftarrow {} & 
    \elabell{W}{U}{V}{S}, \vlabell{W}{U}{T}, S\neq T.
  \\[1mm]
%                        & {} \leftarrow {} &
%     \vlabell{W}{V}{S}, \activ{V}, V\neq W, \naf{\influencel{W}{V}{S}}.\hspace*{-10mm}
  \multicolumn{3}{l}{
                    \qquad {} \leftarrow 
    \vlabell{W}{V}{S}, \activ{V}, \cycle{V}{W}, \naf{\influencel{W}{V}{S}}.
  }
  \end{array}
  }
  \end{array}  
\end{equation}
In comparison to~(\ref{eq:minimal}), the extra condition~\cycle{V}{W} in
the bodies of the first three rules establishes that labels used for
testing minimality are guessed only for pairs~$W$ and~$V$ of vertices
that can potentially jointly belong to a MIC.
The same restriction is used in the next three rules forwarding observed
vertex and edge labels, but now limited to vertices that can jointly belong to
a MIC and to their respective regulators.
Finally, the last two rules and the integrity constraint perform the same test
as in~(\ref{eq:minimal}) for a restricted set of pairs~$W$ and~$V$.
(The fact that~\cycle{V}{W} implies $V\neq W$ in \elabell{W}{U}{V}{S} also allows
 us to drop this condition, used in~(\ref{eq:minimal}), from the bodies of the rules
 defining \atomfont{receive'}.)

The complete optimized MIC encoding consists of the original rules
in~(\ref{eq:known}), (\ref{eq:active}), and~(\ref{eq:inconsistency}),
(\ref{eq:reach}) and~(\ref{eq:cycle}) as add-ons, and
(\ref{eq:minopt}) as a replacement for~(\ref{eq:minimal}).
As regards the computational impact,
we note that the optimized encoding needs less than two seconds for
grounding and finding all MICs % (regardless of the \claspD\ heuristic)
on the case study in Section~\ref{sec:benchmark:case:study},
which took more than a minute % 70 seconds (using ``VMTF'' heuristic)
with the unoptimized encoding.

A second version of the optimized encoding is obtained by tightening the consideration
of connected vertices in $(\vertices{W},\edges{W})$ relative to a MIC candidate~$W$.
This can be achieved by adding condition \activ{V} to the rules in~(\ref{eq:reach})
defining the $\atomfont{edges}$ predicate.
In this way, the static reachability information encoded in~(\ref{eq:reach}),
which is completely evaluated by grounder \gringo, is turned into a dynamic
relation computed % by \claspD\ 
during search.
As it turns out,
there is no significant performance difference between these two versions of the optimized
MIC extraction encoding on the case study in Section~\ref{sec:benchmark:case:study}.
Hence, more real examples are needed to reliably compare their 
grounding and solving efficiency.

%%% Local Variables: 
%%% mode: latex
%%% TeX-master: "paper"
%%% End: 

\section{Empirical Evaluation and Application}
\label{sec:benchmark}

For assessing the scalability of our approach,
we start by conceiving a parameterizable suite of artificial
yet biologically meaningful benchmarks.
After that, we present a typical application stemming from real biological data,
illustrating the exertion in practice.
All experiments were performed using input reduction as explained in
Section~\ref{sec:reduction}. % the previous section.

\subsection{Checking Consistency}
\label{sec:benchmark:consistency}

% \begin{figure}[t]
%   \centering
%   \pgfimage[width=0.49\linewidth]{curve.consistency.grounding}
% 
%   \pgfimage[width=0.49\linewidth]{curve.consistency.solving.cmodels}
%   \pgfimage[width=0.49\linewidth]{curve.consistency.solving.claspD}
%   \caption{Run-times for grounding with 
%     \gringo\ (upper diagram)
%     and solving with either 
%     \cmodels\ (left diagram) or
%     \claspD\ (right diagram).}
%   \label{fig:consistencybench}
% \end{figure}
%
We first evaluate our approach on randomly generated instances,
aiming at structures similar to those found in biological applications.
Instances are composed of an influence graph, a complete labeling of its edges,
and a partial labeling of its vertices. 
Our random generator takes three parameters: 
(i)   the number~$\alpha$ of vertices in the influence graph, 
(ii)  the average degree~$\beta$ of the graph, and 
(iii) the proportion~$\gamma$ of observed variations for vertices. 
To generate an instance, we compute a random graph with~$\alpha$ vertices
(the value of~$\alpha$ varying from $500$ to $4000$)
under the model by Erd\H{o}s-R{\'e}nyi~\shortcite{erdos}.
Each pair of vertices has equal probability to be connected via an edge,
whose label is chosen independently with probability~$0.5$ for both signs.
We fix the average degree~$\beta$ to~$2.5$, 
which is considered to be a typical value for biological networks \cite{barabasi}. 
% The labels of the generated edges are chosen independently with probability~$0.5$ for both signs.
Finally,
$\lfloor \gamma \alpha \rfloor$ vertices are chosen with uniform probability
and assigned a label with probability~$0.5$ for both signs. 
For each number~$\alpha$ of vertices,
we generated 50 instances using five different values for~$\gamma$, viz.,
$0.01$, $0.02$, $0.033$, $0.05$, and~$0.1$.
All instances are available at \cite{bioasptoolchain}.

\begin{table}[t]
   \begin{center}%\scriptsize\renewcommand{\arraystretch}{.95}
  \begin{tabular}{|r|r|r|r|r|r|r|}
    \cline{1-7}
    &\textit{claspD} 
    &\textit{claspD} 
    &\textit{claspD} 
    &\textit{cmodels} 
    &\textit{dlv} 
    &\textit{gnt} \\
    $\alpha$
    &\textit{Berkmin} 
    &\textit{VMTF} 
    &\textit{VSIDS} 
    & 
    & 
    & \\  
    \cline{1-7}
 500&0.14&0.11&0.11&0.16& 0.46& 0.71 \\
1000&0.41&0.25&0.25&0.35& 1.92& 3.34 \\
1500&0.79&0.38&0.38&0.53& 4.35& 7.50 \\
2000&1.33&0.51&0.51&0.71& 8.15&13.23 \\
2500&2.10&0.66&0.66&0.89&13.51&21.88 \\
3000&3.03&0.80&0.79&1.07&20.37&31.77 \\
3500&3.22&0.93&0.92&1.15&21.54&34.39 \\
4000&4.35&1.06&1.06&1.36&30.06&46.14 \\
    \cline{1-7}
  \end{tabular}
   \end{center}
  \caption{Run-times for % grounding with \gringo\ and solving 
           consistency checking with \claspD, \cmodels, \dlv, and \gnt.\label{tab:consistency_bench}}%
\end{table}%
%
%%% Local Variables: 
%%% mode: latex
%%% TeX-master: "paper"
%%% End: 
%
We used \gringo\ (2.0.0) \cite{potasscoManual}
for combining the generated instances and the encoding given in
Section~\ref{sec:checking} into equivalent ground programs.
For checking consistency by computing an answer set (if it exists),
we ran disjunctive ASP solvers 
\claspD\ (1.1) \cite{drgegrkakoossc08a}
with ``Berkmin'', ``VMTF'', and ``VSIDS'' heuristics, % \cite{ryan04a},
\cmodels\ (3.75) \cite{gilima06a}
using \sysfont{zchaff}, 
\dlv\ (BEN/Oct~11) \cite{dlv03a}, and
\gnt\ (2.1) \cite{janisesiyo06a}.
All runs were performed on a Linux machine equipped with 
an AMD Opteron 2 GHz processor and a memory limit of~2GB RAM.

Table~\ref{tab:consistency_bench}
shows average run-times in seconds over 50 instances per number~$\alpha$ of vertices,
including grounding times of \gringo\ and solving times.
% of \cmodels\ and \claspD, respectively.
% Since \dlv\ uses an internal grounder,
% \gringo\ time does not contribute to its overall run-time,
% while the other solvers process the output of \gringo.
We checked that grounding times of \gringo\ increase linearly with the 
number~$\alpha$ of vertices, and %we checked that
they do not vary significantly over~$\gamma$. 
For all solvers, run-times % of \cmodels\ and \claspD\
also increase linearly in~$\alpha$.\footnote{%
Longer run-times of \claspD\ with ``Berkmin'' in comparison to 
the other heuristics are due to a more expensive computation
of heuristic values in the absence of conflict information.
Furthermore, the time needed for performing ``Lookahead''
slows down \dlv\ as well as \gnt.} %, but in a more sophisticated way.
For fixed~$\alpha$ values,
we found two clusters of instances:
consistent ones   where total labelings were easy to compute, and
inconsistent ones where inconsistency was detected from preassigned labels.
This tells us that the influence graphs generated as described above
are usually (too) easy to label consistently, and inconsistency only occurs
if it is explicitly introduced via fixed labels.
However, such constellations are not unlikely in practice (cf.\ Section~\ref{sec:benchmark:case:study}),
and isolating MICs from them, as done in the next subsection, turned out to be hard for most solvers.
Finally, greater values for~$\gamma$ led to an increased
proportion of inconsistent instances, without making them much harder.
%
% Comparing \cmodels\ and \claspD, 
% we observe that \cmodels\ continuously exhibits smaller run-times than \claspD.
% On the one hand,
% this can be explained by the more costly preprocessing \cite{gekanesc08a}
% done in \claspD.
% On the other hand,
% the heuristics of \sysfont{zchaff}, used by \cmodels,
% turned out to be very successful for the considered problem.
% Indeed, the run-times of \claspD\ shown in Figure~\ref{fig:consistencybench}
% have been obtained by switching from \claspD's default heuristics to
% ``VSIDS,'' which is the heuristics of \sysfont{zchaff}.   
% Due to this, the run-times of \claspD\ could be improved by a factor
% of roughly 30, which is an interesting phenomenon that deserves further attention.

\subsection{Identifying Minimal Inconsistent Cores}
\label{sec:benchmark:diagnosis}

% \begin{figure}[t]
%   \centering
%   \pgfimage[width=0.49\linewidth]{curve.diagnosis.grounding}
% %   \pgfimage[width=0.49\linewidth]{curve.diagnosis.solving}
%   \pgfimage[width=0.49\linewidth]{curve.diagnosis.solving.claspD}
%   \caption{Run-times for grounding with 
%     \gringo\ (left)
%     and solving with 
%     \claspD\ (right).%
% }
%   \label{fig:diagnosisbench}
% \end{figure}
%
We now investigate the problem of finding a MIC within the
same setting as in the previous subsection.
Because of the elevated size of ground instantiations and problem difficulty,
we varied the number~$\alpha$ of vertices from~$50$ to~$300$, thus,
using considerably smaller influence graphs than before.
% than in the previous subsection.
We again use \gringo\ for grounding,
now taking the encoding given in Section~\ref{sec:diagnosis}.
As regards solving,
we restrict our attention to \claspD\ because all three of the other solvers
showed drastic performance declines.

\begin{table}[t]
  \begin{center}%\scriptsize\renewcommand{\arraystretch}{.95}
  \begin{tabular}{|r|r|r|r|r|}
    \cline{1-5}
%     &
    &\textit{gringo}
    &\textit{claspD} 
    &\textit{claspD} 
    &\textit{claspD} \\
    $\alpha$
%     &$\gamma$
    &
    &\textit{Berkmin} 
    &\textit{VMTF} 
    &\textit{VSIDS} \\
    \cline{1-5}
 50&0.24&  1.16 (0)&  0.65 (0)&  0.97 (0) \\
 75&0.55& 39.11 (1)&  1.65 (0)&  3.99 (0) \\
100&0.87& 41.98 (1)&  3.40 (0)&  4.80 (0) \\
125&1.37& 15.47 (0)& 47.56 (1)& 10.73 (0) \\
150&2.02& 54.13 (0)& 48.05 (0)& 15.89 (0) \\
175&2.77& 30.98 (0)&116.37 (2)& 23.07 (0) \\
200&3.82& 42.81 (0)& 52.28 (1)& 24.03 (0) \\
225&4.94& 99.64 (1)& 30.71 (0)& 41.17 (0) \\
250&5.98&194.29 (3)&228.42 (5)&110.90 (1) \\
275&7.62&178.28 (2)&193.03 (4)& 51.11 (0) \\
300&9.45&241.81 (2)&307.15 (7)&124.31 (0) \\
    \cline{1-5}
  \end{tabular}
  \end{center}
  \caption{Run-times for grounding with \gringo\ and solving with \claspD.\label{tab:diagnosis_bench}}%
\end{table}%
%%% Local Variables: 
%%% mode: latex
%%% TeX-master: "paper"
%%% End: 
%
Table~\ref{tab:diagnosis_bench} shows
average run-times in seconds over 50 instances per number~$\alpha$ of vertices.  
% for grounding with \gringo\ and solving with \claspD\
% using %its default heuristics 
% ``Berkmin'',``VMTF'', and ``VSIDS'' heuristics.
Timeouts, indicated in parentheses, are taken as maximum time of 1800 seconds.
%%% Sven?
% Note that, in contrast to the good performance in consistency checking,
% \cmodels\ showed drastic declines on this altered problem,
% so we decided to exclude it.
We observe a quadratic increase in grounding times of \gringo,
which is in line with the fact that ground instantiations for our MIC encoding
grow quadratically with the size of influence graphs.
In fact, the schematic rules in Section~\ref{subsec:minimal} give
rise to~$\alpha$ copies of an influence graph.
Considering solving times spent by \claspD\ for finding one MIC (if it exists),
we observe that they are relatively stable,
in the sense that they are tightly correlated to grounding times.
This regularity again confirms that, though it is random, 
the applied generation pattern tends to produce rather uniform influence graphs.
Finally, we observed that unsatisfiable instances, i.e., consistent instances without any MIC,
were easier to solve than the ones admitting answer sets.
We conjecture that this is because consistent total labelings provide a
disproof of inconsistency as encoded in Section \ref{subsec:muc:inconsistent}.
% As regards grounding, we observe a nearly perfect linear relationship
% between the number~$\alpha$ of vertices and \lparse\ times,
% where the slope of the line is $2$.
% That is, grounding is here of complexity $O(n^2)$
% in the size~$n$ of influence graphs, which is in line
% with the fact that our encoding for MICs grows quadratically in the graph size.
% In fact, the schematic rules in Section~\ref{subsec:minimal} give
% rise to~$\alpha$ copies of an influence graph.
%
% Considering solving times of \claspD,
% we also obtain a linear relationship, that is,
% for each~$\alpha$, the 50 instances are distinguished into two groups. 
% Most of them were easily solved in less than a minute. 
% However, some instances turned out to be strikingly more difficult, 
% and it took between 100 and 1000 times longer to solve them.
% However, we observe that the times needed to solve the
% easiest and the hardest, respectively, 
% instances grows linearly. % (in logarithmic scale). 
% This time, the slope of the line is slightly greater than $2$, that is,
% $2.3$ for the easiest and $3.2$ for the hardest instances. 
% Unfortunately, we could not characterize the hardest instances further.

As our experimental results demonstrate,
computing MICs % in an inconsistent instance
is computationally harder than just checking consistency. 
This is not surprising because 
% extracting MUCs from an unsatisfiable set of
the related (yet simpler) decision problem of verifying % or deciding containment in
a MUC is
% complete for the second level of the polynomial hierarchy
D$^{\text{P}\!}$-complete \cite{scalableMUC,papyan82a}
and thus more complex than just deciding satisfiability.
% or $\Sigma_2^P$-complete \cite{grmapi08a}, respectively.
With our declarative technique,
we spot 
the quadratic space blow-up incurred by the MIC encoding % in a fully declarative way
in Section~\ref{sec:diagnosis} 
as a bottleneck. % for very large graphs.
However, there are approaches aiming at a reduction of grounding efforts,
and some of them have been presented in Section~\ref{sec:refinements}.

% However, we note that there are also many easy instances,
% and it is an open question whether instances 
% coming from real data fall into this category.
%
%%% Local Variables: 
%%% mode: latex
%%% TeX-PDF-mode: t
%%% TeX-master: "paper"
%%% End: 

\subsection{Biological Case Study}
\label{sec:benchmark:case:study}

In the following, we present the results of applying our approach
to real-world data of genetic regulations in yeast.
We tested the gene-regulatory network of yeast provided in \cite{guelzim}
against genetic profile data of \emph{snf2} knock-outs \cite{snf2-ko}
from the Saccharomyces Genome Database\footnote{\url{http://www.yeastgenome.org}}.
The regulatory network of yeast contains 909 genetic or
biochemical regulations, all of which have been established experimentally,
among 491 genes.

\begin{figure}[t]
  \centering
  \begin{tikzpicture}[->,semithick,>=stealth']
    \scriptsize
    \tikzstyle{species}=[draw, ellipse, fill=none,text=black,minimum size=17pt]
    \tikzstyle{posspecies}=[style=species,text=black,text opacity=1,minimum size=17pt,fill=lightgray]
    \tikzstyle{negspecies}=[style=species,text=white,minimum size=17pt,fill=black ]
    \node[posspecies] (ume6) at (1.0, 2.0) {ume6};
    \node[posspecies] (hsf1) at (1.0, 1.0) {hsf1};
    \node[posspecies] (ume6_2) at (3.0,2.0) {ume6};
    \node[posspecies] (spo12) at (3.0,1.0) {spo12};
    \node[posspecies] (ume6_3) at (5.0,2.0) {ume6};
    \node[negspecies] (ino2) at (5.0,1.0) {ino2};
    \node[species] (reb1) at (7.0,2.0) {reb1};
    \node[negspecies] (hsc82) at (7.0,1.0) {hsc82};
    \node[posspecies] (ume6_4) at (9.0,2.0) {ume6};
    \node[posspecies] (top1) at (9.0,1.0) {top1};
\node (A) at (0.0,0.0) {}; %invisible node to have some space between the pictures
    \path
    (ume6) edge[-|] (hsf1)
    (ume6_2) edge[-|] (spo12)
    (ume6_3) edge (ino2)
    (reb1) edge (hsc82)
    (reb1) edge (top1)
    (ume6_4) edge[-|] (top1);
  \end{tikzpicture}

  \begin{tikzpicture}[->,semithick,>=stealth']
    \scriptsize
    \tikzstyle{species}=[draw, ellipse, fill=none,text=black,minimum size=17pt]
    \tikzstyle{posspecies}=[style=species,text=black,text opacity=1,minimum size=17pt,fill=lightgray]
    \tikzstyle{negspecies}=[style=species,text=white,text opacity=1,minimum size=17pt,fill=black ]
    \node[species] (reb1) at (3.0,2.0) {reb1};
    \node[negspecies] (rap1) at (3.0,1.0) {rap1};
    \node[posspecies] (ume6_4) at (5.0,2.0) {ume6};
    \node[posspecies] (top1) at (5.0,1.0) {top1};
    \node[species] (reb1_2) at (8.0,4.0) {reb1};
    \node[species] (sin3) at (7.0,3.0) {sin3};
    \node[posspecies] (ume6_5) at (7.0,2.0) {ume6};
    \node[posspecies] (top1_2) at (8.0,1.0) {top1};
    \path
    (reb1)   edge (rap1)
    (reb1)   edge (top1)
    (ume6_4) edge[-|] (top1)
    (reb1_2) edge (sin3)
    (sin3)   edge[-|] (ume6_5) 
    (ume6_5) edge[-|] (top1_2)
    (reb1_2) edge[bend left=50,looseness=0.8] (top1_2);
  \end{tikzpicture}
%   \pgfimage[width=2cm]{benchmarks/yeast_demo/snf2.converted.obs/mic_1}
%   \pgfimage[width=2cm]{benchmarks/yeast_demo/snf2.converted.obs/mic_2}
%   \pgfimage[width=2cm]{benchmarks/yeast_demo/snf2.converted.obs/mic_3}
%   \pgfimage[width=4cm]{benchmarks/yeast_demo/snf2.converted.obs/mic_14}
%   \pgfimage[width=4cm]{benchmarks/yeast_demo/snf2.converted.obs/mic_17}
%   \pgfimage[width=2.6cm]{benchmarks/yeast_demo/snf2.converted.obs/mic_18}

  \caption{Some MICs obtained by comparing the regulatory
    network of yeast % in \cite{guelzim} 
    with a genetic profile.} % from \cite{snf2-ko}.}
  \label{fig:snf2-mics}
\end{figure}
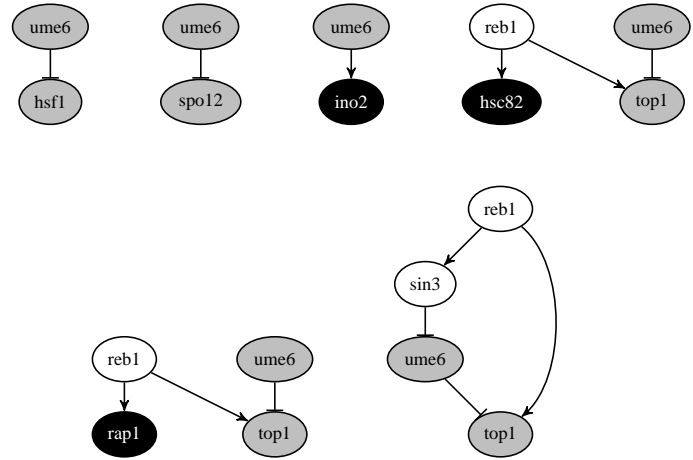

Comparing the yeast regulatory network with the genetic profile of \emph{snf2},
we found the data to be inconsistent with the network,
which was easily detected using the approach of Section~\ref{sec:checking}.
Applying our diagnosis technique from Section~\ref{sec:diagnosis}, we obtained a total of~19 MICs.
While computing the first MIC took less than a second using \gringo\ and \claspD\
(regardless of the heuristic used),
the computation of all MICs was considerably harder.
Using ``VMTF'' as search heuristic on top of the enumeration algorithm \cite{gekanesc07c}
inherited from \clasp\ \cite{gekanesc07b},
\claspD\ had found all 19 MICs in about 30 seconds, while another 40 seconds were needed
to decide that there is no further MIC.
With ``VSIDS'', finding the 19 MICs took about the same time as with ``VMTF'',
but another 80 seconds were   used to verify that all MICs had been found.
Finally, using ``Berkmin'' heuristic,
12 MICs had been found before aborting after 30 minutes.
The observation that search heuristics matter tells us that investigations into
the structure of biological problems and particular methods to solve them efficiently
can earn considerable benefits.%
\footnote{%
Notably, by exploiting additional background knowledge,
the optimized encoding presented in Section~\ref{sec:scc}
requires less than two seconds (regardless of heuristics) for grounding and finding all 19 MICs.
In fact, its ground instantiation contains only 8481 atoms and 10843 rules,
compared to 47260 atoms and 56522 rules with the basic encoding in Section~\ref{sec:diagnosis}.
In addition to problem size,
also the difficulty drops dramatically: from 23345 conflicts down to 270 conflicts,
encountered with ``VMTF'' heuristic during search for all answer sets.%
}
Furthermore, we note that the potential existence of multiple answer sets encompassing the same MIC
did not emerge on the yeast network and \emph{snf2} knock-out data.
That is, we obtained 19 answer sets, each one corresponding one-to-one to a MIC.

Six of the computed MICs are exemplarily shown in Figure~\ref{fig:snf2-mics}.
While the first three of them are pretty obvious,
we also identified more complex topologies.
However,
our example demonstrates that the MICs obtained in practice 
are still small enough to be understood easily.
For finding suitable corrections to the inconsistencies, %i.e., unreliable data or missing reactions,
it is often even more helpful to display the connections between several overlapping MICs.
Observe that all six MICs in Figure~\ref{fig:snf2-mics} are related to gene
\emph{ume6}.
Connecting them yields the subgraph of the yeast regulatory network in Figure~\ref{fig:snf2-mics-connected}.

\begin{figure}[t]
  \centering
%% colored picture
%   \begin{tikzpicture}[->,semithick,>=stealth']
%     \scriptsize
%     \tikzstyle{species}=[draw, ellipse, fill=none,text=black,minimum size=17pt]
%     \tikzstyle{posspecies}=[style=species,text=black,text opacity=1,minimum size=17pt,fill=green, fill opacity=0.5,]
%     \tikzstyle{negspecies}=[style=species,text=black,text opacity=1,minimum size=17pt,fill=red, fill opacity=0.8 ]
%     \node[species]    (reb1) at (4.0,4.0) {reb1};
%     \node[negspecies] (hsc82) at (1.0,3.0) {hsc82};
%     \node[negspecies] (rap1) at (3.0,3.0) {rap1};
%     \node[species] (sin3) at (5.0,3.0) {sin3};
%     \node[posspecies] (ume6) at (5.0,2.0) {ume6};
%     \node[negspecies] (ino2) at (1.0,1.0) {ino2};
%     \node[posspecies] (hsf1) at (3.0, 1.0) {hsf1};
%     \node[posspecies] (spo12) at (5.0,1.0) {spo12};
%     \node[posspecies] (top1) at (7.0,1.0) {top1};
%     \path
%     (reb1) edge[green] (hsc82)
%     (reb1) edge[green] (rap1)
%     (reb1) edge[green] (sin3)
%     (reb1) edge[green,bend left=50,looseness=0.9] (top1)
% %
%     (sin3) edge[red,-|] (ume6)
% %
%     (ume6) edge[green] (ino2)
%     (ume6) edge[red,-|] (hsf1)
%     (ume6) edge[red,-|] (spo12)
%     (ume6) edge[red,-|] (top1);
%   \end{tikzpicture}
%% black white picture
  \begin{tikzpicture}[->,semithick,>=stealth']
    \scriptsize
    \tikzstyle{species}=[draw, ellipse, fill=none,text=black,minimum size=17pt]
    \tikzstyle{posspecies}=[style=species,text=black,text opacity=1,minimum size=17pt,fill=lightgray]
    \tikzstyle{negspecies}=[style=species,text=white,text opacity=1,minimum size=17pt,fill=black ]
    \node[species]    (reb1) at (4.0,4.0) {reb1};
    \node[negspecies] (hsc82) at (1.0,3.0) {hsc82};
    \node[negspecies] (rap1) at (3.0,3.0) {rap1};
    \node[species] (sin3) at (5.0,3.0) {sin3};
    \node[posspecies] (ume6) at (5.0,2.0) {ume6};
    \node[negspecies] (ino2) at (1.0,1.0) {ino2};
    \node[posspecies] (hsf1) at (3.0, 1.0) {hsf1};
    \node[posspecies] (spo12) at (5.0,1.0) {spo12};
    \node[posspecies] (top1) at (7.0,1.0) {top1};
    \path
    (reb1) edge (hsc82)
    (reb1) edge (rap1)
    (reb1) edge (sin3)
    (reb1) edge[bend left=50,looseness=0.9] (top1)
    (sin3) edge[-|] (ume6)
    (ume6) edge (ino2)
    (ume6) edge[-|] (hsf1)
    (ume6) edge[-|] (spo12)
    (ume6) edge[-|] (top1);
  \end{tikzpicture}

%   \pgfimage[width=7cm]{benchmarks/yeast_demo/snf2.converted.obs/mic-subset}
  
  \caption{Subgraph obtained by connecting the six MICs given in Figure~\ref{fig:snf2-mics}.}
  \label{fig:snf2-mics-connected}
\end{figure}
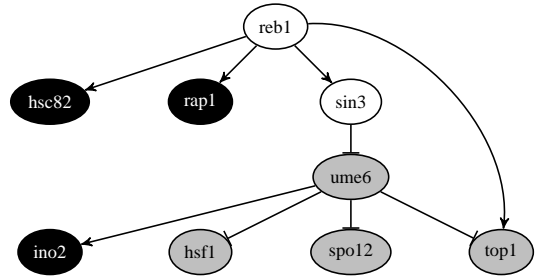
The most obvious problem in Figure~\ref{fig:snf2-mics-connected} is
that the observed increase of \emph{ume6} is incompatible with 
   its four targets. This suggests that either the observation on
\emph{ume6} is incorrect      or that some regulations are missing or
wrongly modeled. In the first hypothesis though, one should note that
the current model cannot explain a decrease of \emph{ume6}: this would
imply an increase of \emph{sin3} and in turn an increase
of \emph{reb1}, but then there would be no explanation left for the
variation of \emph{hsc82} and \emph{rap1}. So, in either case, our model
should be revised. This is not a great surprise: our literature-based
network, although very reliable, was presumably far from being complete.

Regarding the biological background, note that
\emph{ume6} is a known regulator of sporulation in yeast: in case
of nutritional stress, yeast cells stop dividing and produce spores
by meiosis.  These spores are reproductive structures better adapted
to extreme conditions. \emph{ume6}
is known as a key inhibitor of early meiotic genes: upon entry in
meiosis, this inhibitory effect is released and the target genes are
expressed. Notably, a knock-out of \emph{ume6} causes the
expression of meiotic genes during vegetative growth (hence its name,
\emph{Unscheduled Meiotic Expression}) as well as almost complete
failure of sporulation \cite{pmid11238941}. \emph{ume6} seems to have
activation capabilities as well, though in that case the effect is
believed to be indirect~\cite{chen2007}. 

In the current view, \emph{ume6} switches from inhibitor to (indirect)
activator at the beginning of meiosis: Ume6p (the protein
corresponding to the gene \emph{ume6}) has a repressive effect when it
forms a complex with Sin3p (note that \emph{sin3} is in our network)
and Rdp3p, which is degraded upon entry in
meiosis~\cite{mallory2007}. This molecular mechanism can be
interpreted in our model and one possible result is given in
Figure~\ref{fig:snf2-mics-corrected}.  At least for negative targets,
we now have a plausible explanation: the real effector of the
inhibition on \emph{hsf1}, \emph{spo12}, \emph{top1}, and \emph{ume6}
itself is the complex Ume6p-Sin3p, whose variation is unobserved but
depends on the variation of \emph{ume6} and \emph{sin3}.  The
variation of the targets can be explained if the protein complex
decreases, which is in turn possible if \emph{sin3}
decreases. Regretfully \emph{sin3} is not observed in our data, but we
note that a decrease of this gene is fully compatible with the rest
of the network, that is, if we suppose a decrease of \emph{reb1}.
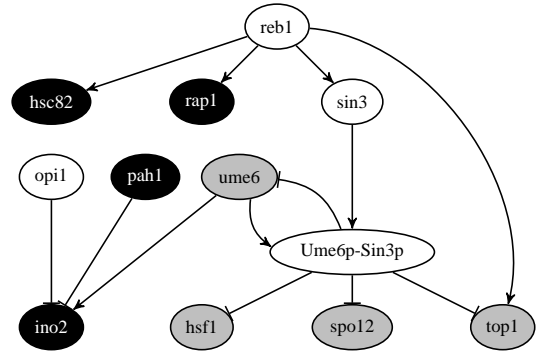
\begin{figure}[t]
  \centering
  \begin{tikzpicture}[->,semithick,>=stealth']
    \scriptsize
    \tikzstyle{species}=[draw, ellipse, fill=none,text=black,minimum size=17pt]
    \tikzstyle{posspecies}=[style=species,text=black,text opacity=1,minimum size=17pt,fill=lightgray]
    \tikzstyle{negspecies}=[style=species,text=white,text opacity=1,minimum size=17pt,fill=black ]
    \node[species]    (reb1) at (4.0,4.0) {reb1};
    \node[negspecies] (hsc82) at (1.0,3.0) {hsc82};
    \node[negspecies] (rap1) at (3.0,3.0) {rap1};
    \node[species] (sin3) at (5.0,3.0) {sin3};
    \node[posspecies] (ume6) at (3.5,2.0) {ume6};
    \node[species]    (ume6sin3) at (5.0,1.0) {Ume6p-Sin3p};
    \node[negspecies] (ino2) at (1.0,0.0) {ino2};
    \node[posspecies] (hsf1) at (3.0, 0.0) {hsf1};
    \node[posspecies] (spo12) at (5.0,0.0) {spo12};
    \node[posspecies] (top1) at (7.0,0.0) {top1};

    \node[species] (opi1) at (1.0,2.0) {opi1};
    \node[negspecies] (pah1) at (2.25,2.0) {pah1};
    \path
    (opi1) edge[-|] (ino2)
    (pah1) edge[-|] (ino2);

    \path
    (reb1) edge (hsc82)
    (reb1) edge (rap1)
    (reb1) edge (sin3)
    (reb1) edge[bend left=50,looseness=0.9] (top1)
    (ume6sin3) edge[bend right=30,looseness=1.1,-|] (ume6)
    (ume6) edge[bend right=30,looseness=1.1] (ume6sin3)
    (ume6) edge (ino2)
    (sin3) edge (ume6sin3)

    (ume6sin3) edge[-|] (hsf1)
    (ume6sin3) edge[-|] (spo12)
    (ume6sin3) edge[-|] (top1);
  \end{tikzpicture}

%   \pgfimage[width=7cm]{benchmarks/yeast_demo/snf2.converted.obs/mic-subset}
  
  \caption{Local correction of the network based on our diagnosis
    method and literature research.}
  \label{fig:snf2-mics-corrected}
\end{figure}
Now concerning \emph{ino2}, our network should be updated with more
recent evidence: as reviewed in~\cite{chen2007}, \emph{ino2} has
several additional regulators, such as \emph{opi1} and
\emph{pah1} (see Figure~\ref{fig:snf2-mics-corrected}). The observed
variation of \emph{pah1} is not useful to explain that of \emph{ino2},
but \emph{opi1} is definitely a plausible candidate.

Here we illustrated one main usage of our diagnosis technique:
identifying poorly modeled regions of a regulatory network that are
incompatible with a given data set. This is definitely a key asset  if
one wants to build a large-scale regulatory database  and check its
coherence with newly produced data on a regular basis. Given new
data,    our diagnosis method produces human-understandable representations
of possible incompatibilities with the current model, which serve as
the basis for a targeted literature research. With this data-driven
approach, a network can then be improved with considerably less effort
than with a random traversal of publications, for a much more
coherent result.

%%% Local Variables: 
%%% mode: latex
%%% TeX-PDF-mode: t 
%%% TeX-master: "paper"
%%% End: 

\section{Web Service}\label{sec:web}

To make our methods easily accessible to a biological audience,
we built a web service\footnote{\url{http://data.haiti.cs.uni-potsdam.de/wsgi/app}}
not requiring any locally installed software on the user side except for a web browser.
% Figure~\ref{screenshot} shows the interface of the web service.
It provides the possibility to upload textual representations of 
biological networks as well as experimental profiles.
Also, a number of predefined examples
allows a user to instantly experience the functionalities of the web service.
These include consistency checking and diagnosis, i.e., finding MICs,
whose implementation has been detailed in Section~\ref{sec:checking} and~\ref{sec:diagnosis}.

% \begin{figure}[t]
%  \centering
% % \includegraphics[width=412px,height=316px,bb=0 0 618 474]{screenshot.png}
% \ifpdf
%  \includegraphics[width=309px,height=237px]{screenshot.png}
% \else
%  \includegraphics[width=309px,height=237px]{screenshot.eps}
% \fi
%  % screenshot.png: 824x632 pixel, 96dpi, 21.80x16.72 cm, bb=0 0 618 474
% \caption{Web interface for consistency checking and diagnosis.\label{screenshot}}%
% \end{figure}

Influence graphs representing biological networks usually contain vertices that
are not subject to any regulation.
Such entities are understood as controlled by external factors,
like environmental or particular experimental conditions.
To avoid trivial inconsistencies due to such unregulated and thus unexplainable vertices,
the web interface provides an option ``Guess input nodes'' for
automatically declaring all vertices without any predecessor as inputs.
While consistency checking simply results in a positive or negative answer,
we offer three diagnosis modes: 
``find one inconsistency'',
 ``find all inconsistencies'', and
``approximate all inconsistencies''.
The first mode aims at finding a single MIC, and the second at finding all of them.
For the latter, we currently use an encapsulating script that repeatedly calls \claspD\
while feeding already identified MICs back as integrity constraints, until no
further answer set exists.
This makes sure that each answer set corresponds to a new MIC and thus
avoids potential repetitions.
The problem of enumerating answer sets that differ on a set of ``relevant'' atoms
(in our case, on instances of predicate $\atomfont{active}$)
is addressed in \cite{gekasc09a}.
The integration of this technique into \claspD, in order to make the wrapper script obsolete,
is subject to future work.
Once MICs have been computed, they can be represented either textually or graphically,
as shown in Figure~\ref{diagnosis:results}.
If the result consists of several MICs, it is possible to view overlapping ones in a combined way,
thus highlighting regions of inconsistency.
Finally, the third diagnosis mode, ``approximate all inconsistencies'', works by marking the vertices of
a computed MIC as inputs before proceeding to look for further MICs.
This approach has been used in previous work \cite{gubomosi09a} and
has been integrated into our framework for comparison.
However, the results obtained with the third mode depend on the order in which MICs are found
and their vertices declared to be inputs in future computations.
Further functionalities, like prediction under consistency \cite{coliRIAMS}
and inconsistency \cite{geguivscsithve09a}, are also
featured by the web service but are outside the scope of this paper.

% the second mode tries to compute all minimal inconsistencies.
% To find all minimal inconsistencies the service computes the first mic, 
%  then adds a constraint to the logic program that forbids the same mic as solution and 
%  computes the next mic until all have been found.
% The third mode works similar, it finds the first minimal inconsistency, 
%  but then marks the nodes as input in the graph and computes then the next minimal inconsistency.
% This mode finds asubset of all minimal inconsistencies,
%  such that the minimal inconsistencies do not overlap.
% The idea behind mode three is to provide a method that is faster than finding all inconsistencies but
%  still gives a good approximation of the inconsistent parts of the graph.
% The results of the third mode depend on the order in which the minimal inconsistencies are found,
%  because for minimal inconsistencies that have a common node only the first computed is included.

% The service generates graphical and textual representation of the computed inconsistencies see Figure~\ref{diagnosis:results}.
% If more than one inconsistency is computed it is also possible to join those,
%  for inconsistencies with common nodes then a connected graph is computed.

\begin{figure}[t]
 \centering
\ifpdf
 \raisebox{66px}{\includegraphics[width=180px,height=138px]{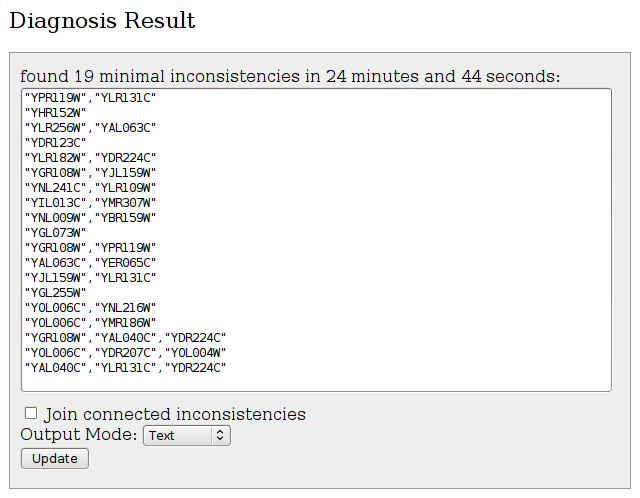}}%
 % screenshot2.png: 640x498 pixel, 96dpi, 16.93x13.17 cm, bb=0 0 480 373
% \includegraphics[width=320px,height=364px,bb=0 0 481 546]{screenshot3.png}
 \includegraphics[width=180px,height=204px]{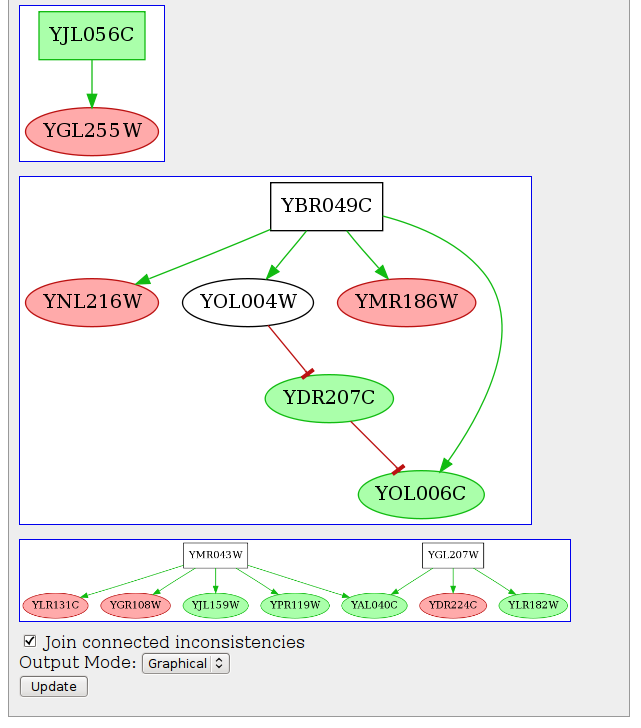}%
 % screenshot3.png: 641x728 pixel, 96dpi, 16.96x19.26 cm, bb=0 0 481 546 
\else
 \raisebox{66px}{\includegraphics[width=180px,height=138px]{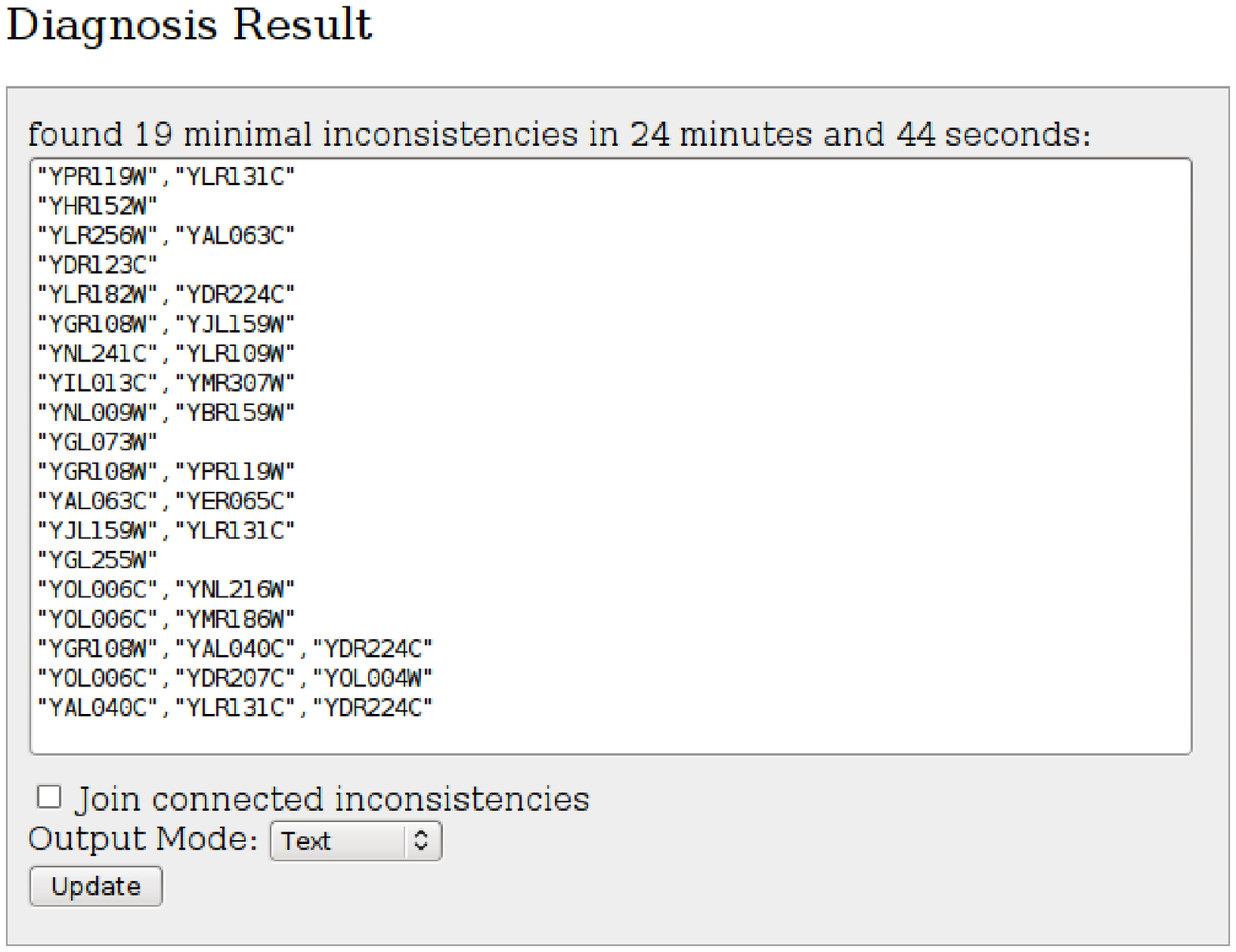}}%
 \includegraphics[width=180px,height=204px]{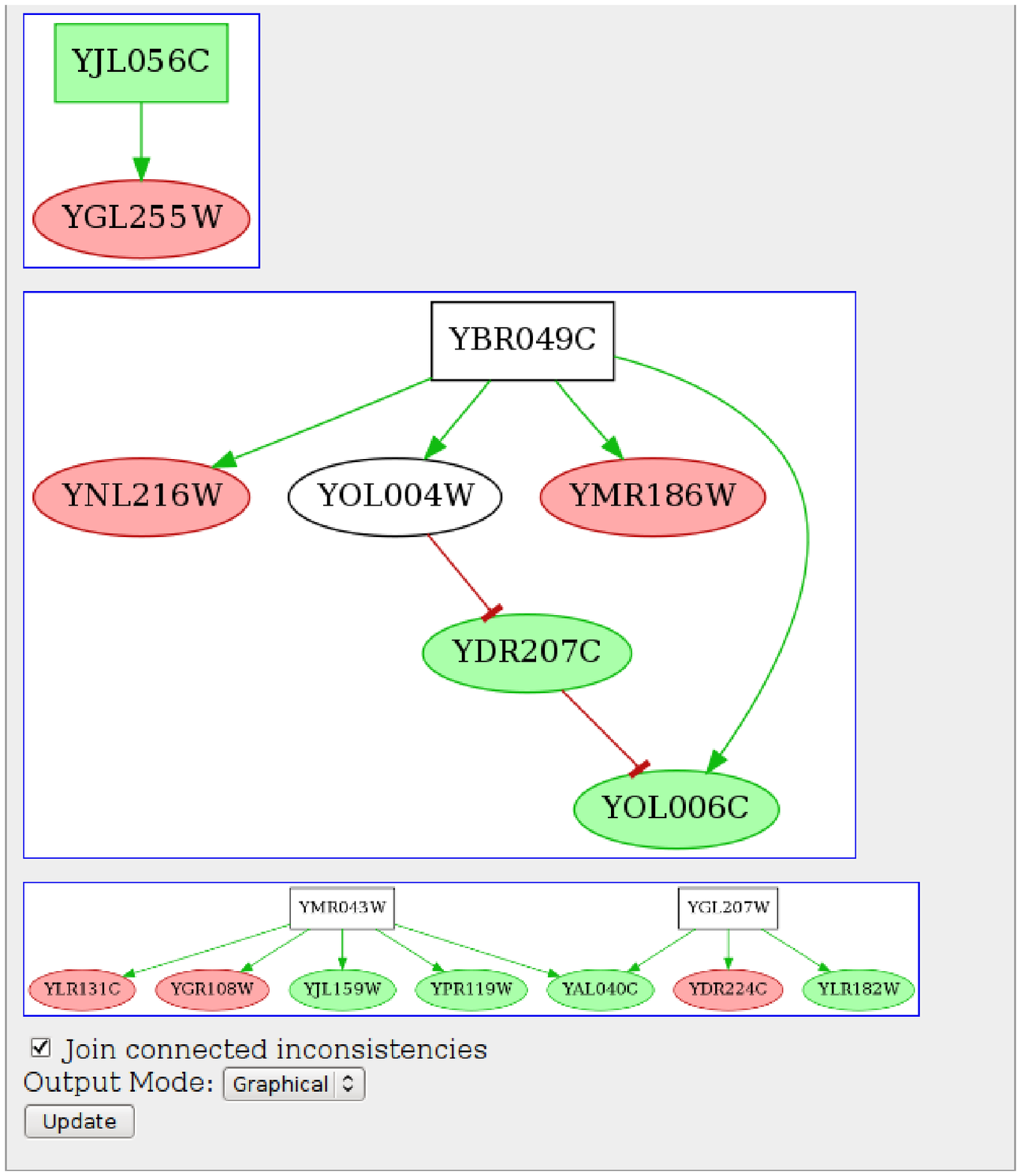}%
\fi
\caption{Representation of identified MICs in textual (left) and graphical (right) mode.\label{diagnosis:results}}
\end{figure}

% The web service is built up on a python library that provides these functionalities.
% The python library encapsulates methods for the transformation of network and observations
%  into logic programs, as well as the calls to the answer set solver and
% the reprocessing of the answer sets into a graphical representation.

%%% Local Variables: 
%%% mode: latex
%%% TeX-PDF-mode: t 
%%% TeX-master: "paper"
%%% End: 

\section{Discussion}\label{sec:discussion}

We have provided an approach based on ASP to investigate the 
consistency between experimental profiles and influence graphs.
In case of inconsistency,
the concept of a MIC can be exploited for identifying concise explanations,
pointing to unreliable data and/or missing reactions.
% Such MICs can also be determined by means of ASP,
% and we have provided an encoding for this purpose.
The problem of finding MICs is closely related to the extraction of
MUCs in the context of SAT.
From a knowledge representation point of view,
however, we argue for our ASP-based technique,
as it provides an easy way to model a problem
in terms of a uniform encoding and specific instances.

The BioQuali system \cite{gubomosi09a} provides functionalities
parallel to our approach.
It also works on influence graphs and applies the same consistency notion.
In preprocessing,
BioQuali reduces an influence graph by iteratively marking unobserved vertices
that have no successors as uncritical.
This technique is also realized by input reduction,
described in Section~\ref{sec:reduction}.
After that,
BioQuali transforms
the reduced subgraph into a Binary Decision Diagram,
used for further computations.
While consistency checking with BioQuali yields the same results as our technique,
its diagnosis functionality works like the
``approximate all inconsistencies'' mode, described in the previous section.
In contrast to our method, this does in general not admit finding all MICs.

By now, a variety of efficient ASP tools are available,
both for grounding and for solving logic programs.
Our empirical assessment of them (on random as well as real data)
has in principle demonstrated the scalability of the approach.
The web service implementation of finding all MICs,
which is genuine to our method and not available in any other existing tool,
is still based on some workarounds for avoiding redundant answer sets.
It is a subject of future work to address this with answer set projection~\cite{gekasc09a}.

As elegance and flexibility in modeling are major advantages of ASP,
our current application makes it attractive also for related biological questions,
beyond the ones addressed in this paper.
For instance, % natural extensions of the presented techniques
ongoing work deals with repair and prediction under consistency as well as inconsistency
\cite{geguivscsithve09a}.
% allow for accomplishing prediction and repair.
In future,
it will also be interesting to explore how far the performance of ASP tools
can be tuned by varying and optimizing encodings for particular tasks.
In turn, challenging applications like the one presented here might contribute to the further
improvement of ASP tools, as they might be geared towards efficiency in such domains.

% %\comment{insist on the fact that in this application we want to find all vertices that are implied in a MIC}
% The main contribution of this paper is to generalize the notion of minimum unsatisfiable core introduced in the context of SAT solving. If a set of constrains cannot be satisfied, then the goal is to find a subset of the constraints that cannot be satisfied either, and that is minimal with respect to inclusion. Some efficient solutions were proposed for SAT solving. Theses proposals are nevertheless specific to this context because they rely on the resolution tree. Here we propose a more general solution based on (disjunctive) answer set programming, which takes advantage of the available high level language constructs.

% Back to our application, namely to check consistency in large regulatory networks, it is of great interest to identify measurements that are involved in at least one MIC. We suggest that this could be a powerful criterion to identify unreliable data or missing regulations.

%%% Local Variables: 
%%% mode: latex
%%% TeX-PDF-mode: t 
%%% TeX-master: "paper"
%%% End: 

\appendix
\section{Proof of Theorem~\ref{thm:cons-sound} and~\ref{thm:cons-compl}}
\label{app:proof:consistency}

%To begin with, 
We formalize the representation of instances,
 as described in Section~\ref{subsec:instance},
 by defining a mapping~$\tau$
 of an influence graph~$(V,E,\sigma)$
%  (where
%  $V$ is the set of vertices,
%  $E$ the set of edges, and 
%  $\sigma : E \rightarrow \{\plus,\minus\}$ a (partial) labeling of the edges)
 and a (partial) vertex labeling~$\mu: V \rightarrow \{\plus,\minus\}$:
\begin{equation}\label{proof:eq:instance}
\begin{array}[b]{l@{}c@{}l@{}l}
  \tau((V,E,\sigma),\mu)
  & {} = {} &
  \{\vertex{i}. & {} \mid i\in V\} 
  \\ & \cup {} &
  \{\edge{j}{i}. & {} \mid (j {\,\rightarrow\,} i)\in E\} 
  \\ & \cup {} &
  \{\obselabel{j}{i}{s}. & {} \mid (j {\,\rightarrow\,} i)\in E, \sigma(j,i)=s\}
  \\ & \cup {} &
  \{\obsvlabel{i}{s}. & {} \mid i\in V, \mu(i)=s\}
  \\ & \cup {} &
  \{\iinput{i}. & {} \mid i\in V \text{ is an input}\}\ \text{.}
\end{array}
\end{equation}
%
%Furthermore,
By~$P_C$,
we denote the %consistency check 
encoding containing the schematic
rules in~(\ref{eq:total}),
(\ref{eq:known}), (\ref{eq:influence}), and~(\ref{eq:inconsistent}).

% \begin{equation}\label{proof:eq:total}
%   \begin{array}{r@{{}\leftarrow{}}l}
%   \vlabel{V}{\plus}   ;\vlabel{V}{\minus}    & \vertex{V}.    \\ 
%   \elabel{U}{V}{\plus};\elabel{U}{V}{\minus} & \edge{U}{V}.
%   \end{array}  
% \end{equation}
% \begin{equation}\label{proof:eq:known}
%   \begin{array}{r@{{}\leftarrow{}}l}
%   \vlabel{V}{S}    & \obsvlabel{V}{S}.    \\ 
%   \elabel{U}{V}{S} & \obselabel{U}{V}{S}.
%   \end{array}  
% \end{equation}
% \begin{equation}\label{proof:eq:influence}
%   \begin{array}{r@{{}\leftarrow{}}l}
%   \influence{V}{\plus}  & \elabel{U}{V}{S}, \vlabel{U}{S}. \\
%   \influence{V}{\minus} & \elabel{U}{V}{S}, \vlabel{U}{T}, S\neq T.
%   \end{array}  
% \end{equation}
% \begin{equation}\label{proof:eq:inconsistent}
%   \begin{array}{@{{}\leftarrow{}}l}
%   \vlabel{V}{S}, \naf{\influence{V}{S}}, \naf{\iinput{V}}.
%   \end{array}  
% \end{equation}
%
%
% In what follows, we recall and prove the soundness and completeness theorems
% formulated in Section~\ref{sec:checking:theorems}.
% %
% \begin{theorem}[Theorem~\ref{thm:cons-sound}]
% Let $(V,E,\sigma)$ be an influence graph and
% $\mu : V \rightarrow \{\plus,\minus\}$ a (partial) vertex labeling.

% If there is an answer set of $P_C\cup\tau((V,E,\sigma),\mu)$,
% then $(V,E,\sigma)$ and~$\mu$ are consistent.
% \end{theorem}

\begin{proof}[Proof of Theorem~\ref{thm:cons-sound}]
Assume that~$X$ is an answer set of $P_C\cup\tau((V,E,\sigma),\mu)$.
Furthermore, let
\begin{equation*}
\begin{array}{r@{}c@{}l}
P^X
& {} = {} &
\{
  (\head{r}\leftarrow\poslits{\body{r}})\theta
  \mid 
{} \\
& & \qquad 
  r\in P_C\cup\tau((V,E,\sigma),\mu),
  (\neglits{\body{r}}\theta)\cap X=\emptyset,
  \theta:\mathit{var}(r)\rightarrow\mathcal{U}
\}
\end{array}
\end{equation*}
where $\mathit{var}(r)$ is the set of all variables that occur in a rule~$r$,
$\mathcal{U}$ is the set of all constants appearing in $P_C\cup\tau((V,E,\sigma),\mu)$,
and~$\theta$ is a ground substitution for the variables in~$r$.
Then, by the definition of an answer set,
we know that~$X$ is a $\subseteq$-minimal model of~$P^X$.

Given~$X$, we define~$\sigma'$ and~$\mu'$ as follows:
\begin{equation*}
\begin{array}{r@{}r@{}r@{}c@{}l}
\sigma'
& {} = {} & \{
  (j {\,\rightarrow\,} i)\mapsto s 
& {} \mid {} & (j {\,\rightarrow\,} i)\in E, \elabel{j}{i}{s}\in X\}
% \\
% & {} \cup {} & \{ 
%   (j {\,\rightarrow\,} i)\mapsto \plus 
% & {} \mid {} & (j {\,\rightarrow\,} i)\in E, \elabell{k}{j}{i}{\plus}\notin X, \elabell{k}{j}{i}{\minus}\notin X\}
\\[1mm]
\mu'
& {} = {} & \{
  i\mapsto s
& {} \mid {} & i\in V, \vlabel{i}{s}\in X\}\ \text{.}
% \\
% & {} \cup {} & \{
%   i\mapsto \plus
% & {} \mid {} & i\in V, \vlabell{k}{i}{\plus}\notin X, \vlabel{k}{i}{\minus}\notin X\}\ \text{.}
\end{array}
\end{equation*}
We show that~$\sigma'$ and~$\mu'$ are total labelings of
edges and vertices, respectively,
such that $\mu'(i)=\mu'(j)\sigma'(j,i)$ holds
for every non-input vertex $i\in V$ and some edge $j {\,\rightarrow\,} i$ in~$E$.

Regarding the uniqueness of labels assigned by~$\sigma'$ and~$\mu'$,
 consider the following rules from~(\ref{eq:total}) and~(\ref{eq:known})
 including predicates \atomfont{labelE} and \atomfont{labelV} in their heads:
\begin{equation}\label{eq:label:def}
  \begin{array}{r@{}c@{}l}
  \vlabel{V}{\plus}; \vlabel{V}{\minus} & {} \leftarrow {} &
    \vertex{V}.
  \\
  \elabel{U}{V}{\plus}; \elabel{U}{V}{\minus} & {} \leftarrow {} &
    \edge{U}{V}.
  \\[1mm]
  \vlabel{V}{S} & {} \leftarrow {} &
    \obsvlabel{V}{S}.
  \\
  \elabel{U}{V}{S} & {} \leftarrow {} &
    \obselabel{U}{V}{S}.
  \end{array}  
\end{equation}
Since the given (partial) labelings~$\sigma$ and~$\mu$ assign unique labels
 to the elements of their domains,
 facts defining \atomfont{observedE} and \atomfont{observedV}
 are of the form 
 $\obselabel{j}{i}{\plus}.$ or $\obselabel{j}{i}{\minus}.$ and
 $\obsvlabel{i}{\plus}.$ or $\obsvlabel{i}{\minus}.$, respectively,
 and at most one of these facts is contained in $\tau((V,E,\sigma),\mu)$
 for an edge $(j {\,\rightarrow\,} i)\in E$ or a vertex $i\in V$.
Because~$X$ is a $\subseteq$-minimal model of~$P^X$, the atoms in the heads of facts
are in~$X$, and all atoms in~$X$ over predicates \atomfont{observedE} and 
\atomfont{observedV} are derived from facts in $\tau((V,E,\sigma),\mu)$,
 in view that these predicates do not occur in the head of any rule in~$P_C$.
Hence, at most one of the atoms
$\elabel{j}{i}{\plus}$ and $\elabel{j}{i}{\minus}$ or
$\vlabel{i}{\plus}$ and $\vlabel{i}{\minus}$, respectively,
 is derivable for an edge $(j {\,\rightarrow\,} i)\in E$
 or vertex $i\in\nolinebreak V$
from a ground instance of the fourth or third rule in~(\ref{eq:label:def})
 and then included in~$X$.
Furthermore, the second and first rule in~(\ref{eq:label:def})
 impose that at least one of
 $\elabel{j}{i}{\plus}$ or $\elabel{j}{i}{\minus}$ and
 $\vlabel{i}{\plus}$ or $\vlabel{i}{\minus}$
 belongs to~$X$ 
 for every edge $(j {\,\rightarrow\,} i)\in E$ and vertex $i\in V$, respectively,
 while the atom containing the opposite label cannot belong to
 a $\subseteq$-minimal model of~$P^X$.
% If either of them
% $\elabel{j}{i}{\plus}$ and $\elabel{j}{i}{\minus}$ or
% $\vlabel{i}{\plus}$ and $\vlabel{i}{\minus}$, respectively, 
%  is included in~$X$,
%  then the ground instance of the second or first rule in~(\ref{eq:label:def}),
%  respectively, is satisfied,
% for an edge $(j {\,\rightarrow\,} i)\in E$  or a vertex $i\in V$ is satisfied,
% so that the atom containing the opposite label cannot belong to
% a $\subseteq$-minimal model of~$P^X$.
Hence, there is at most one term~$s$ %\in\{\plus,\minus\}$
 such that~$\elabel{j}{i}{s}\in X$ or~$\vlabel{i}{s}\in X$
 for an edge $(j {\,\rightarrow\,} i)\in E$ or vertex $i\in V$, respectively,
 and it holds that $s\in\{\plus,\minus\}$, which
% Looking at the definitions of~$\sigma'$ and~$\mu'$,
 allows us to conclude that~$\sigma'$ and~$\mu'$ are total labelings.

As regards extending~$\sigma$ and~$\mu$,
we have that fact $\obselabel{j}{i}{s}.$ or $\obsvlabel{i}{s}.$
belongs to $\tau((V,E,\sigma),\mu)$ if $\sigma(j,i)=s$ or $\mu(i)=s$, respectively,
is given.
This implies that
$\elabel{j}{i}{s}\in X$ or $\vlabel{i}{s}\in X$, respectively,
as the fourth or third rule in~(\ref{eq:label:def}) would be unsatisfied otherwise.
Thus, $\sigma'(j,i)=s$ if $\sigma(j,i)=s$,
and $\mu'(i)=s$ if $\mu(i)=s$.

It remains to be shown that $\mu'(i)$ is consistent
for each non-input vertex $i\in V$.
To this end,
we note that the integrity constraint
\begin{equation*}
\begin{array}{@{{}\leftarrow{}}l}
   \vlabel{V}{S}, \naf{\influence{V}{S}}, \naf{\iinput{V}}.
\end{array}
\end{equation*}
from~(\ref{eq:inconsistent}) necessitates
$\influence{i}{r}\in X$ if $\mu'(i)=r$ (that is, if $\vlabel{i}{r}\in X$)
for a non-input vertex $i\in V$.
Otherwise, $P^X$ would contain an unsatisfied ground instance
in view that $\iinput{i}\in X$ exactly if fact $\iinput{i}.$
is included in $\tau((V,E,\sigma),\mu)$.
However, any ground instances of the integrity constraint contributing 
 to~$P^X$ do not contain atoms over predicate \atomfont{receive}.
 Such atoms % over predicate $\atomfont{receive}$
 can only be derived using the following rules from~(\ref{eq:influence}):
\begin{equation*}%\label{proof:eq:influence}
  \begin{array}{r@{{}\leftarrow{}}l}
  \influence{V}{\plus}  & \elabel{U}{V}{S}, \vlabel{U}{S}. \\
  \influence{V}{\minus} & \elabel{U}{V}{S}, \vlabel{U}{T}, S\neq T.
  \end{array}  
\end{equation*}
Since $X$ is a $\subseteq$-minimal model of~$P^X$,
 $\influence{i}{\plus}\in X$ or $\influence{i}{\minus}\in X$
 is possible only if
 $\elabel{j}{i}{s}\in X$ and $\vlabel{j}{t}\in X$
 such that $s=t$ or $s\neq t$, that is,
 if $\sigma'(j,i)=s$ and $\mu'(j)=t$ such that
 $\mu'(j)\sigma'(j,i)=\plus$ or $\mu'(j)\sigma'(j,i)=\minus$,
respectively.
As $\vlabel{i}{r}$ is accompanied by $\influence{i}{r}$ in~$X$
for each non-input vertex $i\in V$, 
this allows us to conclude that
$\mu'(i)=r$
implies $\mu'(j)\sigma'(j,i)=r$ for some regulator~$j$ of~$i$.
Hence, we have that $\mu'(i)$ is consistent
for each non-input vertex $i\in V$.
\hfill
\end{proof}

% \begin{theorem}[Theorem~\ref{thm:cons-compl}]
% Let $(V,E,\sigma)$ be an influence graph and
% $\mu : V \rightarrow \{\plus,\minus\}$ a (partial) vertex labeling.

% If $(V,E,\sigma)$ and $\mu$ are consistent,
% then there is an answer set of $P_C \cup \tau((V,E,\sigma),\mu)$.
% \end{theorem}

\begin{proof}[Proof of Theorem~\ref{thm:cons-compl}]
Assume that $(V,E,\sigma)$ and $\mu$ are consistent.
Then, there are total extensions
$\sigma' : E \rightarrow \{\plus,\minus\}$ of~$\sigma$
and $\mu' : V \rightarrow \{\plus,\minus\}$ of~$\mu$
such that, for each non-input vertex $i\in V$,
we have $\mu'(i)=\mu'(j)\sigma'(j,i)$ for some edge $j {\,\rightarrow\,} i$ in~$E$.

We consider the following set~$X$ of atoms:
\begin{equation*}
\begin{array}{l@{}c@{}l@{}l}
  X
  & {} = {} &
  \{\vertex{i},\vlabel{i}{s} & {} \mid i\in V, \mu'(i)=s\} 
  \\ & \cup {} &
  \{\edge{j}{i},\elabel{j}{i}{s} & {} \mid (j {\,\rightarrow\,} i)\in E, \sigma'(j,i)=s\} 
  \\ & \cup {} &
  \{\influence{i}{ts} & {} \mid (j {\,\rightarrow\,} i)\in E, \sigma'(j,i)=s,\mu'(j)=t\}
  \\ & \cup {} &
  \{\obselabel{j}{i}{s} & {} \mid (j {\,\rightarrow\,} i)\in E, \sigma(j,i)=s\}
  \\ & \cup {} &
  \{\obsvlabel{i}{s} & {} \mid i\in V, \mu(i)=s\}
  \\ & \cup {} &
  \{\iinput{i} & {} \mid i\in V \text{ is an input}\}\ \text{.}
\end{array}
\end{equation*}
For showing that~$X$ is an answer set of $P_C\cup\tau((V,E,\sigma),\mu)$, 
% (such that $\{i \mid \activ{i}\in X\}=W$)
we need to verify that~$X$ is a $\subseteq$-minimal model of
\begin{equation*}
\begin{array}{r@{}c@{}l}
P^X
& {} = {} &
\{
  (\head{r}\leftarrow\poslits{\body{r}})\theta
  \mid 
{} \\
& & \qquad 
  r\in P_C\cup\tau((V,E,\sigma),\mu),
  (\neglits{\body{r}}\theta)\cap X=\emptyset,
  \theta:\mathit{var}(r)\rightarrow\mathcal{U}
\}
\end{array}
\end{equation*}
 where $\mathit{var}(r)$ is the set of all variables that occur in a rule~$r$,
 $\mathcal{U}$ is the set of all constants appearing in $P_C\cup\tau((V,E,\sigma),\mu)$,
 and~$\theta$ is a ground substitution for the variables in~$r$.

To start with,
 we note that~$X$ includes
 an atom $\vertex{i}$, $\edge{j}{i}$, $\obselabel{j}{i}{s}$,
 $\obsvlabel{i}{s}$, and $\iinput{i}$, respectively,
 exactly if there is a fact with the atom in the head in $\tau((V,E,\sigma),\mu)$.
Each of these facts belongs also to~$P^X$,
 is satisfied by $X$,
 but not by any set~$Y$ of atoms excluding at least one of the head atoms.
Furthermore, since~$\sigma'$ and~$\mu'$ are total mappings, % that extend~$\sigma$ and~$\mu$,
we have that
$|\{\elabel{j}{i}{\plus},\elabel{j}{i}{\minus}\}\cap X|=1$ and
$|\{\vlabel{i}{\plus},\vlabel{i}{\minus}\}\cap X|=1$
for every $(j {\,\rightarrow\,} i)\in E$ and $i\in V$, respectively.
Hence, $X$, but no subset~$Y$ of~$X$ excluding at least one atom
over predicates \atomfont{labelE} and \atomfont{labelV}, satisfies
all ground instances of the following rules from~(\ref{eq:total}) in~$P^X$:
\begin{equation*}
  \begin{array}{r@{}c@{}l}
  \vlabel{V}{\plus}; \vlabel{V}{\minus} & {} \leftarrow {} &
    \vertex{V}.
  \\
  \elabel{U}{V}{\plus}; \elabel{U}{V}{\minus} & {} \leftarrow {} &
    \edge{U}{V}.
  \end{array}  
\end{equation*}
In addition, since~$\sigma'$ and~$\mu'$ extend~$\sigma$ and~$\mu$, respectively,
all ground instances of the following rules from~(\ref{eq:known}) in~$P^X$
are satisfied by~$X$:
\begin{equation*}
  \begin{array}{r@{}c@{}l}
  \vlabel{V}{S} & {} \leftarrow {} &
    \obsvlabel{V}{S}.
  \\
  \elabel{U}{V}{S} & {} \leftarrow {} &
    \obselabel{U}{V}{S}.
  \end{array}  
\end{equation*}
Since $\elabel{j}{i}{s}\in X$ and $\vlabel{j}{t}\in X$
if $\sigma'(j,i)=s$ and $\mu'(j)=t$, respectively,
we have that $\influence{i}{ts}\in X$ exactly if
there is a ground instance of the rules
\begin{equation*}
  \begin{array}{r@{}c@{}l}
  \influence{V}{\plus} & {} \leftarrow {} & 
    \elabel{U}{V}{S}, \vlabel{U}{S}.
  \\
  \influence{V}{\minus} & {} \leftarrow {} & 
    \elabel{U}{V}{S}, \vlabel{U}{T}, S\neq T.
  \end{array}  
\end{equation*}
from~(\ref{eq:influence}) in $P^X$ such that
$\elabel{j}{i}{s},\vlabel{j}{t}\in X$
occur in the body and $\influence{i}{ts}$ in the head.
Hence, no subset~$Y$ of~$X$ excluding any atom over predicate
\atomfont{receive} is a model of~$P^X$.
Finally,
since $\mu'(i)=\mu'(j)\sigma'(j,i)$ %holds
for each non-input vertex $i\in V$ and
some $j {\,\rightarrow\,} i$ in~$E$,
% such that
% $\mu'(i)=\mu'(j)\sigma'(j,i)$,
% holds for 
% every non-input vertex $i\in V$ and some edge $j {\,\rightarrow\,} i$ in~$E$,
$\vlabel{i}{r}\in\nolinebreak X$ implies that $\influence{i}{r}\in X$.
That is, the ground instances of the integrity constraint
\begin{equation*}
\begin{array}{@{{}\leftarrow{}}l}
   \vlabel{V}{S}, \naf{\influence{V}{S}}, \naf{\iinput{V}}.
\end{array}
\end{equation*}
from~(\ref{eq:inconsistent}) that contribute to $P^X$ 
are satisfied by~$X$.

We have now investigated all rules in $P_C\cup\tau((V,E,\sigma),\mu)$
and shown that their ground instances in~$P^X$ are satisfied by~$X$.
Furthermore, we have checked for all atoms in~$X$ that they cannot be
excluded in any model $Y\subset X$ of~$P^X$.
That is, $X$ is indeed a $\subseteq$-minimal model of~$P^X$
and thus an answer set of $P_C\cup\tau((V,E,\sigma),\mu)$.
\hfill
\end{proof}

%%% Local Variables: 
%%% mode: latex
%%% TeX-master: "paper"
%%% End: 

\section{Proof of Theorem~\ref{thm:diag-sound} and~\ref{thm:diag-compl}}
\label{app:proof:diagnosis}

This appendix provides proofs for soundness and completeness of the MIC extraction
encoding in Section~\ref{sec:diagnosis}.
We use $\tau((V,E,\sigma),\mu)$ as defined in~(\ref{proof:eq:instance})
to refer to the facts representing an influence graph $(V,E,\sigma)$
and a (partial) vertex labeling $\mu: V \rightarrow \{\plus,\minus\}$.
% To begin with, we formalize the representation of instances,
% as described in Section~\ref{subsec:instance},
% by defining a mapping~$\tau$
% of an influence graph $(V,E,\sigma)$
% (where
% $V$ is the set of vertices,
% $E$ the set of edges, and 
% $\sigma : E \rightarrow \{\plus,\minus\}$ a (partial) labeling of the edges)
% and a (partial) vertex labeling $\mu: V \rightarrow \{\plus,\minus\}$ by:
% %
% \begin{equation}\label{eq:instance}
% \begin{array}[b]{l@{}c@{}l@{}l}
%   \tau((V,E,\sigma),\mu)
%   & {} = {} &
%   \{\vertex{i}. & {} \mid i\in V\} 
%   \\ & \cup {} &
%   \{\edge{j}{i}. & {} \mid (j {\,\rightarrow\,} i)\in E\} 
%   \\ & \cup {} &
%   \{\obselabel{j}{i}{s}. & {} \mid (j {\,\rightarrow\,} i)\in E, \sigma(j,i)=s\}
%   \\ & \cup {} &
%   \{\obsvlabel{i}{s}. & {} \mid i\in V, \mu(i)=s\}
%   \\ & \cup {} &
%   \{\iinput{i}. & {} \mid i\in V \text{ is an input}\}\ \text{.}
% \end{array}
% \end{equation}
%Furthermore,
By~$P_D$,
we denote the           encoding consisting of the schematic
rules in~(\ref{eq:known}), (\ref{eq:active}), (\ref{eq:inconsistency}), and~(\ref{eq:minimal}).
%by~$P_D$.
%
% Conversely, a set~$P_I$ of ground facts is a \emph{diagnosis instance}
% if $\tau^{-1}(P_I)$ is defined, that is,
% if there are an influence graph~$(V,E,\sigma)$ and a
% vertex labeling~$\mu: V \rightarrow \{\plus,\minus\}$
% such that $\tau((V,E,\sigma),\mu)=P_I$.

As an auxiliary concept,
for any subset~$W\subseteq V$,
we say that
$\sigma': E \rightarrow \{\plus,\minus\}$ and
$\mu': V \rightarrow \{\plus,\minus\}$ are
\emph{witnessing labelings} for~$W$
if the following conditions hold:
\begin{enumerate}
\item $\sigma'$ and $\mu'$ are total,
\item if $\sigma(j,i)$ is defined, then $\sigma'(j,i)=\sigma(j,i)$,
\item if $\mu(i)$ is defined, then $\mu'(i)=\mu(i)$, and
\item $\mu'(i)$ is consistent (relative to~$\sigma'$) for each non-input vertex $i\in W$.
\end{enumerate}
The above conditions make sure that $\sigma'$ and $\mu'$ are
total extensions of~$\sigma$ and~$\mu$, respectively,
such that the variations of vertices in~$W$ are explained.
Comparing Definition~\ref{def:mic}, the first condition
requires the absence of witnessing labelings for a MIC~$W$,
while the second condition stipulates the existence of
witnessing labelings for each $W'\subset W$.

% In what follows, we recall and prove the soundness and completeness theorems
% formulated in Section~\ref{sec:diagnosis:theorems}.
% % The soundness of our diagnosis encoding in Section~\ref{sec:diagnosis}
% % can be stated as follows.
% %
% \begin{theorem}[Theorem~\ref{thm:diag-sound}] % \label{thm:diag-sound}
% Let  $(V,E,\sigma)$ be an influence graph and
% $\mu : V \rightarrow \{\plus,\minus\}$ a (partial) vertex labeling.

% If $X$ is an answer set of $P_D\cup\tau((V,E,\sigma),\mu)$,
% then $\{i \mid \activ{i}\in X\}$ is a MIC.
% \end{theorem}

\begin{proof}[Proof of Theorem~\ref{thm:diag-sound}]
Assume that~$X$ is an answer set of $P_D\cup\tau((V,E,\sigma),\mu)$.
Furthermore, let
\begin{equation*}
\begin{array}{r@{}c@{}l}
P^X
& {} = {} &
\{
  (\head{r}\leftarrow\poslits{\body{r}})\theta
  \mid 
{} \\
& & \qquad 
  r\in P_D\cup\tau((V,E,\sigma),\mu),
  (\neglits{\body{r}}\theta)\cap X=\emptyset,
  \theta:\mathit{var}(r)\rightarrow\mathcal{U}
\}
\end{array}
\end{equation*}
where $\mathit{var}(r)$ is the set of all variables that occur in a rule~$r$,
$\mathcal{U}$ is the set of all constants appearing in $P_D\cup\tau((V,E,\sigma),\mu)$,
and~$\theta$ is a ground substitution for the variables in~$r$.
Then, by the definition of an answer set,
we know that~$X$ is a $\subseteq$-minimal model of~$P^X$.

Let $W=\{i \mid \activ{i}\in X\}$.
We have to show that the following conditions hold:
\begin{enumerate}
\item There are witnessing labelings for each $W'\subset W$.
\item There are no witnessing labelings for~$W$.
\end{enumerate}
We below consider these conditions one after the other.

\paragraph{Condition 1.}
Let $W'=W\setminus\{k\}$ for any $k\in W$.
Furthermore,
define $\sigma'$ and $\mu'$ as follows:
\begin{equation*}
\begin{array}{r@{}r@{}r@{}c@{}l}
\sigma'
& {} = {} & \{
  (j {\,\rightarrow\,} i)\mapsto s 
& {} \mid {} & (j {\,\rightarrow\,} i)\in E, \elabell{k}{j}{i}{s}\in X\}
\\
& {} \cup {} & \{ 
  (j {\,\rightarrow\,} i)\mapsto \plus 
& {} \mid {} & (j {\,\rightarrow\,} i)\in E, \elabell{k}{j}{i}{\plus}\notin X, \elabell{k}{j}{i}{\minus}\notin X\}
\\[1mm]
\mu'
& {} = {} & \{
  i\mapsto s
& {} \mid {} & i\in V, \vlabell{k}{i}{s}\in X\}
\\
& {} \cup {} & \{
  i\mapsto \plus
& {} \mid {} & i\in V, \vlabell{k}{i}{\plus}\notin X, \vlabell{k}{i}{\minus}\notin X\}\ \text{.}
\end{array}
\end{equation*}
We  show that~$\sigma'$ and~$\mu'$ are witnessing labelings for~$W'$.

Regarding the uniqueness of labels assigned by~$\sigma'$ and~$\mu'$,
consider the following rules from~(\ref{eq:minimal}) including
predicates \atomfont{labelE'} and \atomfont{labelV'} in their heads:
\begin{equation}\label{eq:Wsub-label:def}
  \begin{array}{r@{}c@{}l}
  \vlabell{W}{V}{\plus}; \vlabell{W}{V}{\minus} & {} \leftarrow {} &
    \activ{W}, \avertex{V}.
  \\
  \elabell{W}{U}{V}{\plus}; \elabell{W}{U}{V}{\minus} & {} \leftarrow {} &
    \activ{W}, \aedge{U}{V}.
  \\[1mm]
  \vlabell{W}{V}{S} & {} \leftarrow {} &
    \activ{W}, \obsvlabel{V}{S}.
  \\
  \elabell{W}{U}{V}{S} & {} \leftarrow {} &
    \activ{W}, \obselabel{U}{V}{S}.
  \end{array}  
\end{equation}
Since the given (partial) labelings~$\sigma$ and~$\mu$ assign unique labels
to the elements of their domains,
facts defining \atomfont{observedE} and \atomfont{observedV}
are of the form 
$\obselabel{j}{i}{\plus}.$ or $\obselabel{j}{i}{\minus}.$ and
$\obsvlabel{i}{\plus}.$ or $\obsvlabel{i}{\minus}.$, respectively,
and at most one of these facts is contained in $\tau((V,E,\sigma),\mu)$
for an edge $(j {\,\rightarrow\,} i)\in E$ or vertex $i\in V$.
Because~$X$ is a $\subseteq$-minimal model of~$P^X$, the atoms in the heads of facts
are in~$X$, and all atoms in~$X$ over predicates 
\atomfont{observedE} and \atomfont{observedV} are derived from facts in $\tau((V,E,\sigma),\mu)$,
in view that these predicates do not occur in the head of any rule in~$P_D$.
Hence, at most one of the atoms
$\elabell{k}{j}{i}{\plus}$ and $\elabell{k}{j}{i}{\minus}$ or
$\vlabell{k}{i}{\plus}$ and $\vlabell{k}{i}{\minus}$, respectively,
is derivable for an edge $(j {\,\rightarrow\,} i)\in E$ or vertex $i\in\nolinebreak V$
from a ground instance of the fourth or third rule in~(\ref{eq:Wsub-label:def})
and then included in~$X$.
If either of $\elabell{k}{j}{i}{\plus}$ and $\elabell{k}{j}{i}{\minus}$
or $\vlabell{k}{i}{\plus}$ and $\vlabell{k}{i}{\minus}$, respectively,
is included in~$X$,
then the ground instance of the second or first rule in~(\ref{eq:Wsub-label:def})
for~$k$ and an edge $(j {\,\rightarrow\,} i)\in E$ or vertex $i\in V$
is satisfied, 
so that the atom containing the opposite label cannot belong to
a $\subseteq$-minimal model of~$P^X$.
Hence, there is at most one term~$s$ %\in\{\plus,\minus\}$
such that~$\sigma'(j,i)=s$ or~$\mu'(i)=s$ 
for an edge $(j {\,\rightarrow\,} i)\in E$ or vertex $i\in\nolinebreak V$, respectively,
and it holds that $s\in\{\plus,\minus\}$.
Furthermore, looking at the definitions of~$\sigma'$ and~$\mu'$,
it is obvious that both are total, which allows us to conclude
that~$\sigma'$ and~$\mu'$ are total labelings.

As regards extending~$\sigma$ and~$\mu$,
we have that fact $\obselabel{j}{i}{s}.$ or $\obsvlabel{i}{s}.$
belongs to $\tau((V,E,\sigma),\mu)$ if $\sigma(j,i)=s$ or $\mu(i)=s$, respectively,
is given.
Along with the premise that $\activ{k}\in X$, this implies that
$\elabell{k}{j}{i}{s}\in X$ or $\vlabell{k}{i}{s}\in\nolinebreak X$, respectively,
as the fourth or third rule in~(\ref{eq:Wsub-label:def}) would be unsatisfied otherwise.
Hence, we have $\sigma'(j,i)=s$ if $\sigma(j,i)=s$,
and $\mu'(i)=s$ if $\mu(i)=s$.

It remains to be shown that $\mu'(i)$ is consistent
for each non-input vertex $i\in W'$.
% To establish this,
% we first note that vertices in $W$ cannot be input because,
% if fact $\iinput{i}.$ belongs to $\tau((V,E,\sigma),\mu)$,
% then $\iinput{i}$ must be included in $X$, so that rule
% %
% \begin{equation}\label{eq:actrepeat}
% \begin{array}{r@{{}\leftarrow{}}l}
%   \activ{V};\inactiv{V}                      & \vertex{V}, \naf{\iinput{V}}.
% \end{array}
% \end{equation}
% %
% from (\ref{eq:active}) does not contribute a ground instance for~$i$ to~$P^X$.
% Since $\activ{i}$ cannot be derived from any other ground rule in~$P^X$,
% the fact that~$X$ is a $\subseteq$-minimal model of~$P^X$ implies that
% $\activ{i}\notin X$ for any input vertex~$i$.
% Let us next consider the following rules from~(\ref{eq:active}):
To establish this, we first consider the following rules from~(\ref{eq:active}):
\begin{equation}\label{eq:micedge}
\begin{array}{r@{{}\leftarrow{}}l}
  \aedge{U}{V}                               & \edge{U}{V}, \activ{V} . \\
  \avertex{U}                                & \aedge{U}{V}. \\
  \avertex{V}                                & \activ{V}.
\end{array}
\end{equation}
In view that fact $\edge{j}{i}.$ belongs to $\tau((V,E,\sigma),\mu)$
for every $(j {\,\rightarrow\,} i)\in E$,
we conclude that $\edge{j}{i}\in X$.
Along with $\activ{i}\in X$ for every $i\in W$,
it follows that
$\aedge{j}{i}\in X$ for every $(j {\,\rightarrow\,} i)\in E$ such that $i\in W$,
and $\avertex{i}\in X$ for every $i\in W$. % or regulator~$i$ of some vertex in~$W$.
The last observation and the first rule in~(\ref{eq:Wsub-label:def})
imply that $\vlabell{k}{i}{\plus}\in X$ or $\vlabell{k}{i}{\minus}\in X$
for every $i\in W$.
For $i\in W'$, i.e., $i\neq k$,
the integrity constraint
\begin{equation*}
\begin{array}{@{{}\leftarrow{}}l}
   \vlabell{W}{V}{S}, \activ{V}, V\neq W, \naf{\influencel{W}{V}{S}}.
\end{array}
\end{equation*}
from~(\ref{eq:minimal})
imposes $\influencel{k}{i}{\plus}\in X$ if $\vlabell{k}{i}{\plus}\in X$,
and $\influencel{k}{i}{\minus}\in\nolinebreak X$ if $\vlabell{k}{i}{\minus}\in X$,
while any ground instances of the integrity constraint contributing 
to~$P^X$ do not contain atoms over predicate \atomfont{receive'}.
Such atoms % over predicate $\atomfont{receive'}$
can only be derived using the following rules from~(\ref{eq:minimal}):
\begin{equation*}%\label{eq:micedge}
\begin{array}{r@{{}\leftarrow{}}l}
  \influencel{W}{V}{\plus} &  
    \elabell{W}{U}{V}{S}, \vlabell{W}{U}{S}, V\neq W.
  \\
  \influencel{W}{V}{\minus} &  
    \elabell{W}{U}{V}{S}, \vlabell{W}{U}{T}, V\neq W, S\neq T.
\end{array}
\end{equation*}
Since $X$ is a $\subseteq$-minimal model of~$P^X$,
$\influencel{k}{i}{\plus}\in X$ or $\influencel{k}{i}{\minus}\in X$
is possible only if
$\elabell{k}{j}{i}{s}\in X$ and $\vlabell{k}{j}{t}\in X$
such that $s=t$ or $s\neq t$, respectively.
Comparing $\tau((V,E,\sigma),\mu)$ and the rules in~(\ref{eq:Wsub-label:def}), (\ref{eq:micedge}),
as well as~(\ref{eq:actrepeat}) reveals that
$(j {\,\rightarrow\,} i)\in E$ is a necessary condition for $\elabell{k}{j}{i}{s}\in X$,
and the same applies to $j\in V$ and $\vlabell{k}{j}{t}\in X$.
By the construction of~$\sigma'$ and~$\mu'$,
$\elabell{k}{j}{i}{s}\in\nolinebreak X$ implies that $\sigma'(j,i)=s$ and
$\vlabell{k}{j}{t}\in X$ that $\mu'(j)=t$.
We conclude that
$\influencel{k}{i}{\plus}\in X$ or $\influencel{k}{i}{\minus}\in X$
necessitates $\mu'(j)\sigma'(j,i)=\plus$ or $\mu'(j)\sigma'(j,i)=\minus$, respectively,
for some regulator~$j$ of~$i$.
Finally, we have $\mu'(i)=\plus$ if $\vlabell{k}{i}{\plus}\in X$ (and $\influencel{k}{i}{\plus}\in X$),
and $\mu'(i)=\minus$ if $\vlabell{k}{i}{\minus}\in X$ (and $\influencel{k}{i}{\minus}\in X$).
This shows that~$i$ receives some influence matching $\mu'(i)$,
so that $\mu'(i)$ is consistent.
Since $i\in W'$ is arbitrary,
$\sigma'$ and~$\mu'$ are witnessing labelings for~$W'$.

To conclude the proof of the first condition to verify,
we note that witnessing labelings for $W'$ are also
witnessing labelings for all subsets of $W'$.
Hence, it is sufficient to check the existence of
witnessing labelings for sets $W'=W\setminus\{k\}$ for any $k\in W$.
As shown above, 
an answer set~$X$ of $P_D\cup\tau((V,E,\sigma),\mu)$ yields
witnessing labelings for them.
Hence, the second condition in Definition~\ref{def:mic} holds
for $W=\{i \mid \activ{i}\in X\}$.

\paragraph{Condition 2.}
We now show by contradiction that there cannot be witnessing
labelings for~$W$.
To establish this,
we first note that vertices in $W$ cannot be input because,
if fact $\iinput{i}.$ belongs to $\tau((V,E,\sigma),\mu)$,
then $\iinput{i}$ must be included in $X$, so that the rule
\begin{equation}\label{eq:actrepeat}
\begin{array}{r@{{}\leftarrow{}}l}
  \activ{V};\inactiv{V}                      & \vertex{V}, \naf{\iinput{V}}.
\end{array}
\end{equation}
from (\ref{eq:active}) does not contribute a ground instance for~$i$ to~$P^X$.
Since $\activ{i}$ cannot be derived from any other ground rule in~$P^X$,
the fact that~$X$ is a $\subseteq$-minimal model of~$P^X$ implies that
$\activ{i}\notin X$ for any input vertex~$i$.
Furthermore,
% we have that
the integrity constraint
\begin{equation}\label{eq:botget}
\begin{array}{@{{}\leftarrow{}}l}
   \naf{\bottom}.
\end{array}
\end{equation}
from~(\ref{eq:inconsistency}) necessitates $\bottom\in X$
because $X$ cannot be a model of~$P^X$ otherwise.
Then, we get
$\vlabel{i}{\plus},\vlabel{i}{\minus}\in X$ and
$\elabel{j}{i}{\plus},\elabel{j}{i}{\minus}\in X$
for all vertices $i\in V$ and edges $(j {\,\rightarrow\,} i)\in E$, respectively,
due to the following rules from~(\ref{eq:inconsistency}):

\begin{equation}\label{eq:botfollow}
\begin{array}{r@{{}\leftarrow{}}l}
  \vlabel{V}{\plus}  & 
                       \bottom, \vertex{V}.
  \\
  \vlabel{V}{\minus} & 
                       \bottom, \vertex{V}.
  \\
  \elabel{U}{V}{\plus}  & 
                          \bottom, \edge{U}{V}.
  \\
  \elabel{U}{V}{\minus} & 
                          \bottom, \edge{U}{V}.
\end{array}
\end{equation}

We  now show that the existence of witnessing labelings
for~$W$ yields a contradiction to the fact that~$X$ is a
$\subseteq$-minimal model of~$P^X$.
To this end, assume that $\sigma'$ and $\mu'$ are witnessing labelings for~$W$.
Then, let
\begin{equation*}
\begin{array}{l@{}c@{}l@{}l}
Y
& {} = {} &
\multicolumn{2}{@{}l}{
\begin{array}[t]{@{}r@{}l@{}l}
(
  X\setminus
  (
&
     \{\bottom\}
\\ 
{} \cup {} & \{\vlabel{i}{s} & {} \mid \vlabel{i}{s}\in X\}
\\ 
{} \cup {} & \{\elabel{j}{i}{s} & {} \mid \elabel{j}{i}{s}\in X\}
\\ 
{} \cup {} & \{\contrary{j}{i} & {} \mid \contrary{j}{i}\in X\}
  )
)
\end{array}
}
\\
& {} \cup {} &
  \{\vlabel{i}{s} & {} \mid i\in V,\mu'(i)=s\}
\\
& {} \cup {} &
  \{\elabel{j}{i}{s} & {} \mid (j {\,\rightarrow\,} i)\in E,\sigma'(j,i)=s\}
\\
& {} \cup {} &
  \{\contrary{j}{i} & {} \mid (j {\,\rightarrow\,} i)\in E, \mu'(i)\neq\mu'(j)\sigma'(j,i)\}\ \text{.}
\end{array}
\end{equation*}
Since $\bottom\in X\setminus Y$ and~$X$
contains a maximum amount of atoms over
predicates \atomfont{labelV}, \atomfont{labelE}, and \atomfont{opposite}
(the atoms over \atomfont{opposite} are consequences of the inclusion
 of atoms over \atomfont{labelV} and \atomfont{labelE}),
we have that $Y\subset X$, and we  show that~$Y$ is a model of~$P^X$.

Considering the contributions of the
facts in $\tau((V,E,\sigma),\mu)$ and the rules in~(\ref{eq:minimal}) to~$P^X$,
we observe that the atoms over predicates occurring in them are interpreted the same in~$X$ and~$Y$.
Hence, such facts and rules stay satisfied by~$Y$ because they were already satisfied by~$X$.
The same applies to the rules from~(\ref{eq:active}) repeated in~(\ref{eq:micedge})
and~(\ref{eq:actrepeat}).
Furthermore,
since~$\sigma'$ and~$\mu'$ are total and extend~$\sigma$ and~$\mu$, respectively,
the contributions of the following rules from~(\ref{eq:known})
and~(\ref{eq:active}) to $P^X$ are satisfied by~$Y$:
\begin{equation*}
\begin{array}{r@{{}\leftarrow{}}l}
  \vlabel{V}{S}    & \obsvlabel{V}{S}.    \\
  \elabel{U}{V}{S} & \obselabel{U}{V}{S}. \\[1mm]
  \vlabel{V}{\plus}   ;\vlabel{V}{\minus}    & \avertex{V}.  \\
  \elabel{U}{V}{\plus};\elabel{U}{V}{\minus} & \aedge{U}{V}.
\end{array}
\end{equation*}
Since the integrity constraint in~(\ref{eq:botget}) does not belong to~$P^X$
and the rules in~(\ref{eq:botfollow}) are satisfied by~$Y$ in view of
$\bottom\notin Y$,
it remains to consider the following rules from~(\ref{eq:inconsistency}):
\begin{equation*}
\begin{array}{r@{{}\leftarrow{}}l}
  \contrary{U}{V}  & 
                     \elabel{U}{V}{\minus}, \vlabel{U}{S}, \vlabel{V}{S}.
  \\
  \contrary{U}{V}  & 
                     \elabel{U}{V}{\plus},  \vlabel{U}{S}, \vlabel{V}{T}, S\neq T.
  \\[1mm]
  \bottom          & 
                     \activ{V}, \contrary{U}{V}:\edge{U}{V}.
\end{array}
\end{equation*}
The rules defining predicate \atomfont{opposite} are such that,
in order to satisfy their ground instances in~$P^X$,
$Y$ must contain $\contrary{j}{i}$ if
$\elabel{j}{i}{r}$, $\vlabel{j}{s}$, and $\vlabel{i}{t}$ belong to~$Y$ such that
$t\neq sr$.
This matches the definition of~$Y$,
including $\elabel{j}{i}{r}$ if $\sigma'(j,i)\!=\nolinebreak\! r$,
$\vlabel{j}{s}$ if $\mu'(j)\!=\!s$, 
$\vlabel{i}{t}$ if $\mu'(i)\!=\!t$, and
$\contrary{j}{i}$ if $\mu'(i)\!\neq\! \mu'(j)\sigma'(j,i)$.
Hence, rules defining \atomfont{opposite} in~$P^X$
are satisfied by~$Y$.
It remains to be shown that $\bottom$ is not derivable from any ground instance of the last rule.
In this regard, recall that $W=\{i \mid \activ{i}\in X\}=\{i \mid \activ{i}\in Y\}$,
and we have seen above that $\activ{i}$ can only belong to~$X$ if $i$ is not an input.
As $\sigma'$ and~$\mu'$ are witnessing labelings for~$W$,
for every $i\in W$, there is an edge $(j {\,\rightarrow\,} i)\in E$
such that $\mu'(i)=\mu'(j)\sigma'(j,i)$.
By the definition of~$Y$, this implies
$\contrary{j}{i}\notin Y$,
while $\edge{j}{i}$ belongs to~$X$ and~$Y$ because~$X$ and~$Y$
are models of $\tau((V,E,\sigma),\mu)$.
As a consequence, for every $i\in W$,
we have $\{\contrary{j}{i} \mid \edge{j}{i}\in Y\}\not\subseteq Y$,
so that the ground instance for~$i$ in~$P^X$
of the rule with $\bottom$ in the head
is satisfied by~$Y$.
% Since $\activ{i}\notin Y$ for any $i\notin W$,
We have thus established that~$Y\subset X$ is indeed a model of~$P^X$,
a contradiction to the assumption that~$X$ is a $\subseteq$-minimal model of~$P^X$
and an answer set of $P_D\cup\tau((V,E,\sigma),\mu)$.

The above contradiction shows that the second condition to verify,
which is the first condition in Definition~\ref{def:mic},
holds for $W=\{i \mid \activ{i}\in X\}$.
The fact that the second condition in Definition~\ref{def:mic} holds for~$W$ has 
been shown before.
%which now allows us to conclude
%that~
Hence, $W$ is a MIC.
\hfill%
\end{proof}

% The next result shows the inverse direction of Theorem~\ref{thm:diag-sound}
% to hold as well.
% %
% \begin{theorem}[Theorem~\ref{thm:diag-compl}] %\label{thm:diag-compl}
% Let  $(V,E,\sigma)$ be an influence graph and
% $\mu : V \rightarrow \{\plus,\minus\}$ a (partial) vertex labeling.
%
% If $W\subseteq V$ is a MIC, then there is an answer
% set~$X$ of $P_D\cup\tau((V,E,\sigma),\mu)$
% such that $\{i \mid \activ{i}\in X\}=W$.
% \end{theorem}
%
\begin{proof}[Proof of Theorem~\ref{thm:diag-compl}]
Assume that $W=\{k_1,\dots,k_n\}$ is a MIC.
Then, the following conditions hold:
\begin{enumerate}
\item There are witnessing labelings $\sigma_1,\mu_1,\dots,\sigma_n,\mu_n$
      for $W\setminus\{k_1\},\dots,W\setminus\{k_n\}$.
\item There are no witnessing labelings for~$W$.
\end{enumerate}
We  consider the following set~$X$ of atoms:
\begin{equation*}
\begin{array}{l@{}c@{}l@{}l}
  X
  & {} = {} &
  \{\vertex{i} & {} \mid i\in V\} 
  \\ & \cup {} &
  \{\edge{j}{i} & {} \mid (j {\,\rightarrow\,} i)\in E\} 
  \\ & \cup {} &
  \{\obselabel{j}{i}{s} & {} \mid (j {\,\rightarrow\,} i)\in E, \sigma(j,i)=s\}
  \\ & \cup {} &
  \{\obsvlabel{i}{s} & {} \mid i\in V, \mu(i)=s\}
  \\ & \cup {} &
  \{\iinput{i} & {} \mid i\in V \text{ is an input}\}
  \\ & \cup {} &
  \{\activ{i} & {} \mid i\in W\}
  \\ & \cup {} &
  \{\inactiv{i} & {} \mid i\in V\setminus W \text{ is not an input}\}
  \\ & \cup {} &
  \{\aedge{j}{i} & {} \mid (j {\,\rightarrow\,} i)\in E, i\in W\}
  \\ & \cup {} &
  \{\avertex{j} & {} \mid (j {\,\rightarrow\,} i)\in E, i\in W\}
  \\ & \cup {} &
  \{\avertex{i} & {} \mid i\in W\}
  \\ & \cup {} &
  \{\elabell{k_m}{j}{i}{r} & {} \mid (j {\,\rightarrow\,} i)\in E,i\in W,\sigma_m(j,i)=r,1\leq m\leq n\}
  \\ & \cup {} &
  \{\elabell{k_m}{j}{i}{r} & {} \mid (j {\,\rightarrow\,} i)\in E,\sigma(j,i)=r,1\leq m\leq n\}
  \\ & \cup {} &
  \{\vlabell{k_m}{j}{s} & {} \mid (j {\,\rightarrow\,} i)\in E,i\in W,\mu_m(j)=s,1\leq m\leq n\}
  \\ & \cup {} &
  \{\vlabell{k_m}{i}{s} & {} \mid i\in W,\mu_m(i)=s,1\leq m\leq n\}
  \\ & \cup {} &
  \{\vlabell{k_m}{i}{s} & {} \mid i\in V,\mu(i)=s,1\leq m\leq n\}
  \\ & \cup {} &
  \{\influencel{k_m}{i}{sr} & {} \mid (j {\,\rightarrow\,} i)\in E,i\in W,% \text{ or } (i {\,\rightarrow\,} k)\in E \text{ for } k\in W,
  {} \\ & & & \quad
                                      \sigma_m(j,i)=r,\mu_m(j)=s,i\neq k_m,1\leq m\leq n\}
  \\ & \cup {} &
  \{\influencel{k_m}{i}{sr} & {} \mid (j {\,\rightarrow\,} i)\in E,j\in W \text{ or } (j {\,\rightarrow\,} k)\in E \text{ for } k\in W, % \mu(j)=s,
  {} \\ & & & \quad
                                      \sigma(j,i)=r,\mu_m(j)=s,i\neq k_m,1\leq m\leq n\}
%   \\ & \cup {} &
%   \{\influencel{k_m}{i}{sr} & {} \mid (j {\,\rightarrow\,} i)\in E,(k {\,\rightarrow\,} j)\in E,j\in W, % \mu(j)=s,
%   {} \\ & & & \quad
%                                       \sigma(j,i)=r,\mu_m(j)=s,i\neq k_m,1\leq m\leq n\}
  \\ & \cup {} &
  \{\influencel{k_m}{i}{sr} & {} \mid (j {\,\rightarrow\,} i)\in E,
  {} \\ & & & \quad
                                      \sigma(j,i)=r,\mu(j)=s,i\neq k_m,1\leq m\leq n\}
  \\ & \cup {} &
  \{\vlabel{i}{\plus},\vlabel{i}{\minus} & {} \mid i\in V\}
  \\ & \cup {} &
  \{\elabel{j}{i}{\plus},\elabel{j}{i}{\minus} & {} \mid (j {\,\rightarrow\,} i)\in E\}
  \\ & \cup {} &
  \{\contrary{j}{i} & {} \mid (j {\,\rightarrow\,} i)\in E\}
  \\ & \cup {} &
  \{\bottom\}\ \text{.}
\end{array}
\end{equation*}
For showing that~$X$ is an answer set of $P_D\cup\tau((V,E,\sigma),\mu)$
(such that $\{i \mid \activ{i}\in X\}=W$),
we need to verify that~$X$ is a $\subseteq$-minimal model of
\begin{equation*}
\begin{array}{r@{}c@{}l}
P^X
& {} = {} &
\{
  (\head{r}\leftarrow\poslits{\body{r}})\theta
  \mid 
{} \\
& & \qquad 
  r\in P_D\cup\tau((V,E,\sigma),\mu),
  (\neglits{\body{r}}\theta)\cap X=\emptyset,
  \theta:\mathit{var}(r)\rightarrow\mathcal{U}
\}
\end{array}
\end{equation*}
where $\mathit{var}(r)$ is the set of all variables that occur in a rule~$r$,
$\mathcal{U}$ is the set of all constants appearing in $P_D\cup\tau((V,E,\sigma),\mu)$,
and~$\theta$ is a ground substitution for the variables in~$r$.

To start with,
we note that~$X$ includes
an atom $\vertex{i}$, $\edge{j}{i}$, $\obselabel{j}{i}{s}$,
$\obsvlabel{i}{s}$, and $\iinput{i}$, respectively,
exactly if there is a fact with the atom in the head in $\tau((V,E,\sigma),\mu)$.
Each of these facts belongs also to~$P^X$,
is satisfied by $X$,
but not by any set~$Y$ of atoms excluding at least one of the head atoms.

In view that $W$ cannot contain any input
(otherwise, satisfaction of the second condition in Definition~\ref{def:mic}
 would immediately imply violation of the first one),
we have that either $\activ{i}$ or $\inactiv{i}$ belongs to~$X$
for every non-input vertex $i\in V$.
Hence, $X$ satisfies 
all ground instances of the rule
\begin{equation*}
\begin{array}{r@{{}\leftarrow{}}l}
  \activ{V};\inactiv{V}                      & \vertex{V}, \naf{\iinput{V}}.
\end{array}
\end{equation*}
from (\ref{eq:active}) belonging to $P^X$,
while no set~$Y$ of atoms excluding both
$\activ{i}$ and $\inactiv{i}$ for any non-input vertex $i\in V$
satisfies all of these ground instances.

Considering ground instances of the rules
\begin{equation*}
\begin{array}{r@{{}\leftarrow{}}l}
  \aedge{U}{V}                               & \edge{U}{V}, \activ{V} . \\
  \avertex{U}                                & \aedge{U}{V}. \\
  \avertex{V}                                & \activ{V}.
\end{array}
\end{equation*}
from (\ref{eq:active}),
all of them belong to~$P^X$,
are satisfied by~$X$,
but not by any set~$Y$ of atoms such that
$\{\aedge{j}{i}\mid\aedge{j}{i}\in X\}\cup\{\avertex{i}\mid\avertex{i}\in X\}
 \not\subseteq\nolinebreak Y$ and
$\{\activ{i} \mid \activ{i}\in X\}\subseteq\{\activ{i}\mid \activ{i}\in Y\}$,
while it has been shown above that 
$\{\activ{i} \mid \activ{i}\in X\}\not\subseteq\{\activ{i}\mid \activ{i}\in Y\}$
necessitates
$\{\inactiv{i}\mid \inactiv{i}\in Y\}\not\subseteq\{\inactiv{i} \mid \inactiv{i}\in X\}$
for $Y$ being a model of~$P^X$.
Hence, there cannot be any model~$Y\subset X$ of~$P^X$
excluding some atom $\aedge{j}{i}$ or $\avertex{i}$
that belongs to~$X$.

Now turning our attention to atoms of form
$\elabell{k_m}{j}{i}{r}$ and $\vlabell{k_m}{j}{s}$,
we note that they are included in~$X$
if $\aedge{j}{i}\in X$ and $\avertex{j}\in X$, respectively,
and $\sigma_m(j,i)=r,\mu_m(j)=s$ in witnessing 
labelings~$\sigma_m$ and~$\mu_m$ for $W\setminus\{k_m\}$,
where $1\leq m\leq n$,
or if $\sigma(j,i)=r,\mu(j)=s$.
Then, the fact that $\activ{k_m}\in X$ and
labels assigned by~$\sigma_m$ and~$\mu_m$ are unique
and respect those assigned by~$\sigma$ and~$\mu$
implies that none of the atoms can be removed from~$X$
without violating some ground instance
of the rules
\begin{equation*}
  \begin{array}{r@{}c@{}l}
  \vlabell{W}{V}{\plus}; \vlabell{W}{V}{\minus} & {} \leftarrow {} &
    \activ{W}, \avertex{V}.
  \\
  \elabell{W}{U}{V}{\plus}; \elabell{W}{U}{V}{\minus} & {} \leftarrow {} &
    \activ{W}, \aedge{U}{V}.
  \\[1mm]
  \vlabell{W}{V}{S} & {} \leftarrow {} &
    \activ{W}, \obsvlabel{V}{S}.
  \\
  \elabell{W}{U}{V}{S} & {} \leftarrow {} &
    \activ{W}, \obselabel{U}{V}{S}.
  \end{array}  
\end{equation*}
from (\ref{eq:minimal}) that belongs to~$P^X$.
However, $X$ satisfies all of these ground instances by its construction.
We further consider the following rules from~(\ref{eq:minimal}):
\begin{equation*}
  \begin{array}{r@{}c@{}l}
  \influencel{W}{V}{\plus} & {} \leftarrow {} & 
    \elabell{W}{U}{V}{S}, \vlabell{W}{U}{S}, V\neq W.
  \\
  \influencel{W}{V}{\minus} & {} \leftarrow {} & 
    \elabell{W}{U}{V}{S}, \vlabell{W}{U}{T}, V\neq W, S\neq T.
  \end{array}  
\end{equation*}
As shown above,
$\elabell{k_m}{j}{i}{r}$ belongs to~$X$ if
$i\in W$ and $\sigma_m(j,i)=r$, or if
$\sigma(j,i)=\sigma_m(j,i)=r$.
Furthermore,
$\vlabell{k_m}{j}{s}$ is included in~$X$ if
$j\in W$ or $(j {\,\rightarrow\,} k)\in\nolinebreak E,\linebreak[1]k\in W$ and $\mu_m(j)=s$,
or if
$\mu(j)=\mu_m(j)=\nolinebreak s$.
Comparing the cross product of these conditions to the definition of~$X$
yields that an atom $\influencel{k_m}{i}{sr}$ belongs to~$X$ exactly
if $\elabell{k_m}{j}{i}{r}$ and $\vlabell{k_m}{j}{s}$ are in~$X$ and $i\neq k_m$.
Hence,
when excluding any of the atoms $\influencel{k_m}{i}{sr}$ from~$X$,
some ground instance of the above two rules belonging to~$P^X$
becomes unsatisfied, and so we have
that such atoms cannot be removed from~$X$
in order to construct a model $Y\subset X$ of~$P^X$.
% $\{\influencel{k_m}{i}{sr} \mid \influencel{k_m}{i}{sr}\in\nolinebreak X\}\linebreak[1]\subseteq Y$
% for any model~$Y$ of~$P^X$
% with
% $\{\elabell{k_m}{j}{i}{r} \mid \elabell{k_m}{j}{i}{r}\in X\}\cup\{\vlabell{k_m}{j}{s} \mid \linebreak[1] \vlabell{k_m}{j}{s}\in X\}\subseteq Y$.
Moreover, the fact that~$\sigma_m$ and~$\mu_m$ are witnessing labelings
for $W'=W\setminus\{k_m\}$ implies that all ground instances of the
integrity constraint
\begin{equation*}
\begin{array}{@{{}\leftarrow{}}l}
   \vlabell{W}{V}{S}, \activ{V}, V\neq W, \naf{\influencel{W}{V}{S}}.
\end{array}
\end{equation*}
from~(\ref{eq:minimal}) that belong to~$P^X$ are satisfied by~$X$.
In fact, for every $i\in W'$, there is some $(j {\,\rightarrow\,} i)\in E$
such that $\mu_m(i)=\mu_m(j)\sigma_m(j,i)$.
Since $\elabell{k_m}{j}{i}{\sigma_m(j,i)}$ and $\vlabell{k_m}{j}{\mu_m(j)}$
belong to~$X$, 
this implies that each
atom $\vlabell{k_m}{i}{\mu_m(i)}$ for $i\in W'$ is accompanied by
$\influencel{k_m}{i}{\mu_m(i)}=\influencel{k_m}{i}{\mu_m(j)\sigma_m(j,i)}$ in~$X$,
so that the ground instance for~$k_m$, $i$, and~$\mu_m(i)$ of the integrity constraint
is not in~$P^X$.

Finally, we consider atoms of the form
$\vlabel{i}{s}$, $\elabel{j}{i}{s}$, and $\contrary{j}{i}$ that belong to~$X$
for all $i\in V$ and $(j {\,\rightarrow\,} i)\in E$, respectively, and $s\in\{\plus,\minus\}$.
Since $\bottom$ is also in~$X$,
it is clear that the ground instances of the following rules from~(\ref{eq:known}),
(\ref{eq:active}), and~(\ref{eq:inconsistency}),
all of which belong to~$P^X$, are satisfied by~$X$:
\begin{equation*}
  \begin{array}{r@{}c@{}l}
  \vlabel{V}{S}    & {}\leftarrow{} & \obsvlabel{V}{S}.    \\
  \elabel{U}{V}{S} & {}\leftarrow{} & \obselabel{U}{V}{S}. \\
  \vlabel{V}{\plus}   ;\vlabel{V}{\minus}    & {}\leftarrow{} & \avertex{V}.  \\
  \elabel{U}{V}{\plus};\elabel{U}{V}{\minus} & {}\leftarrow{} & \aedge{U}{V}. \\[1mm]
  \contrary{U}{V}  & {}\leftarrow{} &
                     \elabel{U}{V}{\minus}, \vlabel{U}{S}, \vlabel{V}{S}.
  \\
  \contrary{U}{V}  & {}\leftarrow{} &
                     \elabel{U}{V}{\plus},  \vlabel{U}{S}, \vlabel{V}{T}, S\neq T.
  \\[1mm]
  \bottom          & {}\leftarrow{} &
                     \activ{V}, \contrary{U}{V}:\edge{U}{V}. \\[1mm]
  \vlabel{V}{\plus}  & {}\leftarrow{} &
                       \bottom, \vertex{V}.
  \\
  \vlabel{V}{\minus} & {}\leftarrow{} &
                       \bottom, \vertex{V}.
  \\
  \elabel{U}{V}{\plus}  & {}\leftarrow{} &
                          \bottom, \edge{U}{V}.
  \\
  \elabel{U}{V}{\minus} & {}\leftarrow{} &
                          \bottom, \edge{U}{V}.
  \end{array}  
\end{equation*}
As shown above,
any model~$Y\subseteq X$ of~$P^X$ must necessarily include
$\obsvlabel{i}{s}$ if $\mu(i)=\nolinebreak s$,
$\obselabel{j}{i}{s}$ if $\sigma(j,i)=s$,
$\avertex{i}$ if $i\in W$ or $(i {\,\rightarrow\,} k)\in E$ for some $k\in W$,
$\aedge{j}{i}$ if $(j {\,\rightarrow\,} i)\in E$ for some $i\in W$, and
$\activ{i}$ if $i\in W$.
Proceeding by proof by contradiction,
assume that there is a model~$Y\subset X$ of~$P^X$ such that
$\vlabel{i}{s}$, $\elabel{j}{i}{s}$, or $\contrary{j}{i}$ is not in~$Y$
for some $i\in V$ or $(j {\,\rightarrow\,} i)\in E$, respectively, and $s\in\{\plus,\minus\}$.
From the previous considerations and the first two rules repeated above,
we know that $\vlabel{i}{s}$ and $\elabel{j}{i}{s}$ must belong to~$Y$
if $\mu(i)=s$ or $\sigma(j,i)=s$, respectively.
Furthermore, the third rule necessitates
$\{\vlabel{i}{\plus},\vlabel{i}{\minus}\}\cap Y\neq\emptyset$
for every $i\in W$ or $i\in V$ such that $(i {\,\rightarrow\,} k)\in E$ for some $k\in W$,
and the fourth rule implies
$\{\elabel{j}{i}{\plus},\elabel{j}{i}{\minus}\}\cap Y\neq\emptyset$
for every $(j {\,\rightarrow\,} i)\in E$ such that $i\in W$.
In view of the last four rules, we immediately conclude that $\bottom\notin Y$,
which in turn implies that, for every~$i\in W$,
there is some $(j {\,\rightarrow\,} i)\in E$ such that
$\contrary{j}{i}$ does not belong to~$Y$.
Comparing the rules defining \atomfont{opposite},
the exclusion of $\contrary{j}{i}$
is possible only if~$Y$ does not include
$\elabel{j}{i}{r}$, $\vlabel{j}{s}$, and $\vlabel{i}{t}$
such that $t\neq sr$.
As we have shown above that some atoms
$\elabel{j}{i}{r}$, $\vlabel{j}{s}$, and $\vlabel{i}{t}$
for $r,s,t\in\{\plus,\minus\}$ must belong to~$Y$,
we can now conclude that~$t=sr$ holds and that the atoms over
predicates \atomfont{labelE} and \atomfont{labelV} in~$Y$
define (partial) labelings~$\sigma'$ and~$\mu'$ by:
\begin{itemize}
\item For every $i\in W$, pick some edge $(j {\,\rightarrow\,} i)\in E$
      such that $\contrary{j}{i}$ does not belong to~$Y$, and
      let $\sigma'(j,i)=r$ if $\elabel{j}{i}{r}\in Y$,
      $\mu'(j)=s$ if $\vlabel{j}{s}\in Y$, and
      $\mu'(i)=t$ if $\vlabel{i}{t}\in Y$.
\end{itemize}
As we have seen above, such an edge $(j {\,\rightarrow\,} i)\in E$ exists
for every $i\in W$, and the fact that $t\neq sr$ is not obtained for atoms
$\elabel{j}{i}{r}$, $\vlabel{j}{s}$, and $\vlabel{i}{t}$ in~$Y$ implies
that $\sigma'$ and $\mu'$ assign
unique labels to $(j {\,\rightarrow\,} i)$, $j$, and~$i$, respectively.
% of the edges
% $(j {\,\rightarrow\,} i)\in E$ such that $i\in W$ or $\sigma(j,i)=s$ and
% the vertices $i\in V$ such that
% $i\in W$, $(i {\,\rightarrow\,} k)\in E$ for some $k\in W$, or $\mu(i)=s$.
%%% T2M: NEXT PHRASE is TOO LONG !!
When we totalize $\sigma'$ and $\mu'$ by setting
$\sigma'(j,i)=\sigma(j,i)$ and $\mu'(i)=\mu(i)$
if $\sigma(j,i)$ or $\mu(i)$, respectively, is defined,
and $\sigma'(j,i)=\plus$ as well as $\mu'(i)=\plus$ %, otherwise,
for all remaining
edges in~$E$ and vertices in~$V$,
we obtain witnessing labelings for~$W$.
But this is a contradiction to the fact that~$W$ is a MIC,
which allows us to conclude that
there cannot be any model~$Y\subset X$ of~$P^X$ that omits
$\vlabel{i}{s}$, $\elabel{j}{i}{s}$, or $\contrary{j}{i}$ 
for some $i\in V$ or $(j {\,\rightarrow\,} i)\in E$, respectively, and $s\in\{\plus,\minus\}$.

To conclude the proof that~$X$ is a $\subseteq$-minimal model of~$P^X$,
note that the integrity constraint
\begin{equation*}
\begin{array}{@{{}\leftarrow{}}l}
  \naf{\bottom}.
\end{array}
\end{equation*}
from~(\ref{eq:inconsistency}) does not contribute any rule to~$P^X$
because $\bottom\in X$.
We have now investigated all rules in $P_D\cup\tau((V,E,\sigma),\mu)$
and shown that their ground instances in~$P^X$ are satisfied by~$X$.
Furthermore, we have checked for all atoms in~$X$ that they cannot be
excluded in any model $Y\subset X$ of~$P^X$.
That is, $X$ is indeed a $\subseteq$-minimal model of~$P^X$
and thus an answer set of $P_D\cup\tau((V,E,\sigma),\mu)$.
\hfill%
\end{proof}

%%% Local Variables: 
%%% mode: latex
%%% TeX-master: "paper"
%%% End: 

\paragraph{Acknowledgments.}

Philippe Veber was supported by a grant from DAAD. 
This work was partially funded by the GoFORSYS project\footnote{\url{http://www.goforsys.org}} (Grant~0313924).
The authors are grateful to Carito Guziolowski for profitable discussions and
for providing them with the data on yeast,
and to Roland Kaminski for helpful comments on encodings.

%%% Local Variables: 
%%% mode: latex
%%% TeX-master: "paper"
%%% End: 

%\bibliographystyle{acmtrans} % https://svn.haiti.cs.uni-potsdam.de/svn/reposWV/TeXInputs/trunk
%\bibliography{lit,akku,procs}

\begin{thebibliography}{}

\bibitem[\protect\citeauthoryear{Baral}{Baral}{2003}]{baral02a}
{\sc Baral, C.} 2003.
\newblock {\em Knowledge Representation, Reasoning and Declarative Problem
  Solving}.
\newblock Cambridge University Press.

% \bibitem[\protect\citeauthoryear{Baral, Brewka, and Schlipf}{Baral
%   et~al\mbox{.}}{2007}]{lpnmr07}
% {\sc Baral, C.}, {\sc Brewka, G.}, {\sc and} {\sc Schlipf, J.}, Eds. 2007.
% \newblock {\em Proceedings of the Ninth International Conference on Logic
%   Programming and Nonmonotonic Reasoning (LPNMR'07)}. Springer.

\bibitem[\protect\citeauthoryear{Ben-Eliyahu and Dechter}{Ben-Eliyahu and
  Dechter}{1994}]{bendec94a}
{\sc Ben-Eliyahu, R.} {\sc and} {\sc Dechter, R.} 1994.
\newblock Propositional semantics for disjunctive logic programs.
\newblock {\em Annals of Mathematics and Artificial Intelligence\/}~{\em
  12,\/}~1-2, 53--87.

\bibitem[\protect\citeauthoryear{BioASP}{BioASP Tools\ignorespaces}{}]{bioasptoolchain}
BioASP Tools.
\newblock \url{http://www.cs.uni-potsdam.de/wv/bioasp}.

\bibitem[\protect\citeauthoryear{Chen, Hancock, and Lopes}{Chen et~al\mbox{.}}{2007}]{chen2007}
{\sc Chen, M.}, {\sc Hancock, L.}, {\sc and} {\sc Lopes, J.} 2007.
\newblock {T}ranscriptional regulation of yeast phospholipid biosynthetic genes.
\newblock {\em Biochimica et Biophysica Acta\/}~{\em 1771,\/}~3, 310--321.

\bibitem[\protect\citeauthoryear{Dershowitz, Hanna, and Nadel}{Dershowitz
  et~al\mbox{.}}{2006}]{scalableMUC}
{\sc Dershowitz, N.}, {\sc Hanna, Z.}, {\sc and} {\sc Nadel, A.} 2006.
\newblock A scalable algorithm for minimal unsatisfiable core extraction.
\newblock In {\em Proceedings of the Ninth International Conference on Theory
  and Applications of Satisfiability Testing (SAT'06)}, {A.~Biere} {and}
  {C.~Gomes}, Eds.
  Springer, 36--41.

\bibitem[\protect\citeauthoryear{Drescher, Gebser, Grote, Kaufmann, K{\"o}nig,
  Ostrowski, and Schaub}{Drescher et~al\mbox{.}}{2008}]{drgegrkakoossc08a}
{\sc Drescher, C.}, {\sc Gebser, M.}, {\sc Grote, T.}, {\sc Kaufmann, B.}, {\sc
  K{\"o}nig, A.}, {\sc Ostrowski, M.}, {\sc and} {\sc Schaub, T.} 2008.
\newblock Conflict-driven disjunctive answer set solving.
\newblock In {\em Proceedings of the Eleventh International Conference on
  Principles of Knowledge Representation and Reasoning (KR'08)}, {G.~Brewka}
  {and} {J.~Lang}, Eds. AAAI Press, 422--432.

% \bibitem[\protect\citeauthoryear{Eiter and Gottlob}{Eiter and
%   Gottlob}{1992}]{eitgot92}
% {\sc Eiter, T.} {\sc and} {\sc Gottlob, G.} 1992.
% \newblock On the complexity of propositional knowledge base revision, updates,
%   and counterfactuals.
% \newblock {\em Artificial Intelligence\/}~{\em 57}, 227--270.

\bibitem[\protect\citeauthoryear{Eiter and Gottlob}{Eiter and
  Gottlob}{1995}]{eitgot95a}
{\sc Eiter, T.} {\sc and} {\sc Gottlob, G.} 1995.
\newblock On the computational cost of disjunctive logic programming:
  Propositional case.
\newblock {\em Annals of Mathematics and Artificial Intelligence\/}~{\em
  15,\/}~3-4, 289--323.

\bibitem[\protect\citeauthoryear{Erd\H{o}s and R\'enyi}{Erd\H{o}s and
  R\'enyi}{1959}]{erdos}
{\sc Erd\H{o}s, A.} {\sc and} {\sc R\'enyi, P.} 1959.
\newblock On random graphs.
\newblock {\em Publicationes Mathematicae\/}~{\em 6}, 290--297.

\bibitem[\protect\citeauthoryear{Friedman, Linial, Nachman, and
  Pe'er}{Friedman et~al\mbox{.}}{2000}]{Friedman2000}
{\sc Friedman, N.}, {\sc Linial, M.}, {\sc Nachman, I.}, {\sc and} {\sc
  Pe'er, D.} 2000.
\newblock Using {B}ayesian networks to analyze expression data.
\newblock {\em Journal of Computational Biology\/}~{\em 7,\/}~3-4, 601--620.

% \bibitem[\protect\citeauthoryear{Garey and Johnson}{Garey and
%   Johnson}{1979}]{garjoh79}
% {\sc Garey, M.} {\sc and} {\sc Johnson, D.} 1979.
% \newblock {\em Computers and Intractability: A Guide to the Theory of
%   {NP}-Completeness}.
% \newblock W. Freeman and Co. %, New York.

\bibitem[\protect\citeauthoryear{Gebser, Guziolowski, Ivanchev, Schaub, Siegel,
  Thiele, and Veber}{Gebser et~al\mbox{.}}{2010}]{geguivscsithve09a}
{\sc Gebser, M.}, {\sc Guziolowski, C.}, {\sc Ivanchev, M.}, {\sc Schaub, T.},
  {\sc Siegel, A.}, {\sc Thiele, S.}, {\sc and} {\sc Veber, P.} 2010.
\newblock Repair and prediction (under inconsistency) in large biological
  networks with answer set programming.
\newblock In {\em Proceedings of the Twelfth International Conference on
  Principles of Knowledge Representation and Reasoning (KR'10)},
  to appear.

% \bibitem[\protect\citeauthoryear{Gebser, Kaminski, Kaufmann, Ostrowski, Schaub,
%   and Thiele}{Gringo Manual\ignorespaces}{}]{potasscoManual}
% Gringo Manual.
% {\sc Gebser, M.}, {\sc Kaminski, R.}, {\sc Kaufmann, B.}, {\sc Ostrowski, M.},
%   {\sc Schaub, T.}, {\sc and} {\sc Thiele, S.}
% \newblock A user's guide to \sysfont{gringo}, \sysfont{clasp}, \sysfont{clingo},
%   and \sysfont{iclingo}.
% \newblock \url{http://potassco.sourceforge.net}.

\bibitem[\protect\citeauthoryear{Gebser, Kaminski, Ostrowski, Schaub,
  and Thiele}{Gebser et~al\mbox{.}}{2009a}]{potasscoManual}
{\sc Gebser, M.}, {\sc Kaminski, R.}, {\sc Ostrowski, M.},
  {\sc Schaub, T.}, {\sc and} {\sc Thiele, S.} 2009a.
\newblock On the input language of {ASP} grounder \textit{gringo}.
\newblock In {\em Proceedings of the Tenth International Conference on Logic
  Programming and Nonmonotonic Reasoning (LPNMR'09)}, {E.~{Erdem}},
  {F.~Lin}, {and} {T.~Schaub}, Eds. 
  Springer, 502--508.

\bibitem[\protect\citeauthoryear{Gebser, Kaufmann, Neumann, and Schaub}{Gebser
  et~al\mbox{.}}{2007}]{gekanesc07c}
{\sc Gebser, M.}, {\sc Kaufmann, B.}, {\sc Neumann, A.}, {\sc and} {\sc Schaub,
  T.} 2007.
\newblock Conflict-driven answer set enumeration.
\newblock In {\em Proceedings of the Ninth International Conference on Logic
  Programming and Nonmonotonic Reasoning (LPNMR'07)}, {C.~{Baral}},
  {G.~Brewka}, {and} {J.~Schlipf}, Eds. 
  Springer, 136--148.

\bibitem[\protect\citeauthoryear{Gebser, Kaufmann, and Schaub}{Gebser
  et~al\mbox{.}}{2009b}]{gekasc09a}
{\sc Gebser, M.}, {\sc Kaufmann, B.}, {\sc and} {\sc Schaub, T.} 2009b.
\newblock Solution enumeration for projected {B}oolean search problems.
\newblock In {\em Proceedings of the Sixth International Conference on
  Integration of AI and OR Techniques in Constraint Programming for
  Combinatorial Optimization Problems (CPAIOR'09)}, {W.~{van Hoeve}} {and}
  {J.~Hooker}, Eds. 
  Springer, 71--86.

\bibitem[\protect\citeauthoryear{Gebser, Kaufmann, and Schaub}{Gebser
  et~al\mbox{.}}{2009c}]{gekanesc07b}
{\sc Gebser, M.}, {\sc Kaufmann, B.}, {\sc and} {\sc Schaub,
  T.} 2009c.
\newblock The conflict-driven answer set solver \textit{clasp}: Progress
          report.
\newblock In {\em Proceedings of the Tenth International Conference on Logic
  Programming and Nonmonotonic Reasoning (LPNMR'09)}, {E.~{Erdem}},
  {F.~Lin}, {and} {T.~Schaub}, Eds. 
  Springer, 509--514.
% \newblock See \citeN{lpnmr07}, 260--265.

% \bibitem[\protect\citeauthoryear{Gebser, Schaub, and Thiele}{Gebser
%   et~al\mbox{.}}{2007}]{gescth07a}
% {\sc Gebser, M.}, {\sc Schaub, T.}, {\sc and} {\sc Thiele, S.} 2007.
% \newblock {GrinGo}: A new grounder for answer set programming.
% \newblock See \citeN{lpnmr07}, 266--271.

% \bibitem[\protect\citeauthoryear{Gebser, Schaub, Thiele, Usadel, and
%   Veber}{Gebser et~al\mbox{.}}{2008}]{gescthusve08b}
% {\sc Gebser, M.}, {\sc Schaub, T.}, {\sc Thiele, S.}, {\sc Usadel, B.}, {\sc
%   and} {\sc Veber, P.} 2008.
% \newblock Detecting inconsistencies in large biological networks with answer
%   set programming.
% \newblock In {\em Proceedings of the Twenty-fourth International Conference on
%   Logic Programming (ICLP'08)}, {M.~{Garcia de la Banda}} {and} {E.~Pontelli},
%   Eds.  Springer, 130--144.

\bibitem[\protect\citeauthoryear{Gelfond}{Gelfond}{2008}]{gelfond08a}
{\sc Gelfond, M.} 2008.
\newblock Answer sets.
\newblock In {\em Handbook of Knowledge Representation}, {V.~Lifschitz},
  {F.~{van Hermelen}}, {and} {B.~Porter}, Eds. Elsevier, 285--316. % Chapter~7.

\bibitem[\protect\citeauthoryear{Gelfond, Lifschitz, Przymusinska, and
  Truszczy{\'n}ski}{Gelfond et~al\mbox{.}}{1991}]{geliprtr91}
{\sc Gelfond, M.}, {\sc Lifschitz, V.}, {\sc Przymusinska, H.}, {\sc and} {\sc
  Truszczy{\'n}ski, M.} 1991.
\newblock Disjunctive defaults.
\newblock In {\em Proceedings of the Second International Conference on
  Principles of Knowledge Representation and Reasoning (KR'91)}, {J.~Allen},
  {R.~Fikes}, {and} {E.~Sandewall}, Eds. Morgan Kaufmann Publishers, 230--237.

\bibitem[\protect\citeauthoryear{Giunchiglia, Lierler, and Maratea}{Giunchiglia
  et~al\mbox{.}}{2006}]{gilima06a}
{\sc Giunchiglia, E.}, {\sc Lierler, Y.}, {\sc and} {\sc Maratea, M.} 2006.
\newblock Answer set programming based on propositional satisfiability.
\newblock {\em Journal of Automated Reasoning\/}~{\em 36,\/}~4, 345--377.

% \bibitem[\protect\citeauthoryear{Gr{\'e}goire, Mazure, and Piette}{Gr{\'e}goire
%   et~al\mbox{.}}{2008}]{grmapi08a}
% {\sc Gr{\'e}goire, {\'E}.}, {\sc Mazure, B.}, {\sc and} {\sc Piette, C.} 2008.
% \newblock On approaches to explaining infeasibility of sets of {B}oolean
%   clauses.
% \newblock In {\em Proceedings of the Twentieth IEEE International Conference on
%   Tools with Artificial Intelligence (ICTAI'08)}. IEEE Computer Society,
%   74--83.

\bibitem[\protect\citeauthoryear{Guelzim, Bottani, Bourgine, and
  K\'ep\`es}{Guelzim et~al\mbox{.}}{2002}]{guelzim}
{\sc Guelzim, N.}, {\sc Bottani, S.}, {\sc Bourgine, P.}, {\sc and} {\sc
  K\'ep\`es, F.} 2002.
\newblock Topological and causal structure of the yeast transcriptional
  regulatory network.
\newblock {\em Nature Genetics\/}~{\em 31}, 60--63.

\bibitem[\protect\citeauthoryear{Gutierrez-Rios, Rosenblueth, Loza, Huerta,
  Glasner, Blattner, and Collado-Vides}{Gutierrez-Rios
  et~al\mbox{.}}{2003}]{pmid14597655}
{\sc Gutierrez-Rios, R.}, {\sc Rosenblueth, D.}, {\sc Loza, J.}, {\sc Huerta,
  A.}, {\sc Glasner, J.}, {\sc Blattner, F.}, {\sc and} {\sc Collado-Vides, J.}
  2003.
\newblock Regulatory network of {E}scherichia coli: Consistency between
  literature knowledge and microarray profiles.
\newblock {\em Genome Research\/}~{\em 13,\/}~11, 2435--2443.

\bibitem[\protect\citeauthoryear{Guziolowski, Bourde, Moreews, and
  Siegel}{Guziolowski et~al\mbox{.}}{2009}]{gubomosi09a}
{\sc Guziolowski, C.}, {\sc Bourde, A.}, {\sc Moreews, F.}, {\sc and} {\sc
  Siegel, A.} 2009.
\newblock Bioquali cytoscape plugin: analysing the global consistency of
  regulatory networks.
\newblock {\em BMC Genomics\/}~{\em 10}.

\bibitem[\protect\citeauthoryear{Guziolowski, Veber, Le Borgne, Radulescu, and
  Siegel}{Guziolowski et~al\mbox{.}}{2007}]{coliRIAMS}
{\sc Guziolowski, C.}, {\sc Veber, P.}, {\sc Le Borgne, M.}, {\sc Radulescu,
  O.}, {\sc and} {\sc Siegel, A.} 2007.
\newblock Checking consistency between expression data and large scale
  regulatory networks: A case study.
\newblock {\em Journal of Biological Physics and Chemistry\/}~{\em 7,\/}~2,
  37--43.

% \bibitem[\protect\citeauthoryear{Jackson and Lopes}{Jackson and
%   Lopes}{1996}]{pmid8614637}
% {\sc Jackson, J.} {\sc and} {\sc Lopes, J. M.} 1996.
% \newblock The yeast {U}{M}{E}6 gene is required for both negative and positive
%   transcriptional regulation of phospholipid biosynthetic gene expression.
% \newblock {\em Nucleic Acids Research\/}~{\em 24,\/}~7, 1322--1329.

\bibitem[\protect\citeauthoryear{Janhunen, Niemel{\"a}, Seipel, Simons, and
  You}{Janhunen et~al\mbox{.}}{2006}]{janisesiyo06a}
{\sc Janhunen, T.}, {\sc Niemel{\"a}, I.}, {\sc Seipel, D.}, {\sc Simons, P.},
  {\sc and} {\sc You, J.} 2006.
\newblock Unfolding partiality and disjunctions in stable model semantics.
\newblock {\em ACM Transactions on Computational Logic\/}~{\em 7,\/}~1, 1--37.

\bibitem[\protect\citeauthoryear{Jeong, Tombor, Albert, Oltvai, and
  Barab\'asi}{Jeong et~al\mbox{.}}{2000}]{barabasi}
{\sc Jeong, H.}, {\sc Tombor, B.}, {\sc Albert, R.}, {\sc Oltvai, Z.}, {\sc
  and} {\sc Barab\'asi, A.} 2000.
\newblock The large-scale organization of metabolic networks.
\newblock {\em Nature\/}~{\em 407}, 651--654.

% \bibitem[\protect\citeauthoryear{Joyce and Palsson}{Joyce and
%   Palsson}{2006}]{Joyce2006}
% {\sc Joyce, A.} {\sc and} {\sc Palsson, B.} 2006.
% \newblock The model organism as a system: Integrating `omics' data sets.
% \newblock {\em Nature Reviews Molecular Cell Biology\/}~{\em 7,\/}~3, 198--210.

% \bibitem[\protect\citeauthoryear{Kitano}{Kitano}{2002}]{kitano02a}
% {\sc Kitano, H.} 2002.
% \newblock Systems biology: a brief overview.
% \newblock {\em Science\/}~{\em 295,\/}~5560, 1662--1664.

\bibitem[\protect\citeauthoryear{Klamt and Stelling}{Klamt and
  Stelling}{2006}]{revue-fba}
{\sc Klamt, S.} {\sc and} {\sc Stelling, J.} 2006.
\newblock Stoichiometric and constraint-based modelling.
\newblock In {\em System Modeling in Cellular Biology: From Concepts to Nuts
  and Bolts}. MIT Press, 73--96.

\bibitem[\protect\citeauthoryear{Kuipers}{Kuipers}{1994}]{kuipers94a}
{\sc Kuipers, B.} 1994.
\newblock {\em Qualitative reasoning. Modeling and simulation with incomplete
  knowledge}.
\newblock MIT Press.

\bibitem[\protect\citeauthoryear{Leone, Pfeifer, Faber, Eiter, Gottlob, Perri,
  and Scarcello}{Leone et~al\mbox{.}}{2006}]{dlv03a}
{\sc Leone, N.}, {\sc Pfeifer, G.}, {\sc Faber, W.}, {\sc Eiter, T.}, {\sc
  Gottlob, G.}, {\sc Perri, S.}, {\sc and} {\sc Scarcello, F.} 2006.
\newblock The {DLV} system for knowledge representation and reasoning.
\newblock {\em ACM Transactions on Computational Logic\/}~{\em 7,\/}~3,
  499--562.

\bibitem[\protect\citeauthoryear{Mallory, Cooper, and Strich}{Mallory et~al\mbox{.}}{2007}]{mallory2007}
{\sc Mallory, M.}, {\sc Cooper, K.}, {\sc and} {\sc Strich, R.} 2007.
\newblock {M}eiosis-specific destruction of the {U}me6p repressor by the {C}dc20-directed {A}{P}{C}/{C}.
\newblock {\em Molecular Cell}~{\em 27,\/}~6, 951--961.

% \bibitem[\protect\citeauthoryear{Mahajan, Fu, and Malik}{Mahajan
%   et~al\mbox{.}}{2005}]{mafuma04a}
% {\sc Mahajan, Y.}, {\sc Fu, Z.}, {\sc and} {\sc Malik, S.} 2005.
% \newblock Zchaff2004: An efficient {SAT} solver.
% \newblock In {\em Proceedings of the Seventh International Conference on Theory
%   and Applications of Satisfiability Testing (SAT'04)}, {H.~Hoos} {and}
%   {D.~Mitchell}, Eds. Lecture Notes in Computer Science, vol. 3542.
%   Springer, 360--375.

% \bibitem[\protect\citeauthoryear{Minker}{Minker}{1988}]{minker88}
% {\sc Minker, J.}, Ed. 1988.
% \newblock {\em Foundations of Deductive Databases and Logic Programming}.
% \newblock Morgan Kaufmann Publishers. %, Los Altos.

% \bibitem[\protect\citeauthoryear{Mitchell}{Mitchell}{2005}]{mitchell05a}
% {\sc Mitchell, D.} 2005.
% \newblock A {SAT} solver primer.
% \newblock {\em Bulletin of the European Association for Theoretical Computer
%   Science\/}~{\em 85}, 112--133.

\bibitem[\protect\citeauthoryear{Papadimitriou and Yannakakis}{Papadimitriou
  and Yannakakis}{1982}]{papyan82a}
{\sc Papadimitriou, C.} {\sc and} {\sc Yannakakis, M.} 1982.
\newblock The complexity of facets (and some facets of complexity).
\newblock In {\em Proceedings of the Fourteenth Annual ACM Symposium on Theory
  of Computing (STOC'82)}. ACM Press, 255--260.

\bibitem[\protect\citeauthoryear{Remy, Ruet, and Thieffry}{Remy
  et~al\mbox{.}}{2008}]{ruet}
{\sc Remy, {\'E}.}, {\sc Ruet, P.}, {\sc and} {\sc Thieffry, D.} 2008.
\newblock Graphic requirements for multistability and attractive cycles in a
  {B}oolean dynamical framework.
\newblock {\em Advances in Applied Mathematics\/}~{\em 41,\/}~3, 335--350.

\bibitem[\protect\citeauthoryear{Richard, Comet, and Bernot}{Richard
  et~al\mbox{.}}{2004}]{richard07}
{\sc Richard, A.}, {\sc Comet, J.}, {\sc and} {\sc Bernot, G.} 2004.
\newblock R. {T}homas' modeling of biological regulatory networks: Introduction
  of singular states in the qualitative dynamics.
\newblock {\em Fundamenta Informaticae\/}~{\em 65,\/}~4, 373--392.

% \bibitem[\protect\citeauthoryear{Shah, Jensen, Frenz, Johnson, and
%   Johnston}{Shah et~al\mbox{.}}{2001}]{pmid11729145}
% {\sc Shah, R.}, {\sc Jensen, S.}, {\sc Frenz, L.}, {\sc Johnson, A.}, {\sc and}
%   {\sc Johnston, L.} 2001.
% \newblock The {S}po12 protein of {S}accharomyces cerevisiae: a regulator of
%   mitotic exit whose cell cycle-dependent degradation is mediated by the
%   anaphase-promoting complex.
% \newblock {\em Genetics\/}~{\em 159,\/}~3, 965--980.

\bibitem[\protect\citeauthoryear{Siegel, Radulescu, Le Borgne, Veber, Ouy, and
  Lagarrigue}{Siegel et~al\mbox{.}}{2006}]{pmid16556482}
{\sc Siegel, A.}, {\sc Radulescu, O.}, {\sc Le Borgne, M.}, {\sc Veber, P.},
  {\sc Ouy, J.}, {\sc and} {\sc Lagarrigue, S.} 2006.
\newblock Qualitative analysis of the relation between {DNA} microarray data
  and behavioral models of regulation networks.
\newblock {\em Biosystems\/}~{\em 84,\/}~2, 153--174.

% \bibitem[\protect\citeauthoryear{Simons, Niemel{\"a}, and Soininen}{Simons
%   et~al\mbox{.}}{2002}]{siniso02a}
% {\sc Simons, P.}, {\sc Niemel{\"a}, I.}, {\sc and} {\sc Soininen, T.} 2002.
% \newblock Extending and implementing the stable model semantics.
% \newblock {\em Artificial Intelligence\/}~{\em 138,\/}~1-2, 181--234.

% \bibitem[\protect\citeauthoryear{Sontag}{Sontag}{2005}]{sontag05a}
% {\sc Sontag, E.} 2005.
% \newblock Molecular systems biology and control.
% \newblock {\em European Journal of Control\/}~{\em 11,\/}~4-5, 396--435.

\bibitem[\protect\citeauthoryear{Soul\'e}{Soul\'e}{2003}]{SouleComplexus}
{\sc Soul\'e, C.} 2003.
\newblock Graphic requirements for multistationarity.
\newblock {\em Complexus\/}~{\em 1,\/}~3, 123--133.

\bibitem[\protect\citeauthoryear{Soul\'e}{Soul\'e}{2006}]{SouleRevue}
{\sc Soul\'e, C.} 2006.
\newblock Mathematical approaches to differentiation and gene regulation.
\newblock {\em Comptes Rendus Biologies\/}~{\em 329}, 13--20.

\bibitem[\protect\citeauthoryear{Sudarsanam, Iyer, Brown, and
  Winston}{Sudarsanam et~al\mbox{.}}{2000}]{snf2-ko}
{\sc Sudarsanam, P.}, {\sc Iyer, V.}, {\sc Brown, P.}, {\sc and} {\sc Winston,
  F.} 2000.
\newblock Whole-genome expression analysis of snf/swi mutants of
  {S}accharomyces cerevisiae.
\newblock {\em Proceedings of the National Academy of Sciences of the United
  States of America\/}~{\em 97,\/}~7, 3364--3369.

\bibitem[\protect\citeauthoryear{Syrj{\"a}nen}{Lparse Manual\ignorespaces}{}]{lparseManual}
Lparse Manual.
{\sc Syrj{\"a}nen, T.}
\newblock \sysfont{Lparse} 1.0 user's manual.
\newblock \url{http://www.tcs.hut.fi/Software/smodels/lparse.ps.gz}.

\bibitem[\protect\citeauthoryear{Veber, Le Borgne, Siegel, Lagarrigue, and
  Radulescu}{Veber et~al\mbox{.}}{2004}]{ECCS05}
{\sc Veber, P.}, {\sc Le Borgne, M.}, {\sc Siegel, A.}, {\sc Lagarrigue, S.},
  {\sc and} {\sc Radulescu, O.} 2004.
\newblock Complex qualitative models in biology: A new approach.
\newblock {\em Complexus\/}~{\em 2,\/}~3-4, 140--151.

\bibitem[\protect\citeauthoryear{Washburn and Esposito}{Washburn and
  Esposito}{2006}]{pmid11238941}
{\sc Washburn, B.} {\sc and} {\sc Esposito, R.} 2006.
\newblock Identification of the {S}in3-binding site in {U}me6 defines a
  two-step process for conversion of {U}me6 from a transcriptional repressor to
  an activator in yeast.
\newblock {\em Molecular and Cellular Biology\/}~{\em 21,\/}~6, 2057--2069.

\end{thebibliography}

%%% Local Variables: 
%%% mode: latex
%%% TeX-master: "paper"
%%% End: 

\end{document}